\documentclass[11pt, a4paper]{article}
\usepackage{amsmath}
\usepackage{amsfonts}
\usepackage[parfill]{parskip}
\usepackage{fullpage}
\usepackage{cite}
\usepackage[normalem]{ulem}
\usepackage{enumitem}
\usepackage{amsthm}
\usepackage{breqn}
\usepackage{relsize}
\usepackage{graphicx}
\usepackage{tikz}
\usepackage{amssymb}
\usetikzlibrary{decorations.markings}

\tikzset{->-/.style={decoration={
  markings,
  mark=at position #1 with {\arrow{>}}},postaction={decorate}}}

\tikzset{
  big arrow/.style={
    decoration={markings,mark=at position 1 with {\arrow[scale=1.6,#1]{>}}},
    postaction={decorate},
    shorten >=0.4pt}}

\newcommand{\beq}{\begin{equation}}
\newcommand{\eeq}{\end{equation}}
\newcommand{\ba}{\begin{array}}
\newcommand{\ea}{\end{array}}

\newcommand{\sn}{\mathrm{sn}}

\newtheorem{thm}{Theorem}[section]
\newtheorem{lemma}[thm]{Lemma}

\newcommand\UW{Department of Applied Mathematics,\\ University of Washington,\\ Seattle, WA 98195-3925, USA}
\newcommand\UCL{Department of Physics and Astronomy,\\ University College London,\\ London, WC1E 6BT, UK}

\definecolor{UWpurple}{rgb}{.35,.01,.55}
\definecolor{BoxCol}{rgb}{.35,.01,.55}

\title{The Stability Spectrum for Elliptic Solutions to the Sine-Gordon Equation}
\date{\today}
\author{Bernard Deconinck$^1$, Peter McGill$^2$, and Benjamin L. Segal$^1$\\
\\
$^1$\UW \and $^2$\UCL}

\begin{document}
\maketitle

\begin{abstract}
We present an analysis of the stability spectrum for all stationary periodic solutions to the sine-Gordon equation. An analytical expression for the spectrum is given. From this expression, various quantitative and qualitative results about the spectrum are derived. Specifically, the solution parameter space is shown to be split into regions of distinct qualitative behavior of the spectrum, in one of which the solutions are stable. Additional results on the stability of solutions with respect to perturbations of an integer multiple of the solution period are given.
\end{abstract}

\section{Introduction}
The sine-Gordon equation in laboratory coordinates is given by
\beq u_{tt}-u_{xx}+\sin u =0. \label{SGlab}\eeq
Here, $u(x,t)$ is a real-valued function. This equation was first introduced to study surfaces of constant Gaussian curvature in light cone form \cite{Bour}.
Since its introduction it has appeared in various applications including the description of the magnetic flux in long superconducting Josephson junctions \cite{Scott1970,Rem94,Scott69}, the modeling of fermions in the Thirring model \cite{Coleman75}, the study of the stability of structures found in galaxies \cite{liang1979nonlinear,voglis2003solitons,voglis2006invariant}, mechanical vibrations of a ribbon pendulum \cite{Waldram70}, propagation of crystal dislocation \cite{FK39}, propagation of deformations along DNA double helix \cite{Y2006}, among others. A comprehensive discussion of many of these applications is found in the review paper by Barone \cite{Barone1971}.

We consider general traveling wave solutions to (\ref{SGlab}). Defining $z=x-c t,$ $\tau = t$, and introducing $v(z,\tau) = u(x,t)$,
\beq (c^2-1) v_{zz}-2c v_{z\tau}+v_{\tau \tau}+\sin (v) = 0. \label{travelingSG}\eeq
For subsequent discussion we assume that $c\ne 1$. We proceed to look for stationary solutions to (\ref{travelingSG}) of the form
\beq v(z,\tau) = f(z), \label{stationaryanz} \eeq
leading to
\beq (c^2-1) f''(z)+\sin\left(f(z)\right) = 0, \label{pendulumeqn} \eeq
where $'$ denotes a derivative with respect to $z$. Integrating once,
\beq \frac{1}{2} (c^2-1)f'(z)^2 + 1 - \cos\left(f(z)\right) = E, \label{energyeqn} \eeq
where $E$ is a constant of integration referred to as the total energy.
The stationary solutions in this paper are the elliptic solutions to (\ref{energyeqn}) and their limits. These solutions are periodic in $z$ and limit to the well-known kink solutions as their period goes to infinity \cite{dauxois2006physics,NewallSolitons}.

\begin{figure}
\begin{tabular}{cc}
  \includegraphics[width=75mm]{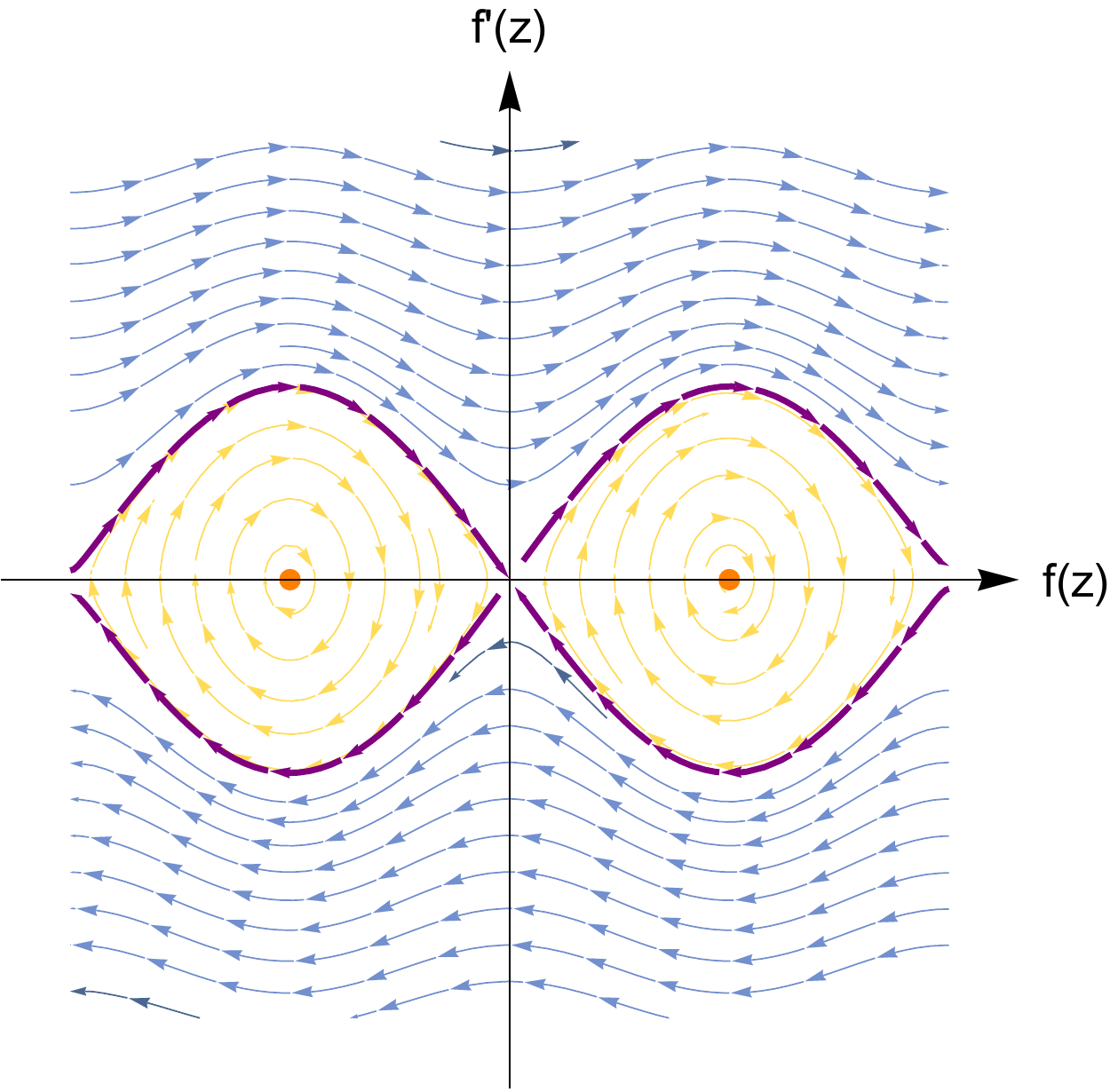} & \includegraphics[width=75mm]{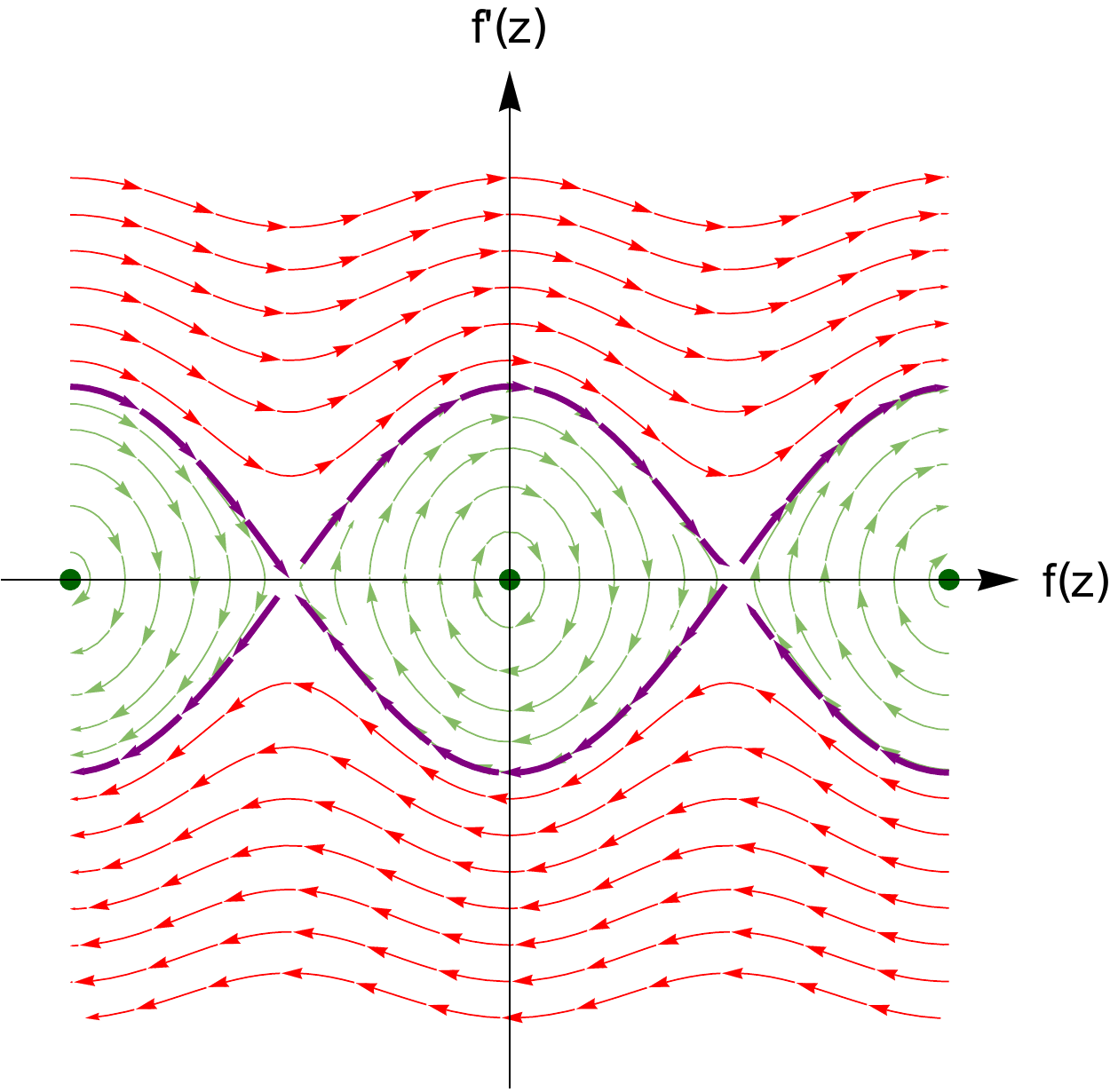} \\
(a) Subluminal: $c^2<1$ & (b)  Superluminal: $c^2>1$ 
\end{tabular}
\caption{Phase portraits of the solutions showing both librational waves (closed orbits inside the separatrix)  in yellow for (a) and green for (b) and rotational waves (orbits outside the separatrix) in blue for (a) and red for (b). The separatrix is denoted in purple.}
\label{SupSub}
\end{figure}

We call stationary solutions $f(z)$ with waves speeds satisfying $c^2<1$ (respectively $c^2>1$) subluminal (superluminal). Representative phase portraits of subluminal and superluminal solutions to (\ref{energyeqn}) are shown in Figure~\ref{SupSub}. Additionally, we call solutions $f(z)$ whose orbits in phase space lie within the separatrix librational, and those whose orbits lie outside the separatrix rotational. This distinction is illustrated in Figure~\ref{SupSub} in both the subluminal and superluminal cases. Librational waves correspond to $E\in(0,2).$ For rotational waves, $E<0$ for subluminal waves and $E>2$ for superluminal waves.

\begin{figure}
\centering
\includegraphics[width=10cm]{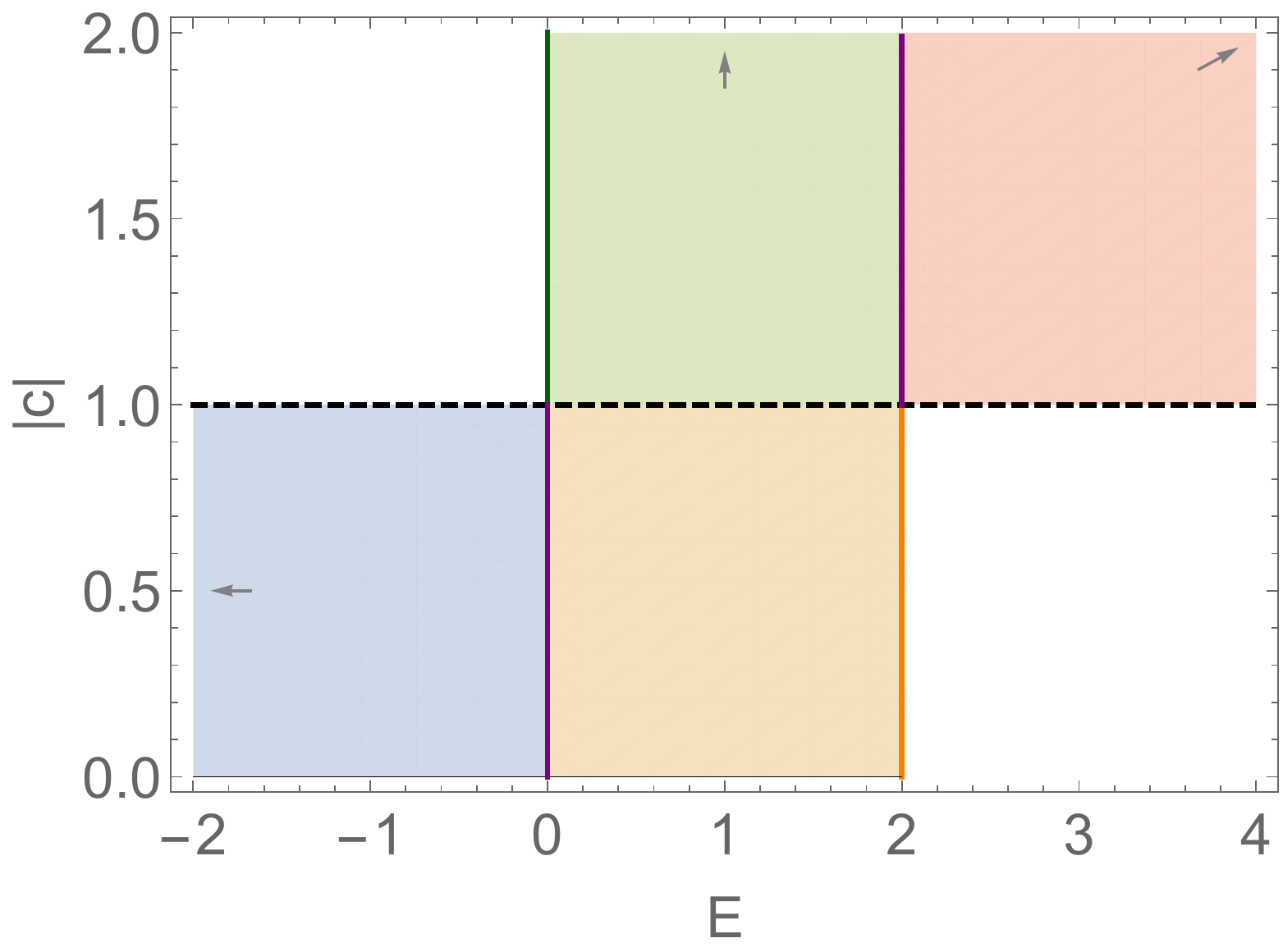} 
\caption{Subregions of Parameter space. Colors correspond to solutions in Figure~\ref{SupSub}. Blue: subluminal rotational ($0\le |c| \, <1,\,E<0$), orange: subluminal librational ($0\le |c| \, <1,\,0<E \le 2$), green: superluminal librational ($|c| \, >1,\,0\le E< 2$), red: superluminal rotational ($0\le|c| \, >1,\,E>2$). Subregions extend to infinity in directions of arrows. Subluminal kink solutions occur for $E = 0,0\le |c| \, <1,$ and superluminal kink solutions occur for $E = 2,|c| \, >1$.}
\label{SGregion}
\end{figure}

Scott \cite{scott1969waveform} was the first to study the stability of periodic traveling wave solutions to (\ref{SGlab}). He classified subluminal rotational waves as spectrally stable and determined spectral instability for all other types of waves, but these instability results were based on an incorrect claim that the spectrum in all cases was strictly confined to the real and imaginary axes.  His proof has been corrected \cite{JMMP} and extended to the Klein-Gordon equation \cite{JMMP2}. Using entirely different methods, we confirm the results in \cite{JMMP} and explicitly characterize all of parameter space. We also provide stability results for solutions perturbed by integer multiples of their fundamental period.

In Section~\ref{ellipticsolnsection} we present the elliptic solutions to (\ref{energyeqn}) in Jacobi elliptic form from \cite{JMMP}, and then reformulate the solutions into Weierstrass elliptic form. 
In Sections \ref{setup}, \ref{lax-section} and \ref{sqeig-section}, using the same methods as \cite{BD,BDN,DN,DS17}, we exploit the integrability of (\ref{SGlab}) to associate the spectrum of the linear stability problem with the Lax spectrum using the squared eigenfunction connection \cite{AKNS}. 
This allows us to obtain an analytical expression for the spectrum of the operator associated with the linearization of (\ref{SGlab}) in the form of a condition on the real part of an integral over one period of some integrand. Similar to \cite{DS17}, we proceed by integrating the integrand explicitly in Section \ref{Weierstrass}. 
Next, using the expressions obtained, we prove results concerning the location of the stability spectrum on the imaginary axis in Section \ref{imagaxisregions}. In Section \ref{regions}, we present analytical results about the spectrum, and we make use of the integral condition to split parameter space into different regions where the spectrum shows qualitatively different behavior. Finally, in Section~\ref{subharmonic} we examine the spectral stability of solutions with respect to perturbations of an integer multiple of their fundamental period and prove various stability results.

\section{Elliptic solutions}\label{ellipticsolnsection}

The derivation of the solutions is presented in the appendix of \cite{JMMP}. We limit our presentation to what is necessary for the following sections. For solutions to be real and nonsingular for real $z$ we require the following constraints:
\begin{align}
\text{subluminal,}&\text{ rotational:}  &   0\le |c|&<1,  &   E&<0,\\
\text{superluminal,}&\text{ rotational:}   &     |c|&>1,   &  E&>2,\\
\text{subluminal,}&\text{ librational:}    &    0\le |c|&<1, &    0<E&\le 2,\\
\text{superluminal,}&\text{ librational:}    &   |c|&>1,  &  0<E&\le 2.
\end{align} 
Solutions to (\ref{energyeqn}) are of the form
\beq \cos\left(f(z)\right) = \alpha+\beta \sn^2(\lambda z,k), \label{genJac}\eeq
with the following parameter values for the various cases:
\begin{align}
\text{subluminal,}&\text{ rotational:}  & \alpha&=-1,  &  \beta&= 2,  &  \lambda&=\sqrt{\frac{2-E}{2(1-c^2)}},  &  k&=\sqrt{\frac{2}{2-E}}, \label{subrot} \\
\text{superluminal,}&\text{ rotational:}   &   \alpha&=1,  &   \beta&=-2,  &  \lambda&=\sqrt{\frac{E}{2(c^2-1)}},   &   k&=\sqrt{\frac{2}{E}}, \\
\text{subluminal,}&\text{ librational:}   &   \alpha&=-1,  &   \beta&=2-E,  &  \lambda&=\sqrt{\frac{1}{1-c^2}},   &   k&=\sqrt{\frac{2-E}{2}}, \\
\text{superluminal,}&\text{ librational:}   &   \alpha&=1,  &   \beta&=-E,  &  \lambda&=\sqrt{\frac{1}{c^2-1}},   &   k&=\sqrt{\frac{E}{2}}. \label{suplum}
\end{align}
Here $\sn(x,k)$ is the Jacobi elliptic $\sn$ function with elliptic modulus $k$ \cite{BF,L13,WW,O}. We are neglecting to include a horizontal shift in $z$. This additional parameter does not change the qualitative results and it is not included here. 

Of some importance are the limits of these solutions on the boundaries of their regions of validity. On the boundaries for subluminal waves and superluminal waves the rotational and librational solutions limit to kink solutions. For subluminal waves that limit occurs when $E=0$:
\beq \cos\left(f(z)\right) = -1 +2 \tanh^2\left(\frac{z}{\sqrt{1-c^2}}\, \right), \eeq
while for superluminal waves the limit is when $E=2$:
\beq \cos\left(f(z)\right) = 1 -2 \tanh^2\left(\frac{z}{\sqrt{c^2-1}}\,\right). \eeq
These solutions are seen as the separatices in Figure~\ref{SupSub} in purple and are on the purple curves in parameter space in Figure~\ref{SGregion}.
The other limits for librational waves are when solutions limit to a constant. In the subluminal cases this occurs when $E=2$ and $\cos\left(f(z)\right) = -1,$ or in the superluminal case when $E=0$ and $\cos\left(f(z)\right) = 1$. For a general solution which is not on the boundary in parameter space, the solutions in (\ref{subrot}-\ref{suplum}) are periodic in $z$ with period $2 \mathcal{K}(k)$ where 
\beq \mathcal{K}(k) = \int_0^{\pi/2} \frac{1}{\sqrt{1-k^2 \sin^2 y}} dy, \eeq
the complete elliptic integral of the first kind.

We reformulate our elliptic solutions to (\ref{SGlab}) using Weierstrass elliptic functions \cite{O} rather than Jacobi elliptic functions. This will simplify working with the integral condition (\ref{intcond}) in Section~\ref{lax-section}, as formulas for integrating Weierstrass elliptic functions are well documented \cite{BF,GR}. It is important to note that nothing is lost by switching to Weierstrass elliptic functions, as we can map any Weierstrass elliptic function to a Jacobi elliptic function, and visa versa \cite{O,DS17}.  Let
\beq \wp(z+\omega_3,g_2,g_3)-e_3 =\left(\frac{\mathcal{K}(k)k}{\omega_1}\right)^2 \sn^2\left(\frac{\mathcal{K}(k)z}{\omega_1},k\right), \label{Jac2Wei} \eeq
with $g_2$ and $g_3$ the lattice invariants of the Weierstrass $\wp$ function, $e_1$, $e_2$, and $e_3$ the zeros of the polynomial $4t^3-g_2 t-g_3 ,$ and $\omega_1$ and $\omega_3$ the half-periods of the Weierstrass lattice given by
\beq \omega_1 = \int _{e_1} ^\infty \frac{\textrm{d}z}{\sqrt{4z^3-g_2 z-g_3}}, \eeq
\beq \omega_3 = i \int _{-e_3} ^\infty \frac{\textrm{d}z}{\sqrt{4z^3-g_2 z+g_3}}.\eeq
Using (\ref{Jac2Wei}) we convert our general solution in terms of Jacobi elliptic functions (\ref{genJac}) to one in terms of Weierstrass elliptic functions:
\beq \cos\left(f(z)\right) = \alpha +\frac{\beta}{k^2 \lambda^2}\left(\wp(z+\omega_3,g_2,g_3)-e_3\right), \label{feqnW} \eeq
with 
\beq g_2 = \frac{4}{3}\left(1-k^2+k^4\right)\lambda^4, \eeq
\beq g_3 = \frac{4}{27}\left(2-3k^2-3k^4+2k^6\right)\lambda^6, \eeq
\beq e_1 = \frac{1}{3}\left(2-k^2\right)\lambda^2,\;\; e_2 = \frac{1}{3}\left(-1+2k^2\right)\lambda^2,\;\; e_3 = \frac{1}{3}(-1-k^2)\lambda^2, \eeq
\beq \label{omega-nums} \omega_1 = \frac{\mathcal{K}(k)}{\lambda},\;\; \omega_3 = \frac{i \mathcal{K}'(k)}{\lambda}, \eeq
where $\mathcal{K}'(k)$ is the complement to $\mathcal{K}(k)$ given by $\mathcal{K}'(k) = \mathcal{K}(1-k^2)$. 
For all cases,
\beq g_2 = \frac{4-2E+E^2}{3(c^2-1)^2}, \label{g2eqn}\eeq
\beq g_3 = \frac{8-6E-3E^2+E^3}{27(c^2-1)^3}. \label{g3eqn}\eeq
One motivation for using Weierstrass elliptic functions instead of Jacobi elliptic functions is that there is a unique expression for the lattice invariants $g_2$ and $g_3$ see (\ref{g2eqn}-\ref{g3eqn}) which holds for all cases, as opposed to Jacobi elliptic functions where a different elliptic modulus $k$ is used for each case see (\ref{subrot}-\ref{suplum}).
The zeros of the polynomial $4s^3-g_2 s-g_3$ are 
\beq s= \frac{E-1}{3(c^2-1)},\;\; s=\frac{E+2}{6(1-c^2)},\;\; s=\frac{4-E}{6(c^2-1)}. \eeq
These roots correspond to $e_1$, $e_2$, and $e_3$ where $e_1>e_2>e_3$. For the various cases:
\begin{align}
\text{subluminal,}&\text{ rotational:}  & e_1 &= \frac{E-1}{3(c^2-1)} ,  &  e_2 &= \frac{E+2}{6(1-c^2)} ,  &  e_3 &=\frac{4-E}{6(c^2-1)},\label{subrotes} \\
\text{superluminal,}&\text{ rotational:}  & e_1 &= \frac{E-1}{3(c^2-1)} ,  &  e_2 &=\frac{4-E}{6(c^2-1)}  ,  &  e_3 &=\frac{E+2}{6(1-c^2)}, \\
\text{subluminal,}&\text{ librational:}    & e_1 &=\frac{E+2}{6(1-c^2)}  ,  &  e_2 &= \frac{E-1}{3(c^2-1)} ,  &  e_3 &=\frac{4-E}{6(c^2-1)}, \\
\text{superluminal,}&\text{ librational:}  & e_1 &=\frac{4-E}{6(c^2-1)}  ,  &  e_2 &=\frac{E-1}{3(c^2-1)} ,  &  e_3 &= \frac{E+2}{6(1-c^2)}. \label{suplumes}
\end{align}

\section{The linear stability problem}\label{setup}
To examine the linear stability of our solutions, we consider
\beq v(z,\tau) = f(z) + \epsilon w(z,\tau)+ \mathcal{O}\left(\epsilon^2\right), \label{wpert} \eeq
where $\epsilon$ is a small parameter. Substituting (\ref{wpert}) into (\ref{travelingSG}), we obtain at order $\epsilon$
\beq (c^2-1) w_{z z}-2c w_{z\tau}+w_{\tau\tau}+\cos\left(f(z)\right) w = 0. \label{weqn} \eeq
Letting $w_1(z,\tau) = w(z,\tau)$ and $w_2(z,\tau) = w_\tau(z,\tau)$ we rewrite (\ref{weqn}) as a first-order system of equations
\beq \frac{\partial}{\partial\tau}\left( \ba{c} w_1 \\ w_2 \ea \right) = \mathcal{L} \left( \ba{c} w_1 \\ w_2 \ea \right)= \left(\ba{cc} 0 & 1 \\
-(c^2-1)\partial_z^2-\cos\left(f(z)\right) & 2c \partial_z \ea \right) \left( \ba{c} w_1 \\ w_2 \ea \right).\label{wsystem}
\eeq

An elliptic solution $f(z)$ is linearly stable if for all $\epsilon >0$ there exists a $\delta>0$ such that if $||w(z,0)|| \, <\delta$ then $||w(z,\tau)|| \, <\epsilon$ for all $\tau>0$. This definition depends on the choice of norm $||\cdot||$, which is specified in the definition of the spectrum in (\ref{spectral-norm}) below.

Since (\ref{wsystem}) is autonomous in $\tau$, we separate variables to look for solutions of the form
\beq \left(\ba{c} w_1(z,\tau) \\ w_2(z,\tau) \ea \right) = e^{\lambda \tau} \left(\ba{c} W_1(z) \\ W_2(z) \ea \right), \label{wWeqn} \eeq
resulting in the spectral problem 
\beq \lambda \left( \ba{c} W_1 \\ W_2 \ea \right) = \mathcal{L} \left( \ba{c} W_1 \\ W_2 \ea \right)= \left(\ba{cc} 0 & 1 \\
-(c^2-1)\partial_z^2-\cos\left(f(z)\right) & 2c \partial_z \ea \right) \left( \ba{c} W_1 \\ W_2 \ea \right).\label{spectralproblem}
\eeq
Here
\beq \label{spectral-norm} \sigma_{\mathcal{L}}=\{\lambda\in\mathbb{C}: \max_{x\in \mathbb{R}} \left( |W_1(x)|,|W_2(x)| \right) < \infty\}, \eeq
or
\beq W_1,W_2\in C_b^0(\mathbb{R}). \eeq
For spectral stability, we need to demonstrate that the spectrum $\sigma_{\mathcal{L}}$ does not enter the open right half of the complex $\lambda$ plane.
Since (\ref{SGlab}) is Hamiltonian \cite{AS}, the spectrum of its linearization is symmetric with respect to both the real and imaginary axis \cite{W}. In other words, proving spectral stability for elliptic solutions to (\ref{SGlab}) amounts to proving that the stability spectrum lies strictly on the imaginary axis. We show that the elliptic solutions are spectrally stable only in the subluminal rotational case. We demonstrate spectral elements in the right-half plane near the origin for all choices of the parameters $E$ and $c$ outside the subluminal rotational regime.

\section{The Lax problem} \label{lax-section}
We wish to obtain an analytical representation for the spectrum $\sigma_{\mathcal{L}}$. As mentioned in the introduction, this analytical representation comes from the squared eigenfunction connection between the linear stability problem (\ref{spectralproblem}) and the Lax pair of (\ref{SGlab}).
The Lax pair for sine-Gordon is well known \cite{AKNS-SG,AS,AKNS,Kaup75}. The compatibility condition $\chi_{xt} = \chi_{tx}$ of the Lax pair,
\beq  \chi_{x} = \left( \ba{cc}\displaystyle -\frac{i \zeta}{2}+\frac{i \cos(u)}{8\zeta} & \displaystyle \frac{i \sin(u)}{8\zeta}-\frac{1}{4}\left(u_x+u_t\right) \\ \displaystyle
 \frac{i \sin(u)}{8\zeta}+\frac{1}{4}\left(u_x+u_t\right) & \displaystyle  \frac{i \zeta}{2}-\frac{i \cos(u)}{8\zeta}  \ea \right) \chi , \eeq
\beq \chi_{t} = \left( \ba{cc} \displaystyle -\frac{i \zeta}{2}-\frac{i \cos(u)}{8\zeta} & \displaystyle -\frac{i \sin(u)}{8\zeta}-\frac{1}{4}\left(u_x+u_t\right) \\ \displaystyle
 -\frac{i \sin(u)}{8\zeta}+\frac{1}{4}\left(u_x+u_t\right) & \displaystyle  \frac{i \zeta}{2}+\frac{i \cos(u)}{8\zeta}  \ea \right) \chi , \eeq
is (\ref{SGlab}). 
We transform the Lax pair by moving into a traveling reference frame letting $z=x-ct,$ $\tau=t,$ and $v(z,\tau)= u(x,t).$ Additionally, to examine the stationary solutions we let $v(z,\tau) = f(z)$ so that
\beq \chi_z = \left( \ba{cc} C & D \\ -D^* & -C \ea \right)\chi, \label{LaxZsimp} \eeq
\beq \chi_\tau = \left( \ba{cc}  A & B \\ -B^* & -A  \ea \right)\chi, \label{LaxTausimp} \eeq
where $^*$ represents the complex conjugate, and 
\begin{align}
A &=  -\frac{ i \left(4(1+c)\zeta^2-(c-1) \cos(f(z))\right)}{8\zeta}, \\
B &= \frac{(c-1) \left(i \sin(f(z))+2(c+1) \zeta f'(z)\right)}{8\zeta}, \\
C &=  -\frac{i \zeta}{2}+\frac{i \cos(f(z))}{8\zeta},\label{Aeqn}\\
D &= \frac{i \sin(f(z))}{8\zeta}-\frac{f'(z)}{4}+\frac{cf'(z)}{4},
\label{Deqn}
\end{align}
whose compatibility condition  $\chi_{z\tau} = \chi_{\tau z}$ is (\ref{pendulumeqn}).
We define $\sigma_L$, or informally the Lax spectrum, as the set of all $\zeta$ for which (\ref{LaxZsimp}) has a bounded (in z) solution. 
Examining (\ref{LaxTausimp}), since $A$ and $B$ are independent of $\tau,$ we separate variables. Let
\beq \chi(z,\tau) = e^{\Omega \tau} \varphi(z), \label{chiO} \eeq
with $\Omega$ being independent of $\tau,$ but possibly depending on $z$. Substituting (\ref{chiO}) into (\ref{LaxTausimp}) and canceling the exponential, we find 
\beq \left( \ba{cc} A-\Omega & B \\ -B^* & -A- \Omega \ea \right) \varphi = 0. \label{OmegaMat}\eeq
To have nontrivial solutions, we require the determinant of (\ref{OmegaMat}) to be zero. Using the definitions of $A$ and $B$, we get
\beq \Omega^2 = A^2 -B B^* = \frac{1}{64} \left( -8(c^2-1)(E-1) -\frac{(c-1)^2}{\zeta^2}- 16(c+1)^2 \zeta^2\right).\label{Omegaeqn} \eeq
As expected, $\Omega$ is independent of both $\tau$ (by construction) and $z$ (by integrability). Thus $\Omega$ is strictly a function of $\zeta$ and the solution parameters $c$ and $E$. We remark that $\Omega$ takes the form (\ref{Omegaeqn}) for all values of $c$ and $E$ regardless of where we are in parameter space.

To satisfy (\ref{OmegaMat}), we let
\beq \varphi(z) = \gamma(z) \left(\ba{c} -B(z) \\ A(z)-\Omega \ea \right), \label{varphieqn} \eeq
where $\gamma(z)$ is a scalar function. By construction of $\varphi(z)$, $\chi(z,\tau)$ satisfies (\ref{LaxTausimp}). Since (\ref{LaxZsimp}) and (\ref{LaxTausimp}) commute, it is possible to choose $\gamma(z)$ such that $\chi$ also satisfies (\ref{LaxZsimp}). Indeed, $\gamma(z)$ satisfies a first-order linear equation, whose solution is given by
\beq \gamma(z) =\gamma_0 \exp\left(\int \frac{-C(A-\Omega)+B D^*-A_z}{A-\Omega}\text{d}z\right). \label{gammaeqn} \eeq
For almost every $\zeta\in\mathbb{C}$, we have explicitly determined the two linearly independent solutions of (\ref{LaxZsimp}), {\em i.e.}, those corresponding to the positive and negative signs of $\Omega$ in (\ref{Omegaeqn}). Assuming $\Omega \ne 0$ these two solutions are, by construction linearly independent. In the case where $\zeta$ is a root of $\Omega,$ the second solution to (\ref{LaxZsimp}) can be determined via the reduction-of-order method. 

Since (\ref{LaxZsimp}) and (\ref{LaxTausimp}) share eigenfunctions, $\sigma_L$ is the set of all $\zeta \in \mathbb{C}$ such that (\ref{varphieqn}) is bounded for all $z\in \mathbb{R}$. The vector part of $\varphi$ is bounded for all $z$, so we only need that the scalar function $\gamma(z)$ is bounded as $z\rightarrow \pm \infty$. A necessary and sufficient condition for this is 
\beq \left< \text{Re} \left( \frac{-C(A-\Omega)+B D^*-A_z}{A-\Omega} \right) \right> =0 \label{intcond},\eeq
where $\left<\cdot\right>$ is the average over one period $2\mathcal{K}(k)$ of the integrand, and $\text{Re}$ denotes the real part. The integral condition (\ref{intcond}) completely determines the Lax spectrum $\sigma_L$.

\section{The squared eigenfunction connection}\label{sqeig-section}
A connection between the eigenfunction of the Lax pair (\ref{LaxZsimp}) and (\ref{LaxTausimp}) and the eigenfunctions of the linear stability problem (\ref{spectralproblem}) using squared eigenfunctions is well known \cite{AKNS}. We prove the following theorem.
\vspace{2mm}
\begin{thm}
The sum of squares,
\beq w(z,\tau) = \chi_1(z,\tau)^2+\chi_2(z,\tau)^2, \label{sqeig} \eeq
satisfies the linear stability problem (\ref{weqn}) for $f(z)$. Here $\chi=(\chi_1,\chi_2)^T$ is any solution of (\ref{LaxZsimp}-\ref{LaxTausimp}).
\end{thm} 
\begin{proof}
The proof is done by direct calculation. Substitute $w(z,\tau)$ into the left-hand side of (\ref{weqn}). Eliminate $z$-derivatives of $\chi_1$ and $\chi_2$ (up to order 2) using (\ref{LaxZsimp}) and eliminate $\tau$-derivatives of $\chi_1$ and $\chi_2$ (up to order 2) using (\ref{LaxTausimp}). The resulting expression for the left-hand side is 0, thus demonstrating that (\ref{weqn}) is satisfied, finishing the proof.
\end{proof}
To establish the connection between $\sigma_{\mathcal{L}}$ and $\sigma_L,$ we examine the right- and left-hand sides of (\ref{wWeqn}). Substituting (\ref{sqeig}) and (\ref{chiO}) to the left hand side of (\ref{wWeqn}) we find
\beq e^{2\Omega \tau} \left( \ba{c} \varphi_1^2+\varphi_2^2 \\ 2\Omega\left(\varphi_1^2+\varphi_2^2\right) \ea \right) = e^{\lambda \tau} 
\left( \ba{c} W_1(z) \\ W_2(z) \ea \right), \label{sqeigconnection}
\eeq
so we conclude that 
\beq \lambda = 2\Omega(\zeta), \label{lambdaconnection} \eeq
with eigenfunctions given by 
\beq \left( \ba{c} W_1(z) \\ W_2(z) \ea \right) =  \left( \ba{c} \varphi_1^2+\varphi_2^2 \\ 2\Omega\left(\varphi_1^2+\varphi_2^2\right) \ea \right). \label{W12eqns} \eeq
This gives the connection between the $\sigma_L$ spectrum and the $\sigma_{\mathcal{L}}$ spectrum. It is also necessary to check that indeed all solutions of (\ref{spectralproblem}) are obtained through (\ref{sqeigconnection}). This is not shown explicitly here, but is done analogous to the work in \cite{BD,BDN}.

Although in principle the above construction determines $\sigma_{\mathcal{L}}$, it remains to be seen how practical this determination is. In the following section we discuss a technique for explicitly integrating (\ref{intcond}) using Weierstrass elliptic functions, leading to a more explicit characterization of $\sigma_{\mathcal{L}}$.

\section{The Lax spectrum in terms of elliptic functions}\label{Weierstrass}
In terms of Weierstrass elliptic functions, (\ref{intcond}) becomes
\beq \text{Re} \int_0^{2\omega_1}  \frac{-C(A-\Omega)+B D^*-A_z}{A-\Omega} \text{d}z =0, \label{intcond2} \eeq
with $A$, $B$, $C,$ and $D$ given in (\ref{Aeqn}-\ref{Deqn}). 
Substituting in for $f(z)$ using (\ref{feqnW}) we find that (\ref{intcond2}) is
\beq \text{Re} \int_0^{2\omega_1}  \frac{C_1 +C_2 \wp (z) + C_3 \wp'(z)}{C_4+C_5 \wp(z)} \text{d}z =0, \label{intcond3} \eeq
with $\wp(z) = \wp(z+\omega_3,g_2,g_3)$ with the dependence on $\omega_3$, $g_2,$ and $g_3$ suppressed. The $C_j$'s depend on $\zeta$ but are independent of $z$. Like $\Omega(\zeta)$, the $C_j$'s take one form regardless of where the solution is in parameter space. They are given by
\begin{align}
C_1 &= \frac{1}{3} \left(-3i -16i(E-1) \zeta^2+48 i \zeta^4+3i c\left(1+8(E-1)\zeta^2+16\zeta^4\right)+\left(8(e-1) \zeta+96 \zeta^3\right)  \Omega(\zeta) \right), \\
C_2 &= 16(c^2-1) \zeta \left(i \zeta + \Omega(\zeta) \right), \\
C_3 &= -8 (c-1)^2(c+1) \zeta, \\
C_4 &= \frac{8}{3} \zeta \left( (c-1)(E-1)+12 (c+1) \zeta^2 - 24 i \zeta \Omega(\zeta) \right),\\
C_5 &= 16 (c-1)^2 (c+1) \zeta.
\end{align}
We evaluate the integral in (\ref{intcond3}) explicitly \cite{GR}. We find
\beq \text{Re}\left[\frac{2 \omega_1 C_2}{C_5} +\frac{4\left(C_1 C_5- C_2 C_4\right)}{\wp'(\rho) C_5^2} \left(\zeta_w(\rho) \omega_1-\zeta_w(\omega_1) \rho\right)\right] = 0,\label{intcond4} \eeq
with 
\beq \rho = \rho(\zeta) = \wp^{-1}\left(- \frac{C_4(\zeta)}{C_5(\zeta)},g_2,g_3\right), \eeq
and $\zeta_w$ is the Weierstrass Zeta function \cite{O}. Note that $\wp^{-1}$ is a multivalued function, but for our analysis $\rho$ is chosen as any value such that $\wp(\rho) = -C_4(\zeta)/C_5(\zeta).$ 
Substituting for the $C_j$'s (\ref{intcond4}) becomes
\beq \text{Re} \left[ \frac{2\omega_1\left(i \zeta +\Omega(\zeta)\right)}{c-1}+\frac{4\zeta \left(-i(c-1)(E-1)-4i(c+1)\zeta^2-8 \zeta \Omega(\zeta)\right)}{(c-1)^3(c+1)\wp'(\rho)}\left(\zeta_w(\rho) \omega_1-\zeta_w(\omega_1) \rho\right) \right] =0. \label{intcond5} \eeq
We simplify this further by recognizing that 
\beq {\wp'}^2(\rho)=4\wp^3(\rho)-g_2 \wp(\rho)-g_3 = 4\left(-\frac{C_4(\zeta)}{C_5(\zeta)}\right)^3-g_2\left(-\frac{C_4(\zeta)}{C_5(\zeta)}\right)-g_3. \eeq
Substituting in for $C_4(\zeta)$ and $C_5(\zeta)$ gives
\beq {\wp'}^2(\rho) = 4\left(\frac{\zeta \left(-i(c-1)(E-1)-4i(c+1)\zeta^2-8 \zeta \Omega(\zeta)\right)}{(c-1)^3(c+1)}\right)^2. \label{wpprime} \eeq
Thus (\ref{intcond5}) simplifies to
\beq \text{Re}\left( \frac{2\omega_1\left(i \zeta +\Omega(\zeta)\right)}{c-1}+ 2\nu\left(\zeta_w(\rho) \omega_1-\zeta_w(\omega_1) \rho\right) \right) =0, \label{intcond6}\eeq
where 
\beq \nu = \begin{cases} +1 &\mbox{if } -\frac{\pi}{2} <\arg\left(\frac{\zeta \left(-i(c-1)(E-1)-4i(c+1)\zeta^2-8 \zeta \Omega(\zeta)\right)}{(c-1)^3(c+1)}\right) \le \frac{\pi}{2}, \\ 
-1 & \mbox{otherwise.} \end{cases} \eeq
Using (\ref{omega-nums}), and applying the formula for the Weierstrass $\zeta$ function evaluated at a half period \cite{BF}, $\zeta_w(\omega_1)=\sqrt{e_1-e_3}\left(\mathcal{E}(k)-\frac{e_1}{e_1-e_3} \mathcal{K}(k)\right),$ (\ref{intcond6}) becomes
\beq \text{Re}\left[  \frac{2\mathcal{K}(k)\left(i \zeta +\Omega(\zeta)\right)}{c-1}+ 2\nu\left(\zeta_w(\rho) \mathcal{K}(k)-\sqrt{e_1-e_3}\left(\mathcal{E}(k)-\frac{e_1}{e_1-e_3} \mathcal{K}(k)\right)\rho\right) \right] =0. \label{intcond7} \eeq
Here 
\beq \mathcal{E}(k) = \int_0^{\pi/2} \sqrt{1-k^2 \sin^2 y}\, \text{d}y, \eeq
is the complete elliptic integral of the second kind. 
We have simplified the integral condition (\ref{intcond2}) significantly. Thus $\zeta\in \sigma_L$ if and only if (\ref{intcond7}) is satisfied. To simplify notation, let
\beq I(\zeta) =  \frac{2\omega_1\left(i \zeta +\Omega(\zeta)\right)}{c-1}+ 2\nu\left(\zeta_w(\rho) \omega_1-\zeta_w(\omega_1) \rho\right), \label{Ieqn} \eeq
so that (\ref{intcond7}) is 
\beq \text{Re} \left[ I(\zeta) \right] =0. \label{intcond-simple}\eeq
Next, we wish to examine the level sets of the left-hand side of (\ref{intcond-simple}). To this end, we differentiate $I(\zeta)$ with respect to $\zeta$. To evaluate this derivative we use the chain rule and note that
\beq \frac{\partial}{\partial \zeta} \zeta_w(\rho) = -\wp(\rho)\frac{\partial \rho}{\partial \zeta} = \frac{C_4(\zeta)}{C_5(\zeta)}\frac{\textrm{d} \wp^{-1}}{\textrm{d} \zeta} \left(-\frac{C_4(\zeta)}{C_5(\zeta)},g_2,g_3\right) \left(-\frac{C_4(\zeta)}{C_5(\zeta)}\right)'.\eeq
Since
\beq \frac{\textrm{d}}{\textrm{d}z} \wp^{-1}\left(-\frac{C_4(\zeta)}{C_5(\zeta)},g_2,g_3\right)=\frac{1}{\wp'\left (\wp^{-1}\left(-\frac{C_4(\zeta)}{C_5(\zeta)},g_2,g_3\right)\right)}=\frac{1}{\wp'(\rho)}, \eeq
we use (\ref{wpprime}) to obtain
\beq \frac{\text{d}I(\zeta)}{\text{d}\zeta}=\frac{3 (c-1)\omega_1-48(c+1)\zeta^4 \omega_1-8 \zeta^2\left(6(c^2-1)\zeta_w(\omega_1)+(1-E)\omega_1\right)}{96\zeta^3\Omega(\zeta)}. \label{derivintcond} \eeq 
Taking the real part of (\ref{derivintcond}) does not give another characterization of the spectrum.  Instead, if we think of (\ref{intcond7}) as restricting to the zero level set of the left-hand side. Then we use (\ref{derivintcond}) to determine a tangent vector field which allows us to map out level curves originating from any point. This is explained in more detail in Section~\ref{imagaxisregions}. 

\section{The $\sigma_{\mathcal{L}}$ spectrum on the imaginary axis}\label{imagaxisregions}
In this section we discuss $\sigma_{\mathcal{L}}\cap i\mathbb{R}$. As we demonstrate, this corresponds to the part of $\sigma_L$ lying on the real axis for both rotational and librational waves, as well as a part of $\sigma_L$ lying on the imaginary axis for rotational waves. Using (\ref{intcond7}) we obtain analytic expressions for $\sigma_L$ corresponding to $\sigma_{\mathcal{L}}\cap i \mathbb{R}$.

We begin by considering $\zeta\in \mathbb{R}$. As we demonstrate below, (\ref{intcond7}) is satisfied for any real $\zeta$. Using (\ref{Omegaeqn}) and (\ref{lambdaconnection}), we determine the corresponding parts of $\sigma_{\mathcal{L}}$.

\vspace{2mm}
\begin{thm}\label{realzetathm}
The condition (\ref{intcond7}) is satisified for $\zeta\in\mathbb{R}$.
\end{thm}
\begin{proof}
Since $k$, $c$, $\mathcal{K}(k)$, and $\mathcal{E}(k)$ are real, it suffices to show that $\Omega(\zeta)\in i \mathbb{R}$, $\rho\in i \mathbb{R}$, and $\zeta_w(\rho)\in i \mathbb{R}$. Rewriting (\ref{Omegaeqn}) in the superluminal case,
\beq \Omega^2(\zeta) = -\frac{1}{64} \left(\left(-4(1+c)\zeta+\frac{c-1}{\zeta}\right)^2+8 E (c^2-1)\right), \label{Omegasup} \eeq
and in the subluminal case,
\beq \Omega^2(\zeta) = -\frac{1}{64} \left(\left(-4(1+c)\zeta-\frac{c-1}{\zeta}\right)^2+8 (2-E) (1-c^2)\right). \label{Omegasub} \eeq
In either case $\Omega^2 \le 0$ and $\Omega(\zeta)\in i \mathbb{R}$.
Since $\zeta_w$ with $g_2,g_3\in \mathbb{R}$ takes real values to real values and purely imaginary values to purely imaginary values \cite{O}, to prove $\zeta_w(\rho)\in i \mathbb{R}$ it suffices to show that $\rho = \wp^{-1}\left(-\frac{C_4(\zeta)}{C_5(\zeta)},g_2,g_3\right)\in i \mathbb{R}.$ For $g_2,g_3\in \mathbb{R}$, $\wp(\mathbb{R},g_2,g_3)$ maps to $[e_1,\infty)$, and since $\wp(i x,g_2,g_3)=-\wp(x,g_2,-g_3)$ we have that $\wp(i \mathbb{R},g_2,g_3)$ maps to $(-\infty,e_3]$. Thus we need to show that for $\zeta\in \mathbb{R}$, $-\frac{C_4(\zeta)}{C_5(\zeta)}\le e_3.$ Again we split into cases. In the superluminal case, we want to show
\beq \frac{(c-1)(E-1)+12 (c+1)\zeta^2-24i \zeta \Omega(\zeta)}{6(c-1)(1-c^2)}\le \frac{E+2}{6(1-c^2)}. \eeq
Substituting in for $\Omega(\zeta)$ using (\ref{Omegasup}) and simplifying the left- and right-hand sides of this expression yields
\beq \frac{4(c+1)\zeta^2}{c-1}+\frac{\sqrt{\left(-4(1+c)\zeta^2+(c-1)\right)^2+8 E(c^2-1)\zeta^2}}{c-1}\ge 1. \label{supeqnWTS} \eeq
Since $ \sqrt{\left(-4(1+c)\zeta^2+(c-1)\right)^2+8 E(c^2-1)\zeta^2} \ge \sqrt{\left(-4(1+c)\zeta^2+(c-1)\right)^2}$, (\ref{supeqnWTS}) is satisfied.
For the subluminal case we proceed in a similar manner, noting the different value of $e_3$ from (\ref{subrotes}-\ref{suplumes}). We want to show
\beq \frac{(c-1)(E-1)+12 (c+1)\zeta^2-24i \zeta \Omega(\zeta)}{6(c-1)(1-c^2)}\le \frac{E-4}{6(1-c^2)}. \eeq
Substituting in for $\Omega(\zeta)$ using (\ref{Omegasub}) and simplifying the left- and right-hand sides of this expression yields
\beq \frac{4(c+1)\zeta^2+\sqrt{\left(-4(c+1)\zeta^2+(1-c)\right)^2+8 (2-E)(1-c^2)\zeta^2}}{1-c}\ge 1. \label{subeqnWTS} \eeq
Since $\sqrt{\left(-4(c+1)\zeta^2+(1-c)\right)^2+8 (2-E)(1-c^2)\zeta^2}\ge \sqrt{\left(-4(c+1)\zeta^2+(1-c)\right)^2}$, (\ref{subeqnWTS}) is satisfied.
\end{proof}

\begin{figure}
\begin{tabular}{cc}
  \includegraphics[width=75mm]{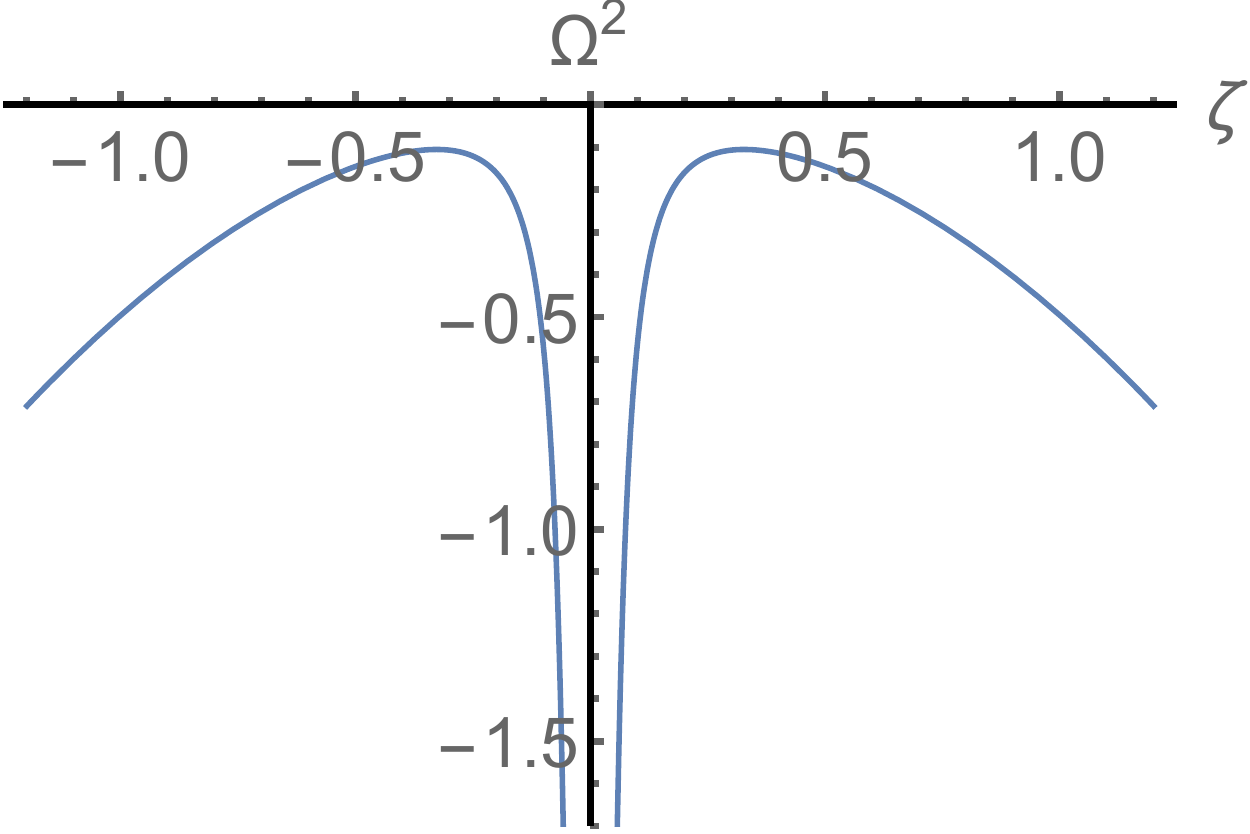} & \includegraphics[width=75mm]{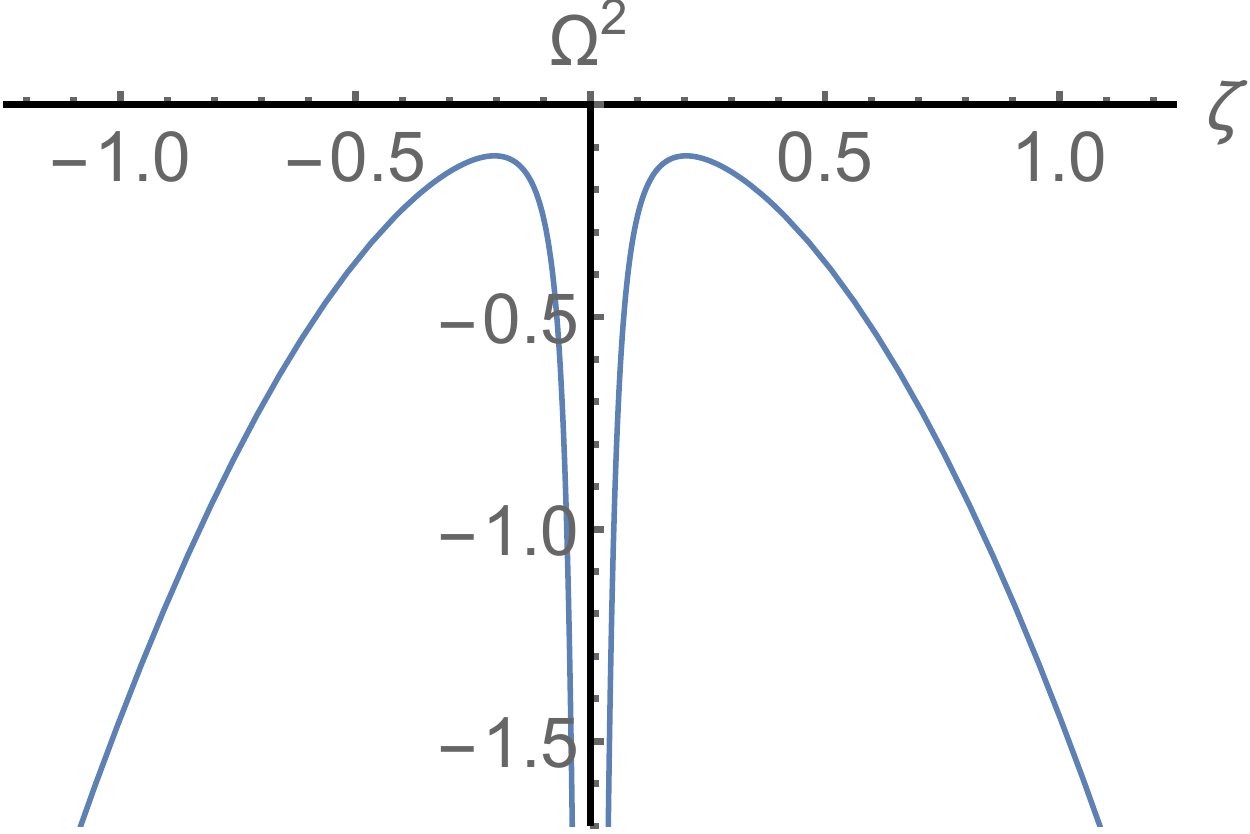} \\
(a) Subluminal: $c^2<1$ & (b)  Superluminal: $c^2>1$ 
\end{tabular}
\caption{$\Omega^2$ as a function of real $\zeta$ for subluminal and superluminal waves: (a) subluminal: $c=0.4$ and $E=1,$ and (b) superluminal: $c=1.4$ and $E=1$. }
\label{SupSubOmega}
\end{figure}

At this point, we know that $\mathbb{R}\subset \sigma_L$. We wish to see what this corresponds to for $\sigma_{\mathcal{L}}$. For convenience define 
\beq S_\Omega = \left\{\Omega : \Omega^2 = \frac{1}{64} \left( -8(c^2-1)(E-1) -\frac{(c-1)^2}{\zeta^2}- 16(c+1)^2 \zeta^2\right) \text{ and } \zeta\in\sigma_L\right\}. \eeq
As was seen in the proof of Theorem~\ref{realzetathm}, when $\zeta\in\mathbb{R}$, $\Omega(\zeta)\in i \mathbb{R}$. Applying (\ref{lambdaconnection}), we see that $\zeta\in \mathbb{R}$ corresponds to $\sigma_{\mathcal{L}}\cup i \mathbb{R}$. Representative plots of $\Omega^2$ are shown in Figure~\ref{SupSubOmega}. The subset of $S_\Omega$ corresponding to $\zeta\in \mathbb{R}$ consists of $(-i \infty,-i|\Omega_m|]\cup [i |\Omega_m|,i \infty),$ where $\Omega_m^2$ is the maximum value of $\Omega^2$. The set $S_\Omega$ corresponding to $\zeta\in \mathbb{R}$ is quadruple covered as for all values of $\Omega$ there are four values of $\zeta$ which map to it, except at $\Omega=\pm \Omega_m$, where just two values of $\zeta$ map to it. $\Omega_m$ can be found explicitly by finding the extrema of $\Omega^2(\zeta)$. In the subluminal case, $\Omega^2(\zeta)$ reaches its maxima at 
\beq \zeta_m = \pm \sqrt{\frac{1-c}{4(1+c)}}, \;\; \Omega_m^2 = -\frac{1}{8}(1-c^2)(2-E), \eeq
and in the superluminal case, $\Omega^2(\zeta)$ reaches its maxima at 
\beq \zeta_m = \pm \sqrt{\frac{c-1}{4(1+c)}}, \;\; \Omega_m^2 = -\frac{1}{8}(c^2-1)E. \eeq
Applying (\ref{lambdaconnection}) we have $\left(-i \infty, -\lambda_1\right] \cup \left[\lambda_1,i \infty\right) \subset \sigma_{\mathcal{L}}$ where 
\beq \lambda_1 = i \sqrt{\frac{(1-c^2)(2-E)}{2}}, \eeq
in the subluminal case, and 
\beq \lambda_1 = i \sqrt{\frac{E(c^2-1)}{2}}, \eeq
in the superluminal case.

If $\zeta$ satisfies $\Omega^2(\zeta)=0$, then $\zeta$ must satisfy (\ref{intcond7}). This is due to the fact that the origin is always included in $\sigma_{\mathcal{L}}$ and hence in $S_\Omega$. The fact that there are four roots of $\Omega^2(\zeta)=0$ corresponds to the fact that $0\in \sigma_{\mathcal{L}}$ with multiplicity four. This is seen from the symmetries of (\ref{SGlab}) and by applying Noether's Theorem \cite{S07,KP}. For rotational waves, the roots of $\Omega^2(\zeta)$ lie on the imaginary axis. For the subluminal rotational case the roots are:
\beq \zeta_c = \left\{ \frac{\sqrt{1-c^2}}{2\sqrt{2}(c+1)} \left(\sqrt{-E}\pm \sqrt{2-E}\right) i, -\frac{\sqrt{1-c^2}}{2\sqrt{2}(c+1)} \left(\sqrt{-E}\pm \sqrt{2-E}\right) i\right\}, \eeq
and in the superluminal rotational case the roots are:
\beq \zeta_c = \left\{\frac{\sqrt{c^2-1}}{2 \sqrt{2}(c+1)} \left(\sqrt{E}\pm \sqrt{E-2} \right) i,- \frac{\sqrt{c^2-1}}{2 \sqrt{2}(c+1)} \left(\sqrt{E}\pm \sqrt{E-2} \right) i\right\}. \label{suprotrootsOmega} \eeq
We label the four roots $\zeta_1$, $\zeta_2$, $\zeta_3$, and $\zeta_4$ where $\text{Im}(\zeta_1)< \text{Im}(\zeta_2)< \text{Im}(\zeta_3)< \text{Im}(\zeta_4).$ They are labeled for reference in Figure~\ref{SupSubRotOmega}.
\vspace{2mm}
\begin{thm}\label{imagzetathm}
For rotational waves, the condition (\ref{intcond7}) is satisified for all $\zeta\in i \mathbb{R}$ such that $\text{Im}(\zeta_1)\le \text{Im}(\zeta) \le \text{Im}(\zeta_2)$ or $\text{Im}(\zeta_3)\le \text{Im}(\zeta) \le \text{Im}(\zeta_4)$.
\end{thm}
\begin{proof}
The level curve (\ref{intcond-simple}), is exactly the condition (\ref{intcond7}). We examine the tangent vector field to (\ref{intcond-simple}). If we let $\zeta = \zeta_r+i \zeta_i$, then
\beq I(\zeta) = I(\zeta_r+i \zeta_i ) = \frac{2\omega_1\left(i (\zeta_r+i \zeta_i) +\Omega(\zeta_r+i \zeta_i) \right)}{c-1}+ 2\nu\left(\zeta_w(\rho) \omega_1-\zeta_w(\omega_1) \rho\right). \eeq
Taking derivatives with respect to $\zeta_r$ and $\zeta_i$ gives a normal vector field to level curves of the general condition $\text{Re}\left[I(\zeta)\right]=C$ for any constant $C$, specifically, the normal vector is given by 
$$ \left(\frac{\textrm{d}\textrm{Re} \left[I(\zeta_r+i \zeta_i)\right]}{\textrm{d} \zeta_r},\frac{\textrm{d} \textrm{Re}\left[I(\zeta_r+i \zeta_i)\right]}{\textrm{d} \zeta_i} \right).$$
Thus, the tangent vector field is
$$ \left(-\frac{\textrm{d} \textrm{Re}\left[I(\zeta_r+i \zeta_i)\right]}{\textrm{d} \zeta_i},\frac{\textrm{d}\textrm{Re} \left[I(\zeta_r+i \zeta_i)\right]}{\textrm{d} \zeta_r} \right).$$
By applying the chain rule and using $\textrm{Re}[i z] = -\textrm{Im}[z],$ we have that the tangent vector field to the level curves is
\beq \left(\textrm{Im}\left[ \frac{\textrm{d} I}{\textrm{d}\zeta} \right],\textrm{Re} \left[ \frac{\textrm{d} I}{\textrm{d} \zeta}\right] \right). \label{tangentvectorfield} \eeq
Where $ \frac{\textrm{d} I}{\textrm{d}\zeta}$ is given in (\ref{derivintcond}). We note that the numerator of (\ref{derivintcond}) is strictly real for $\zeta\in i \mathbb{R}$, thus
\beq \text{Im}\left[\frac{\textrm{d} I}{\textrm{d} \zeta}\right] = \left(3 (c-1)\omega_1-48(c+1)\zeta^4 \omega_1-8 \zeta^2\left(3(c^2-1)\zeta_w(\omega_1)+(1-E)\omega_1\right)\right) \text{Im}\left[\frac{1}{96\zeta^3\Omega(\zeta)}\right].
\eeq
Since $\Omega^2(\zeta) \le 0$ for $\zeta\in [\zeta_1,\zeta_2]\subset i \mathbb{R}$ and for $\zeta \in [\zeta_3,\zeta_4]\subset i \mathbb{R}$ we have that
\beq \text{Im}\left[\frac{\textrm{d} I}{\textrm{d} \zeta}\right] = 0, \eeq
and thus $\text{Re}\left[I(\zeta)\right] = C$ on these intervals. Since $\text{Re}\left[I(\zeta)\right] = 0$ at the endpoints $\zeta_c$, $C=0$ and (\ref{intcond7}) is satisfied.
\end{proof}

\begin{figure}
\begin{tabular}{cc}
  \includegraphics[width=75mm]{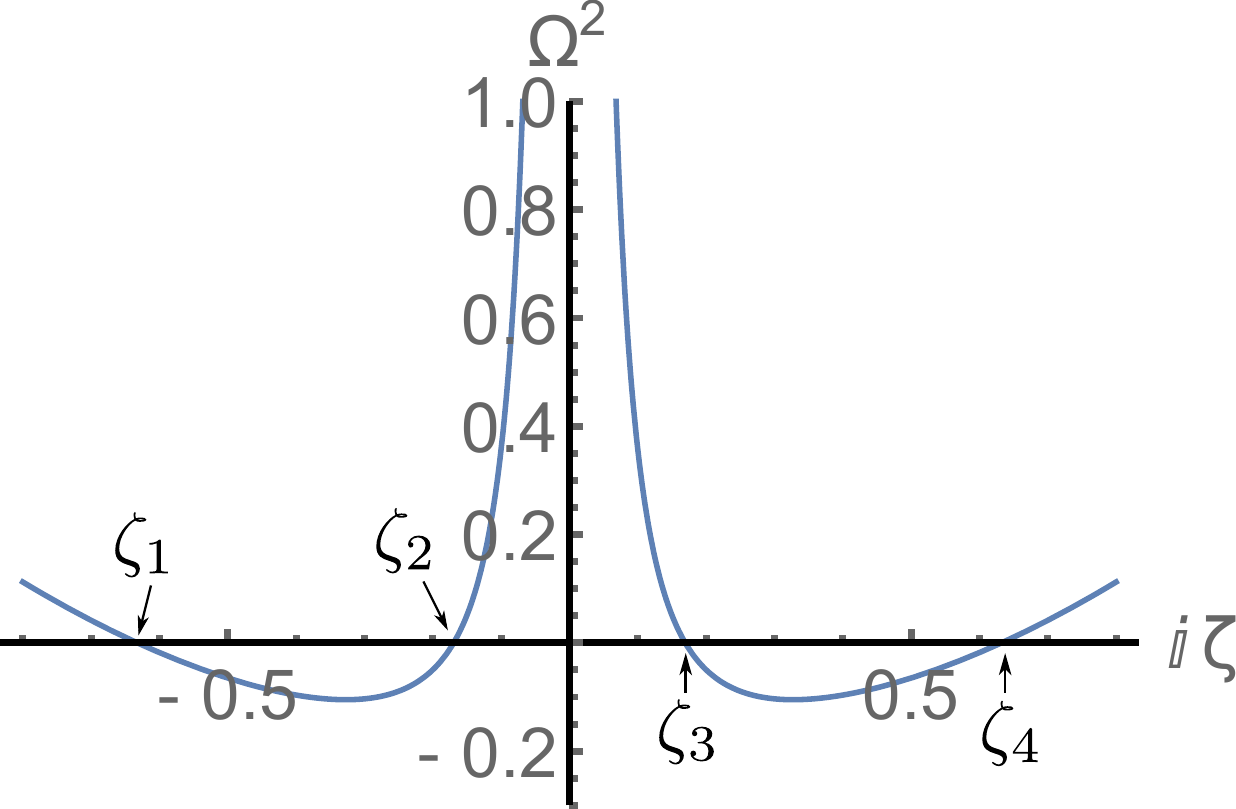} & \includegraphics[width=75mm]{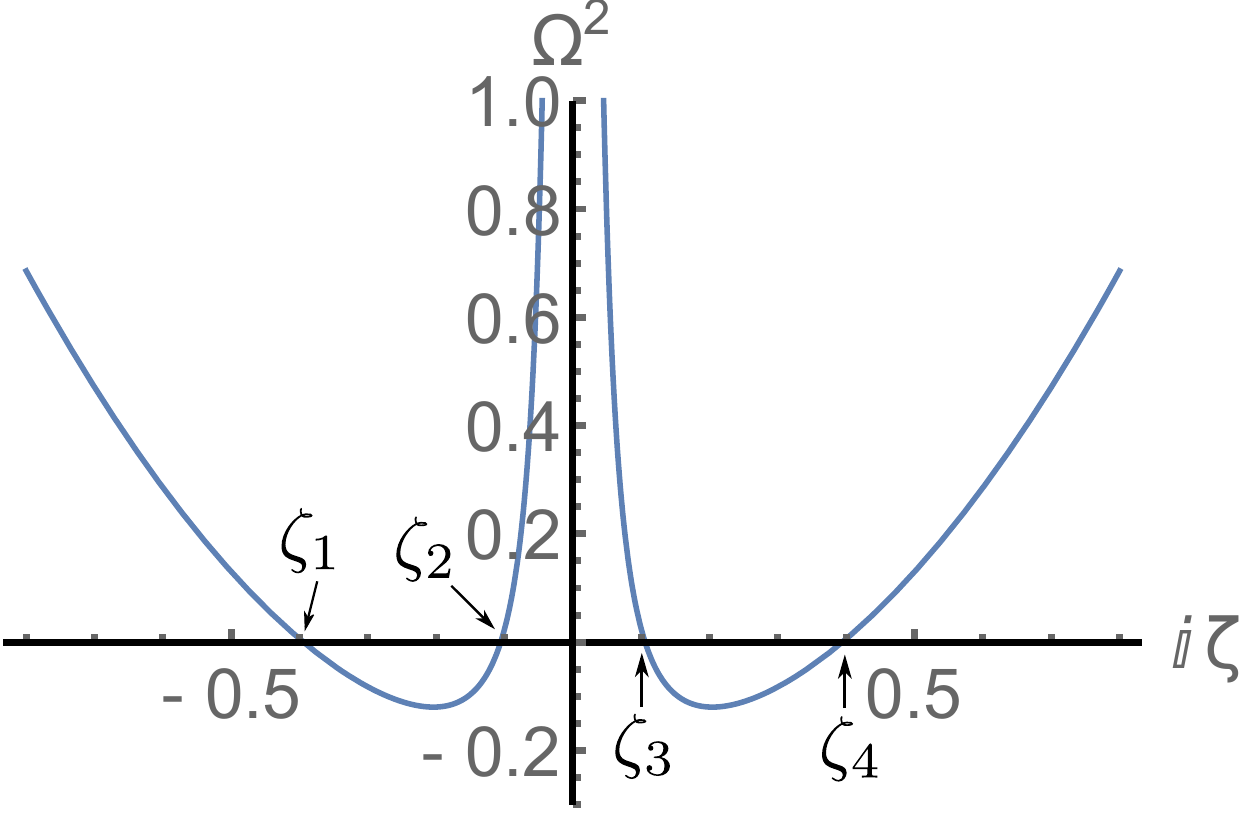} \\
(a) Subluminal: $c^2<1$ & (b)  Superluminal: $c^2>1$ 
\end{tabular}
\caption{$\Omega^2$ as a function of $i \zeta$, $\zeta\in \mathbb{R}$ for subluminal and superluminal rotational waves: (a) subluminal rotational: $c=0.4$ and $E=-1,$ and (b) superluminal rotational: $c=1.4$ and $E=3$. }
\label{SupSubRotOmega}
\end{figure}

At this point we know that $[\zeta_1,\zeta_2]\cup [\zeta_3,\zeta_4]\subset \sigma_L$. We wish to see what this corresponds to for $\sigma_{\mathcal{L}}$. 
Representative plots of $\Omega^2(i \zeta)$, $\zeta\in\mathbb{R}$ are shown in Figure~\ref{SupSubRotOmega}. The subset of $S_\Omega$ corresponding to $\zeta\in [\zeta_1,\zeta_2]\cup [\zeta_3,\zeta_4]$ consists of $[- i |\Omega_n|,0]$, where $\Omega_n^2$ is the minimum value of $\Omega^2$. The set $S_\Omega$ corresponding to $\zeta\in [\zeta_1,\zeta_2]\cup [\zeta_3,\zeta_4]$ is quadruple covered, except at the points $\pm \Omega_n$, where the set is double-covered.  $\Omega_n$ can be found explicitly by finding the extrema of $\Omega^2(i \zeta)$. In the subluminal rotational case, $\Omega^2(\zeta)$ reaches its minima at 
\beq \zeta_n =\pm \sqrt{\frac{1-c}{4(c+1)}} i,\;\; \Omega_n^2(\zeta) = \frac{1}{8}(1-c^2)E,\eeq
and in the superluminal rotational case, $\Omega^2(\zeta)$ reaches its minima at 
\beq \zeta_n = \pm \sqrt{\frac{c-1}{4(c+1)}} i,\;\; \Omega_n^2(\zeta) = \frac{1}{8}(c^2-1)(E-2).\eeq
Applying (\ref{lambdaconnection}) we have $\left[-\lambda_2,\lambda_2\right]\subset \sigma_{\mathcal{L}}$ where 
\beq \lambda_2 = \sqrt{\frac{(1-c^2)E}{2}} i ,\eeq
in the subluminal rotational case, and
\beq \lambda_2 = \sqrt{\frac{(c^2-1)(E-2)}{2}} i, \eeq
in the superluminal rotational case.

\begin{figure}
\centering
\includegraphics[width=10cm]{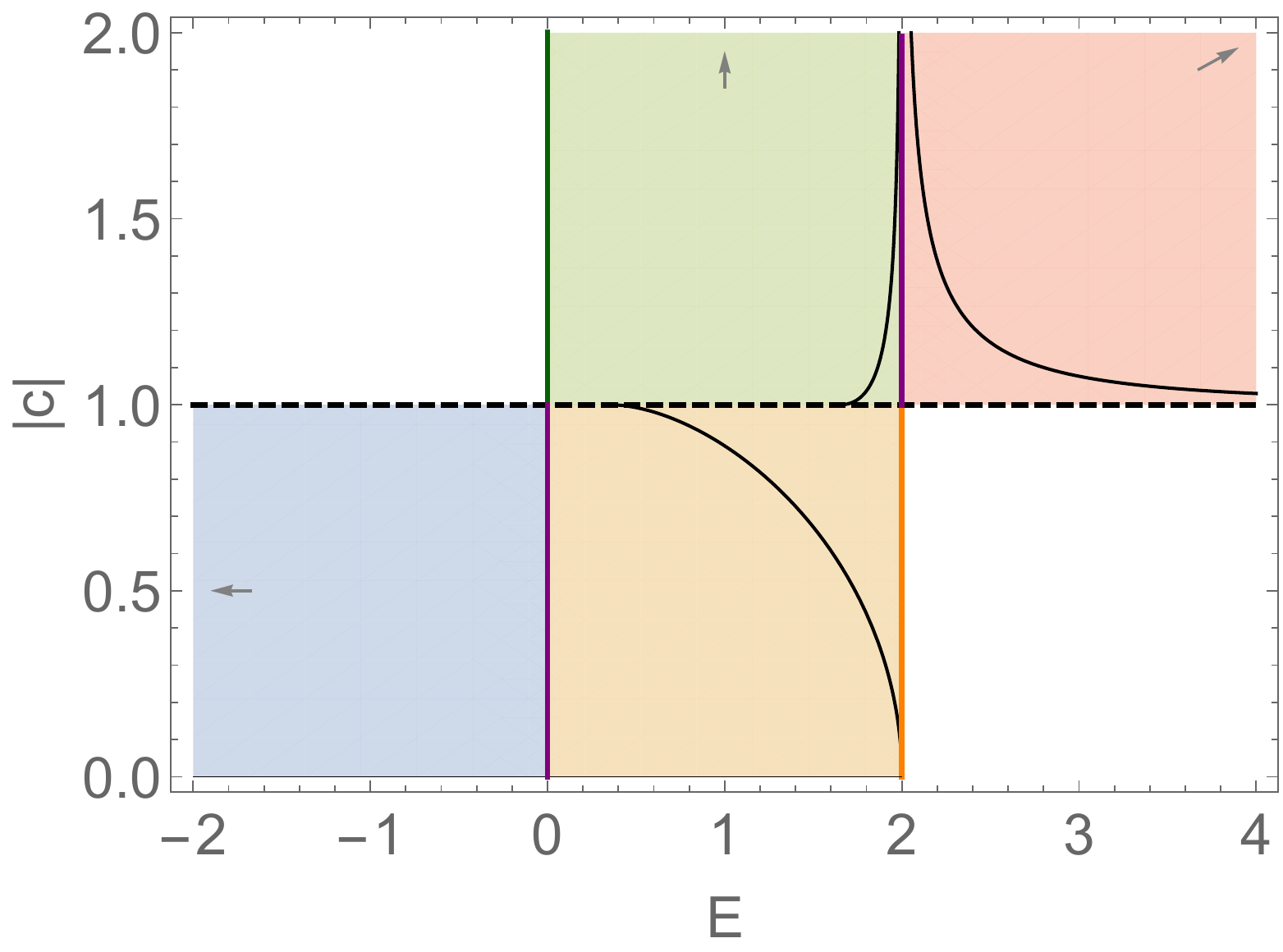} 
\caption{Parameter space with regions corresponding to different qualitative behavior in the linear stability spectrum separated by black curves. Colors correspond to solutions in Figure~\ref{SupSub}. Blue: subluminal rotational ($0\le |c| \, <1,\,E<0$), orange: subluminal librational ($0\le |c| \, <1,\,0<E \le 2$), green: superluminal librational ($|c| \, >1,\,0\le E< 2$), red: superluminal rotational ($0\le|c| \, >1,\,E>2$).}
\label{SGregion-split}
\end{figure}

\section{Qualitatively different parts of the spectrum} \label{regions}

Up to this point we have discussed only the subset of $\sigma_{\mathcal{L}}$ that is on the imaginary axis. In this section we discuss the rest of the spectrum. Except in the subluminal, rotational case, a part of $\sigma_{\mathcal{L}}$ is in the right-half plane (corresponding to unstable modes). For each of the other three regions we split parameter space into two subregions where $\overline{\sigma_{\mathcal{L}}\setminus i \mathbb{R}}$ is qualitatively different. Here $\overline{\sigma_{\mathcal{L}}\setminus i \mathbb{R}}$ is the closure of $\sigma_{\mathcal{L}}$ not on the imaginary axis.

We refer to Figure~\ref{SGregion-split}, which shows $(E,c)$ parameter space with curves that split it into subregions where $\overline{\sigma_{\mathcal{L}}\setminus i \mathbb{R}}$ is qualitatively different. The exact curves splitting up the regions, and their derivations, are given below. In Figure~\ref{spectrumcases} we show representative plots of $\sigma_{\mathcal{L}}$ for all qualitatively different spectra, and in Figure~\ref{laxspectrumcases} we show the corresponding $\sigma_L$ spectrum. 

\begin{itemize}
\item The spectral stability of subluminal rotational solutions is well known \cite{JMMP,JMMP2} and $\sigma_{\mathcal{L}}\subset i \mathbb{R}$. A representative plot of $\sigma_{\mathcal{L}}$ is seen in Figure~\ref{spectrumcases}(e).  
\item For the subluminal librational solutions, $\overline{\sigma_{\mathcal{L}}\setminus i \mathbb{R}}$ consists of either a double-covered infinity symbol, see Figure~\ref{spectrumcases}(f), or a double-covered figure 8 inset inside a double-covered ellipse-like curve, see Figure~\ref{spectrumcases}(g). The boundary between these regions is given explicitly below and a representative plot of $\sigma_{\mathcal{L}}$ on this boundary is seen in Figure~\ref{spectrumspecialcases}(3a).
\item For the superluminal librational solutions, $\overline{\sigma_{\mathcal{L}}\setminus i \mathbb{R}}$ consists of either a double-covered figure 8, see Figure~\ref{spectrumcases}(a), or a double-covered infinity symbol inset inside a double-covered ellipse-like curve, see Figure~\ref{spectrumcases}(b). The boundary between these regions is given below and a representative plot of $\sigma_{\mathcal{L}}$ on this boundary is seen in Figure~\ref{spectrumspecialcases}(1a). 
\item For the superluminal rotational solutions, $\overline{\sigma_{\mathcal{L}}\setminus i \mathbb{R}}$ consists of either a double-covered ellipse-like curve surrounding the origin, see Figure~\ref{spectrumcases}(c), or a double-covered ellipse-like curve in the upper- and lower-half plane, see Figure~\ref{spectrumcases}(d). The boundary between these regions is given explicitly below and a representative plot of $\sigma_{\mathcal{L}}$ on this boundary is seen in Figure~\ref{spectrumspecialcases}(2a). 
\end{itemize}

For all these cases, much can be proven and quantified explicitly, {\em i.e.}, not in terms of special functions. Specifically, we calculate explicit expressions for $\sigma_{\mathcal{L}} \cap i \mathbb{R}$ and in the librational case we find explicit expressions for the tangents to $\sigma_{\mathcal{L}}$ around the origin. In fact, we are able to approximate the spectrum at the origin and around all points $\sigma_{\mathcal{L}} \cap i \mathbb{R}$ using a Taylor series to arbitrary order. These series give good approximations to the greatest real part of $\sigma_{\mathcal{L}}$ using only a few terms. They are not given in this paper, but follow from the same procedure as outlined in \cite{DS17}.

A method for determining $\sigma_L$ is to take known points satisfying (\ref{intcond7}) and to follow the tangent vector field (\ref{tangentvectorfield}) from those points. We apply this technique from $\zeta\in \mathbb{R}$ which we know to satisfy (\ref{intcond7}) from Theorem~\ref{realzetathm} as well as from the points $\zeta$ satisfying $\Omega^2(\zeta) = 0$.

\subsection{Subluminal librational solutions}\label{sublibsubsection}
The roots of $\Omega^2(\zeta)=0$ are given by
\beq \zeta_c = \left\{\frac{\sqrt{1-c^2}\sqrt{E}}{2\sqrt{2}(c+1)}\pm \frac{\sqrt{1-c^2}\sqrt{2-e}}{2\sqrt{2}(c+1)} i,-\frac{\sqrt{1-c^2}\sqrt{E}}{2\sqrt{2}(c+1)}\pm \frac{\sqrt{1-c^2}\sqrt{2-e}}{2\sqrt{2}(c+1)} i \right\}, \eeq
seen as red crosses in Figure~\ref{laxspectrumcases}(f,g). 
For convenience, we label these four roots $\zeta_1,\,\zeta_2,\,\zeta_3,\,\zeta_4,$ where the subscript corresponds to the quadrant on the real and imaginary plane the root is in. 
In this case, (\ref{derivintcond}) is
\beq \frac{\text{d}I(\zeta)}{\text{d}\zeta} = \sqrt{1-c^2} \frac{16 \zeta^2 \mathcal{E}(k)+\left(c-1-8 \zeta^2-16(c+1)\zeta^4\right)\mathcal{K}(k)}{32 \zeta^3 \Omega(z)}. \label{derivintcondsublib}\eeq

\begin{figure}
\centering
\begin{tabular}{cccc}
  \includegraphics[height=36mm]{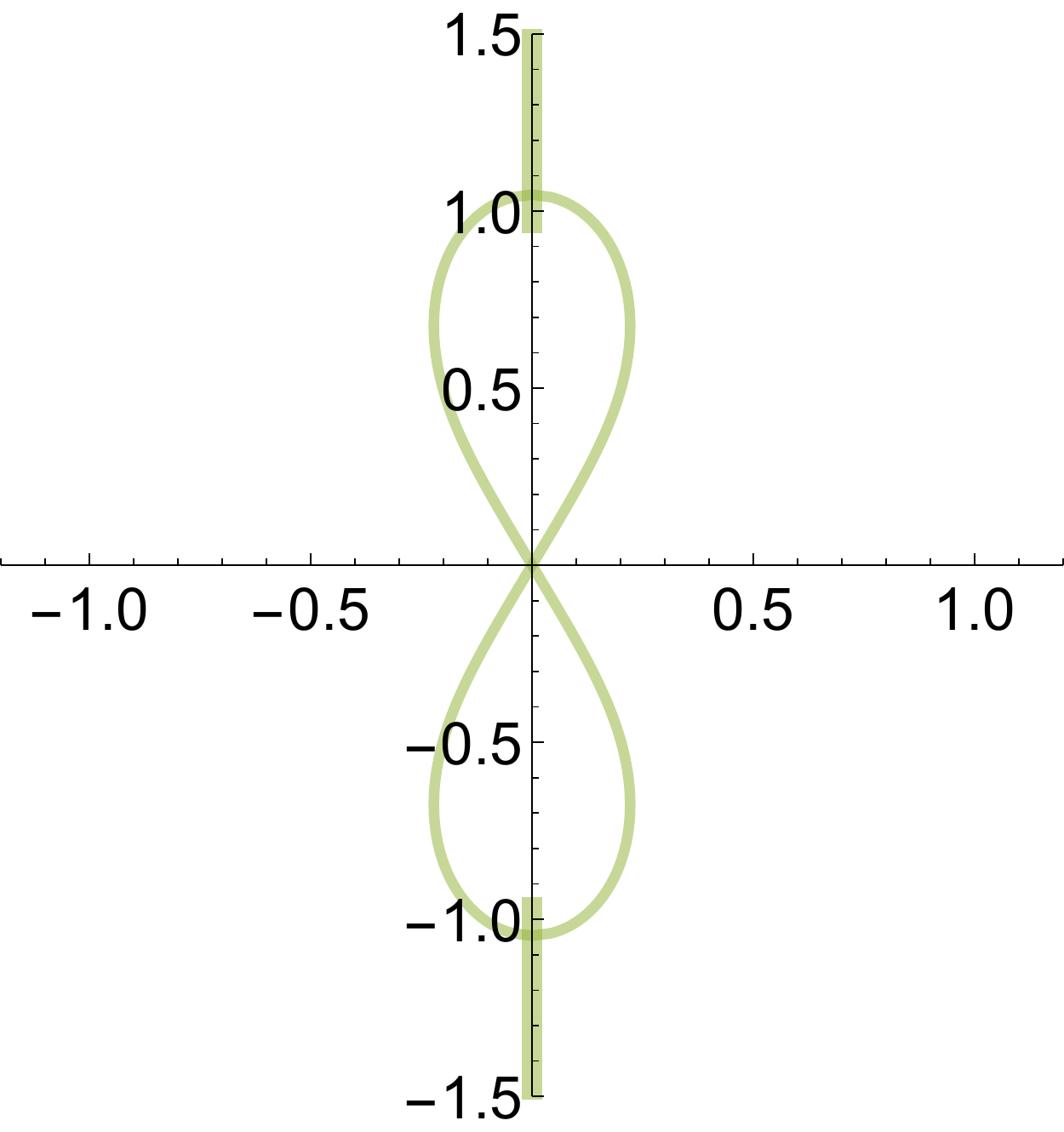} & \includegraphics[height=36mm]{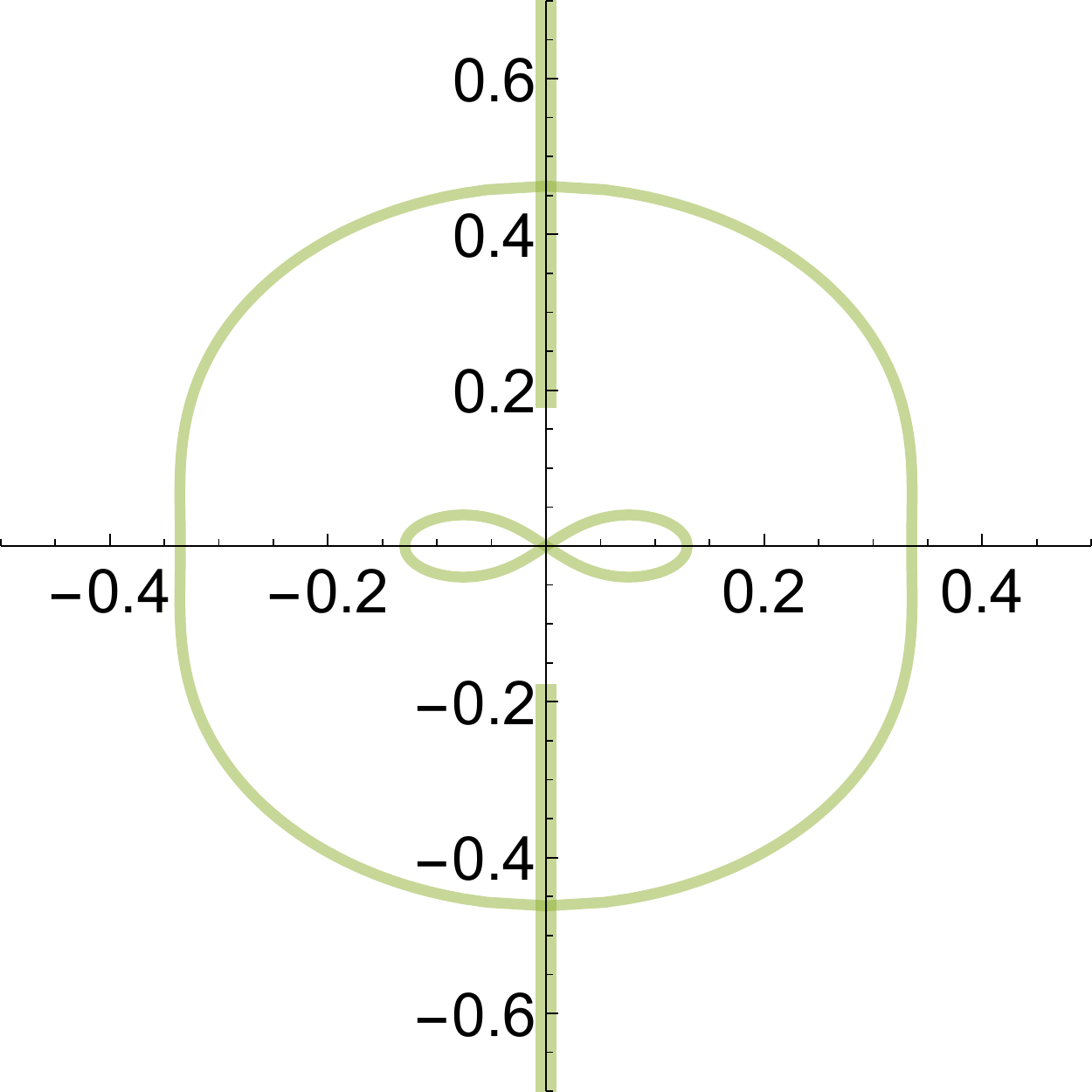} & \includegraphics[height=36mm]{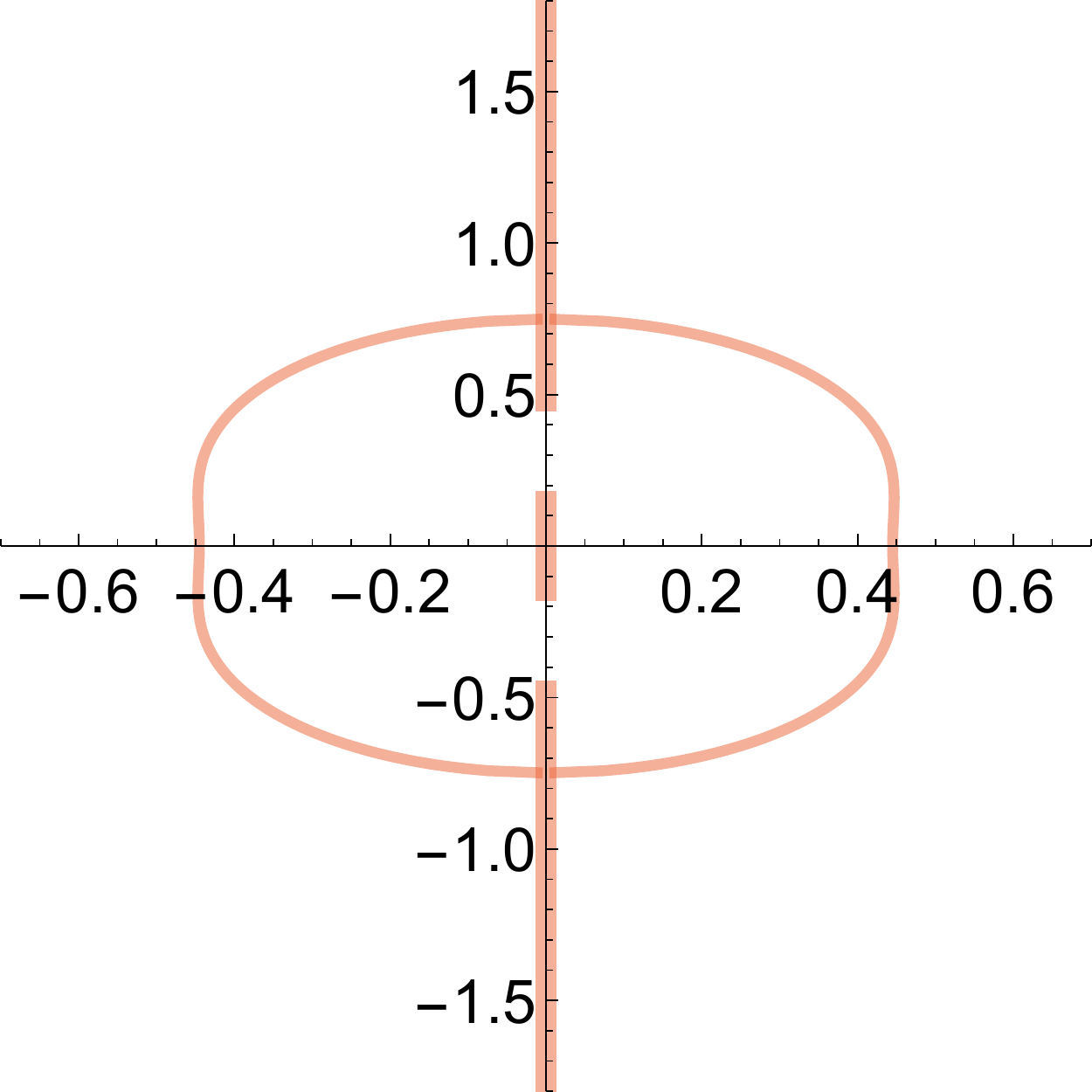} & \includegraphics[height=36mm]{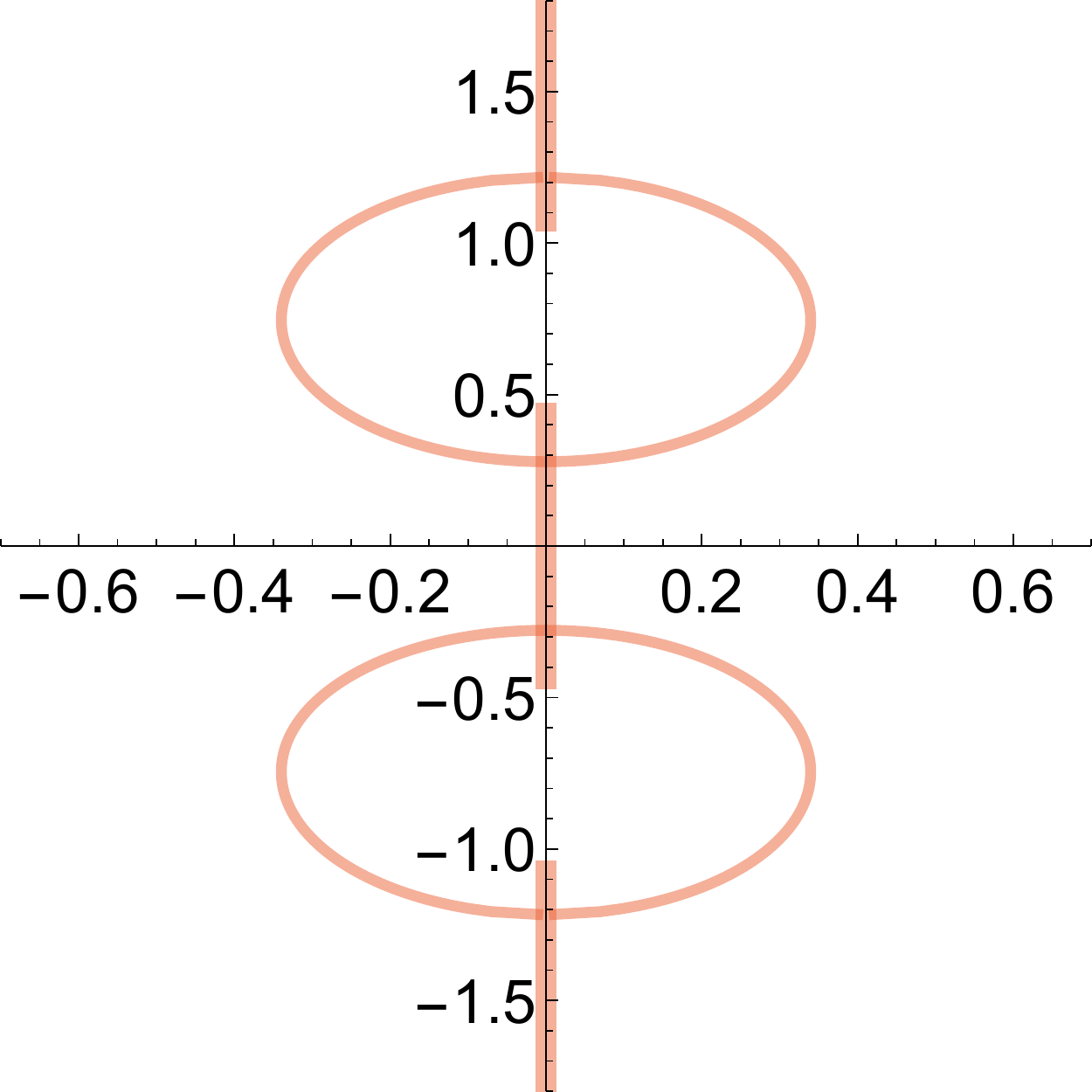} \\
(a) & (b) & (c) & (d) 
\end{tabular}
\begin{tabular}{ccc}
  \includegraphics[height=36mm]{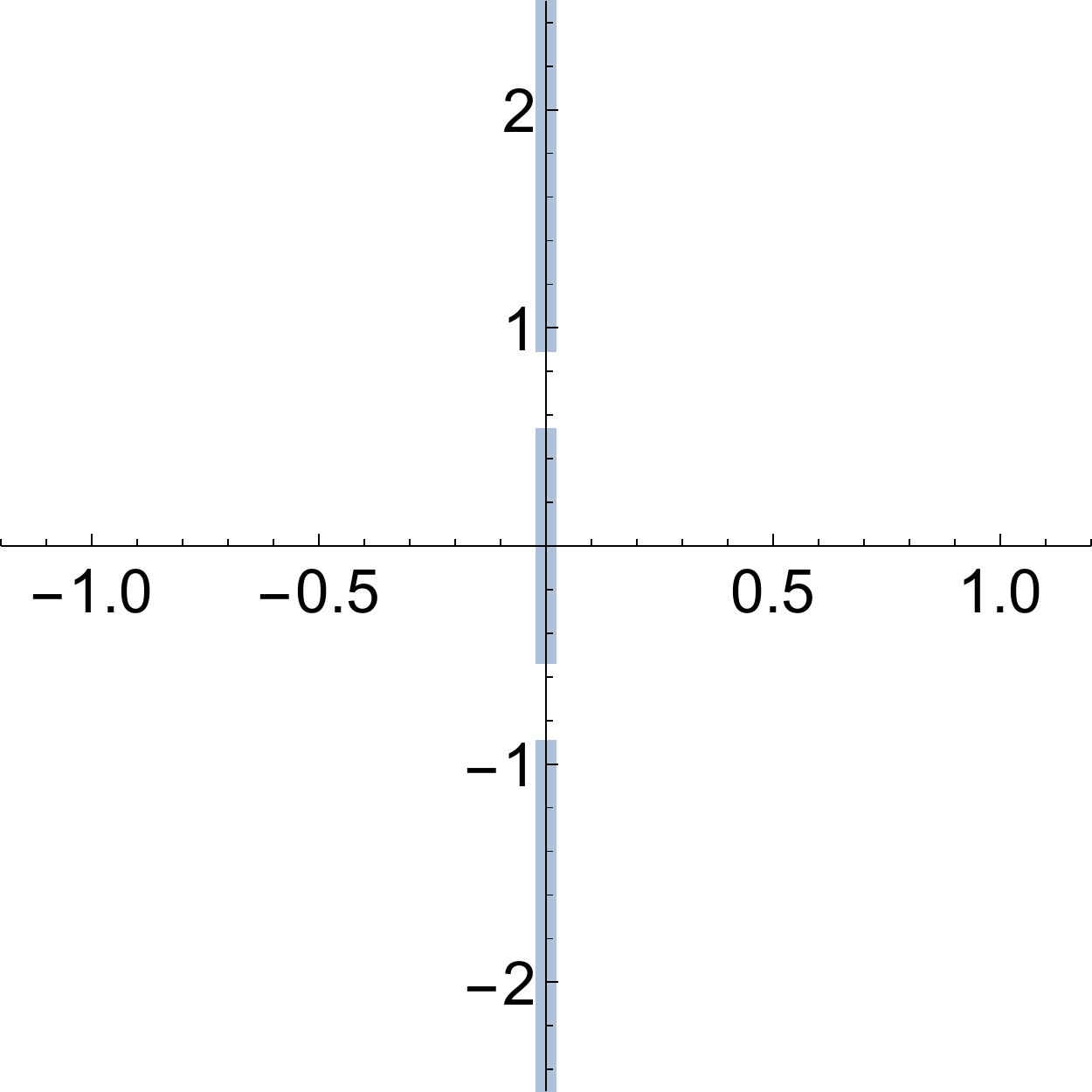} & \includegraphics[height=36mm]{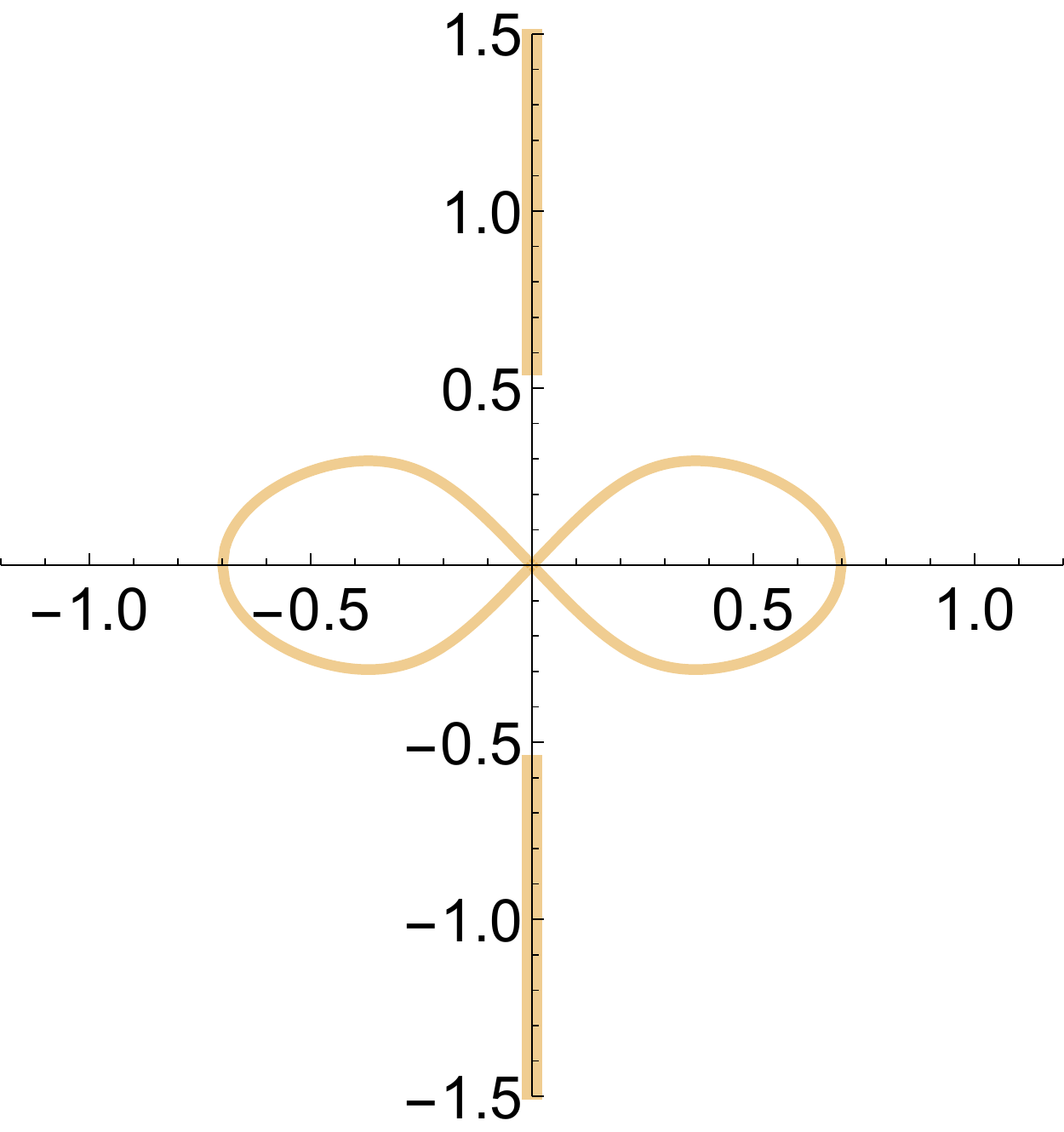} & \includegraphics[height=36mm]{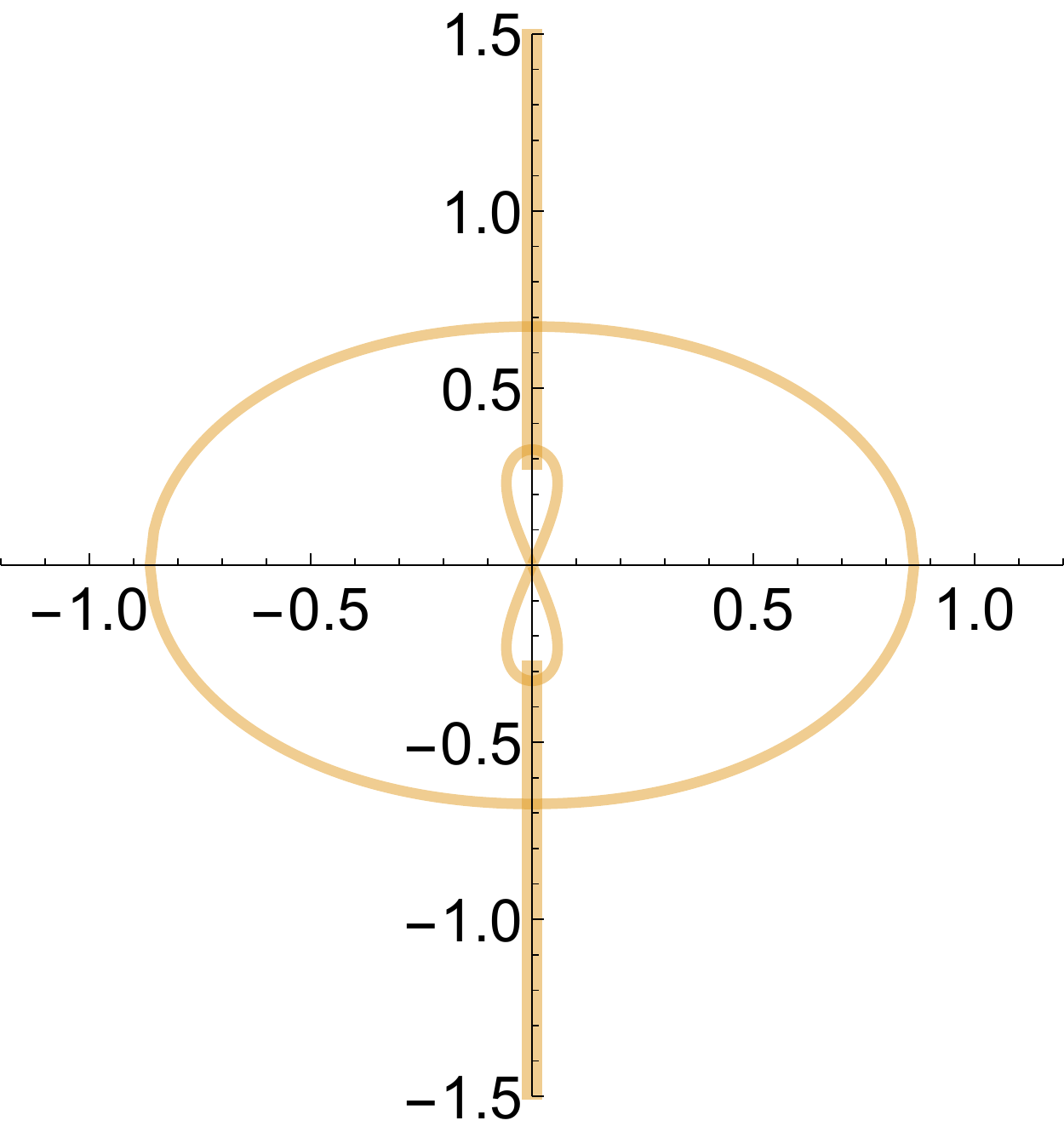}  \\
(e) & (f) & (g) 
\end{tabular}
\vspace{-2mm}
\caption{The stability spectrum for superluminal (a-d), subluminal (e-g), librational (a,b,f,g) and rotational (c,d,e) waves. (a) $c=1.5,\,E=1.5;$ (b) $c=1.02,\,E=1.8;$ (c) $c=1.1,\,E=2.2;$ (d) $c=1.4,\,E=2.4;$ (e) $c=0.6,\,E=-0.75;$ (f) $c=0.6,\,E=1.0;$ (g) $c=0.8,\,E=1.5;$   Colors correspond to Figure~\ref{SGregion}, thickness of lines corresponds to double or quadruple covering of spectrum.}
\label{spectrumcases}
\end{figure} 

Examining (\ref{tangentvectorfield}) for $\zeta\in \mathbb{R}$, for a vertical tangent in $\sigma_L$ to occur, we need the numerator of (\ref{derivintcondsublib}) to be zero. Using the discriminant of $16 \zeta^2 \mathcal{E}(k)+\left(c-1-8 \zeta^2-16(c+1)\zeta^4\right)\mathcal{K}(k)$ as a function of $\zeta,$ we find the condition
\beq c = \frac{2 \sqrt{-\mathcal{E}^2(k)+\mathcal{E}(k)\mathcal{K}(k)}}{\mathcal{K}(k)}, \label{sublibcond} \eeq
for a vertical tangent to occur on the real axis. This condition is plotted as the black curve in the subluminal rotational region of Figure~\ref{SGregion-split}, and defines the split between qualitatively different spectra. Representative spectral plots for $E$ and $c$ on this boundary are seen in Figure~\ref{spectrumspecialcases}(3). For solutions satisfying (\ref{sublibcond}) we have the following 
\beq \zeta_{t1} =\pm \frac{1}{2} \sqrt{\frac{1}{(1+c)\mathcal{K}(k)} \left(2\mathcal{E}(k)-\mathcal{K}(k)+\sqrt{4\mathcal{E}^2(k)-4\mathcal{E}(k)\mathcal{K}(k)+c^2\mathcal{K}^2(k)}\right)}, \eeq

\begin{figure}
\centering
\begin{tabular}{cccc}
  \includegraphics[height=36mm]{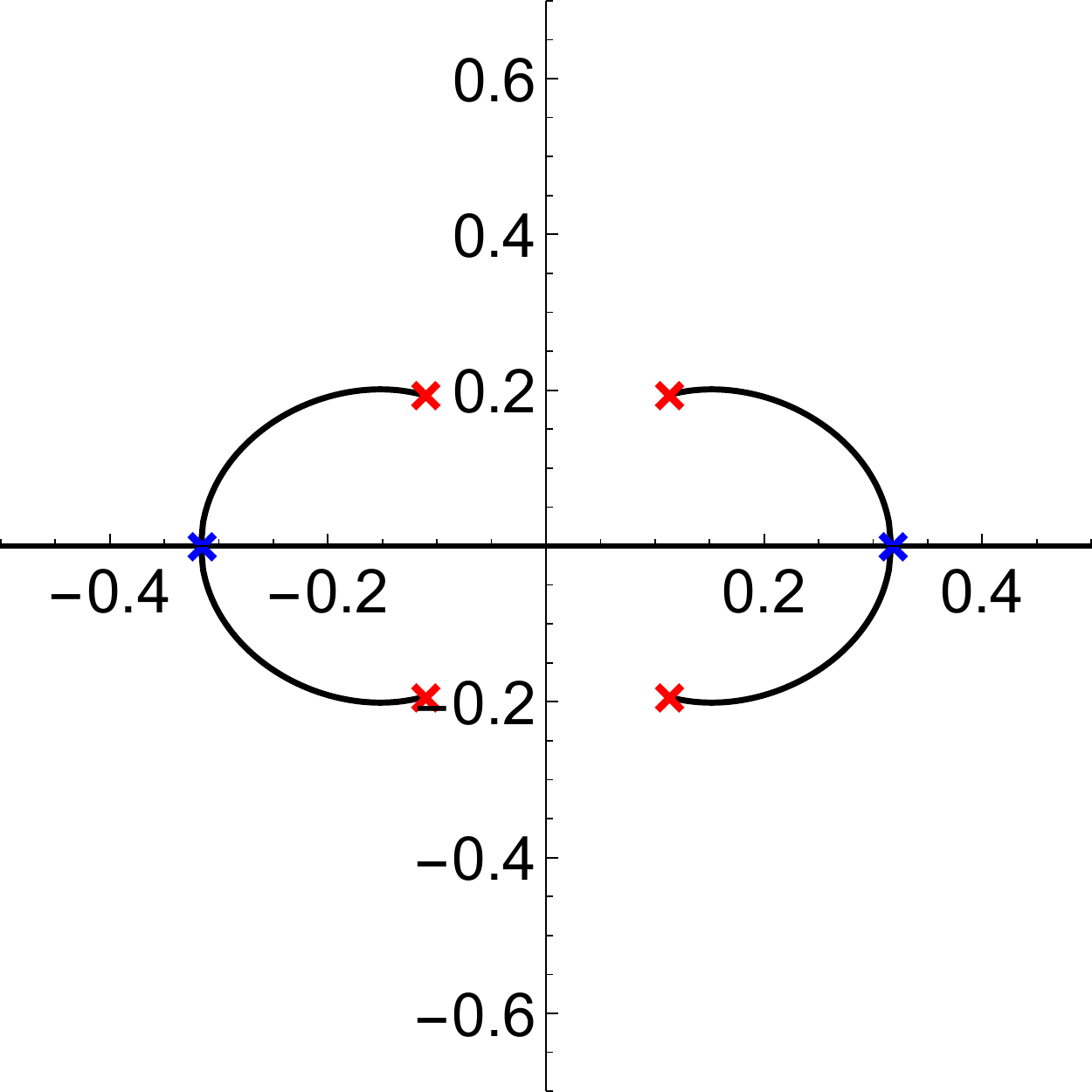} & \includegraphics[height=36mm]{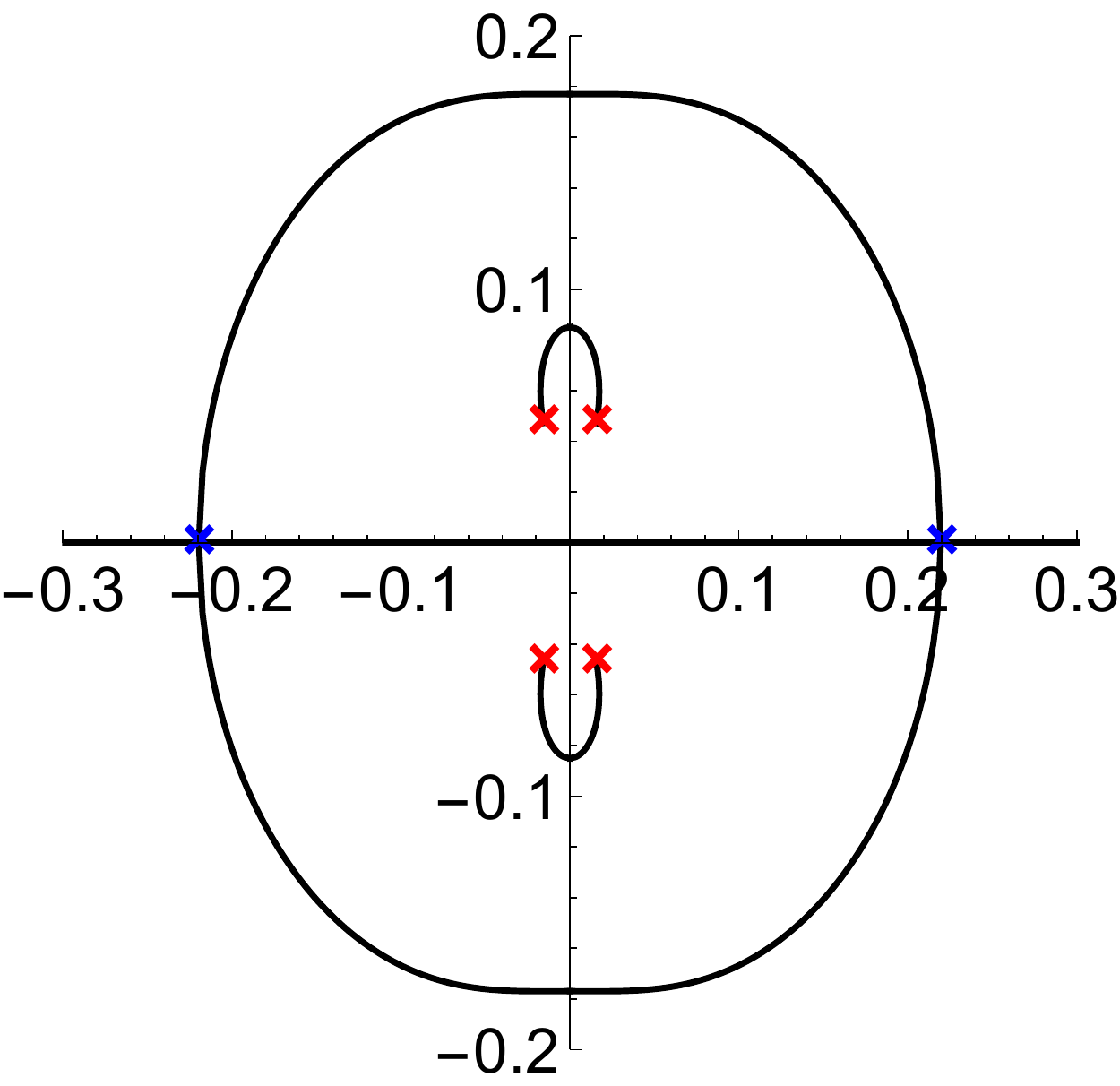} & \includegraphics[height=36mm]{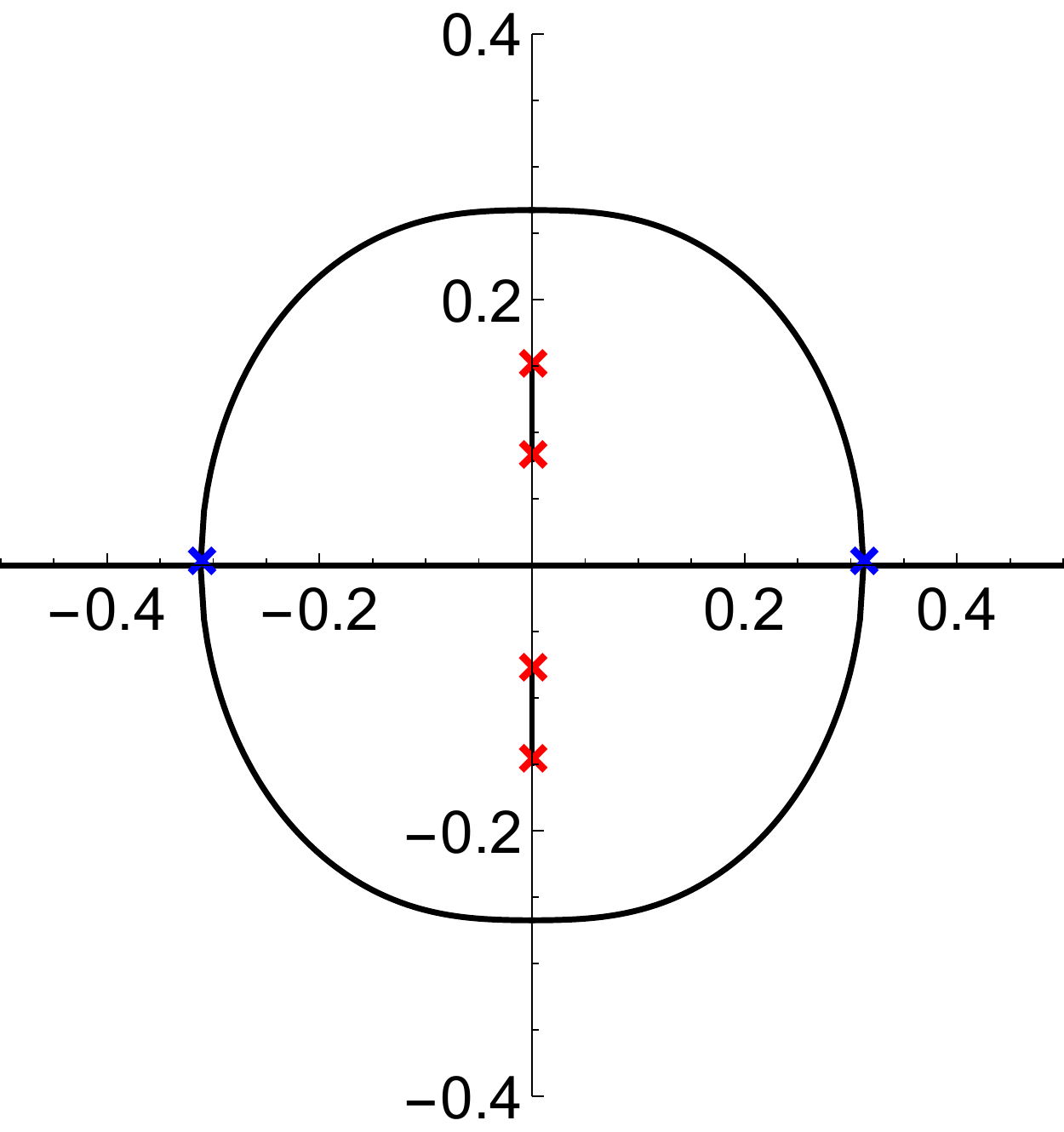} & \includegraphics[height=36mm]{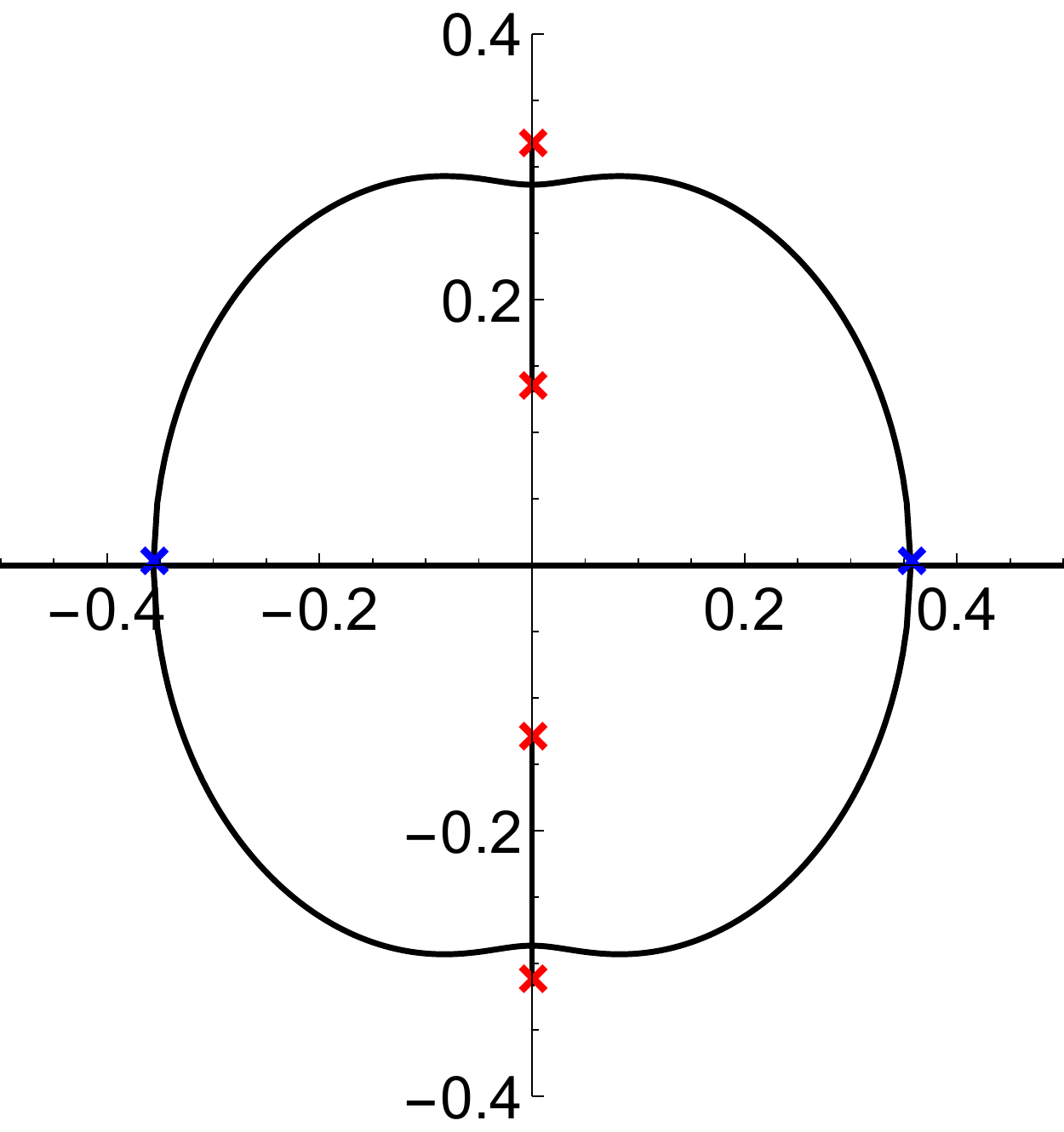} \\
(a) & (b) & (c) & (d) 
\end{tabular}
\begin{tabular}{ccc}
  \includegraphics[height=36mm]{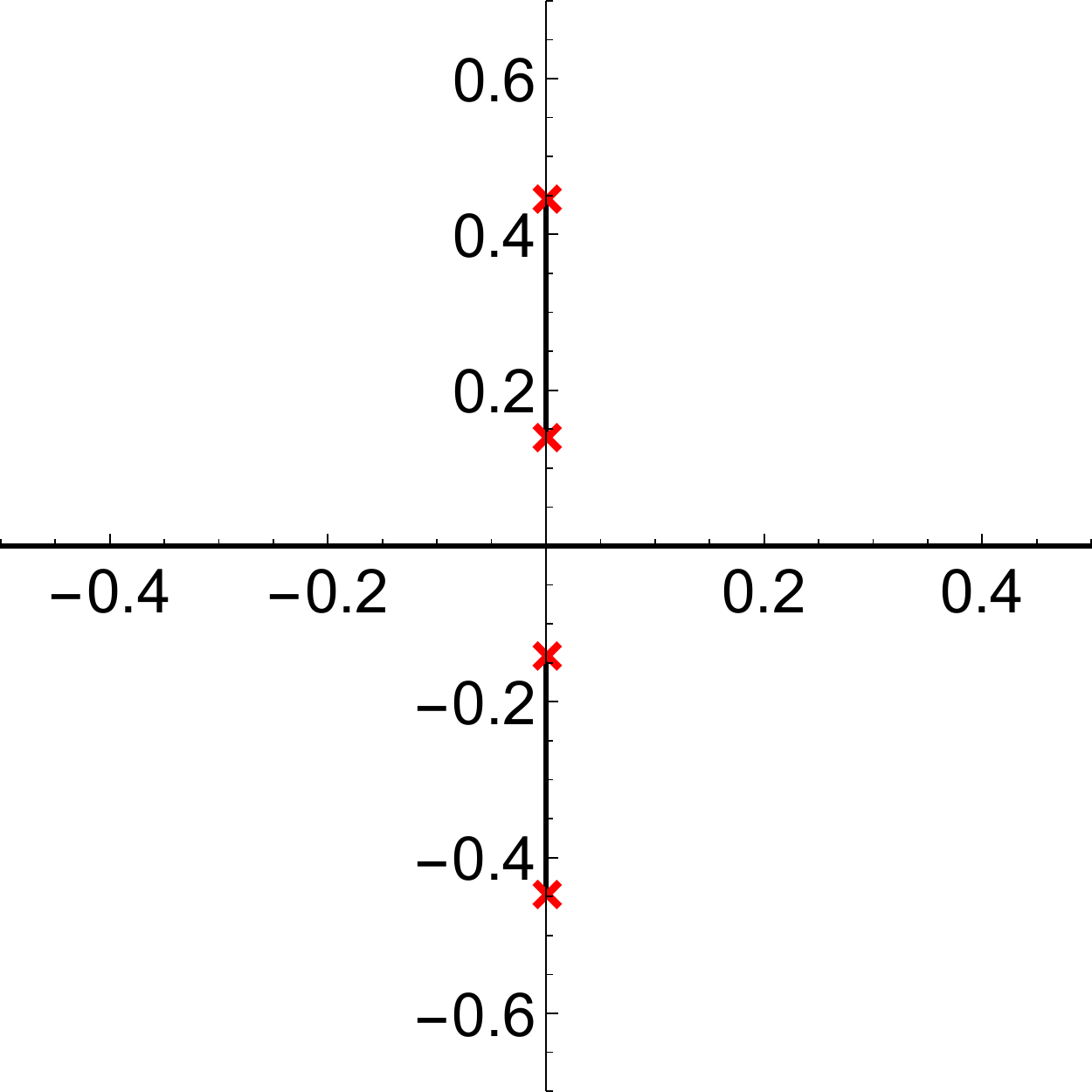} & \includegraphics[height=36mm]{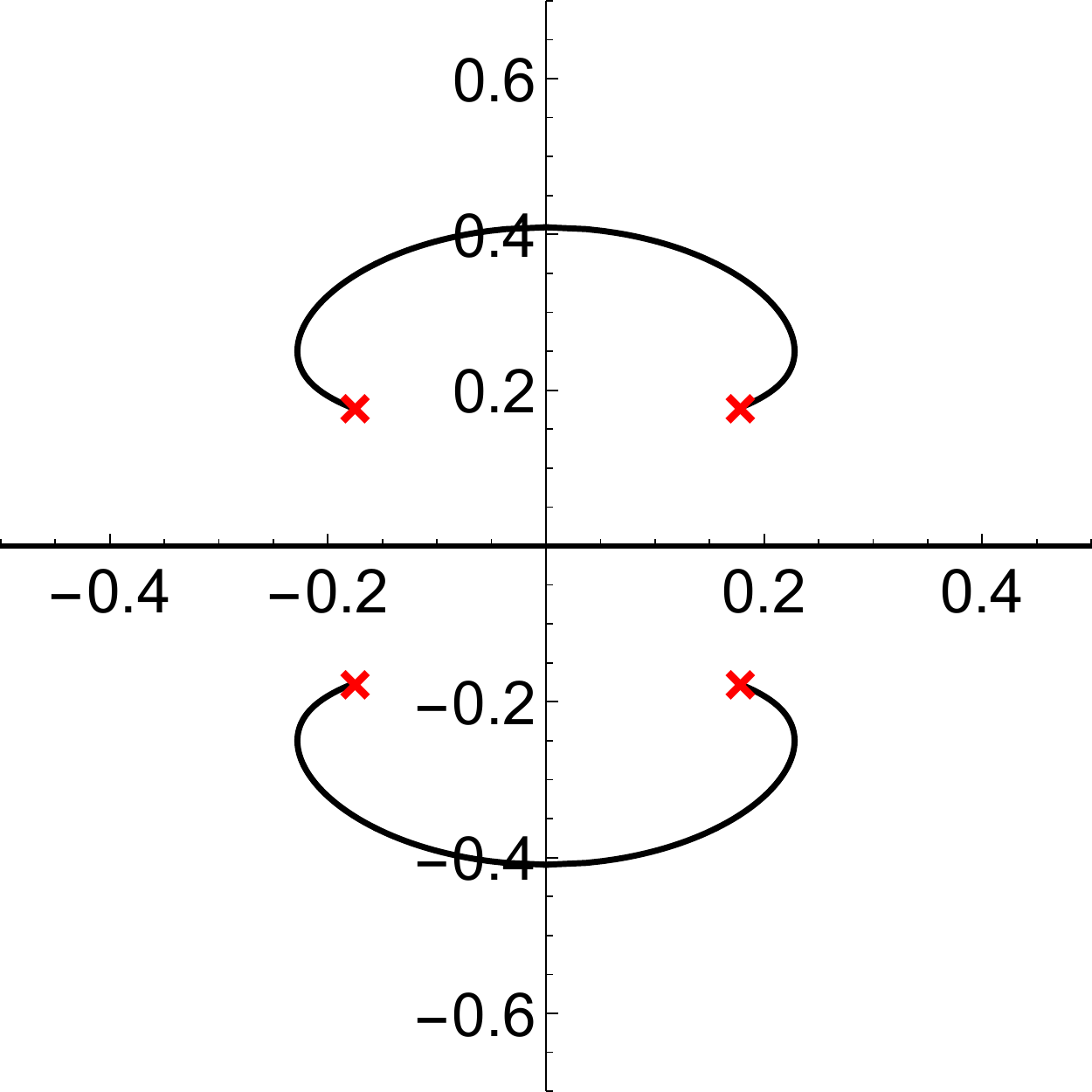} & \includegraphics[height=36mm]{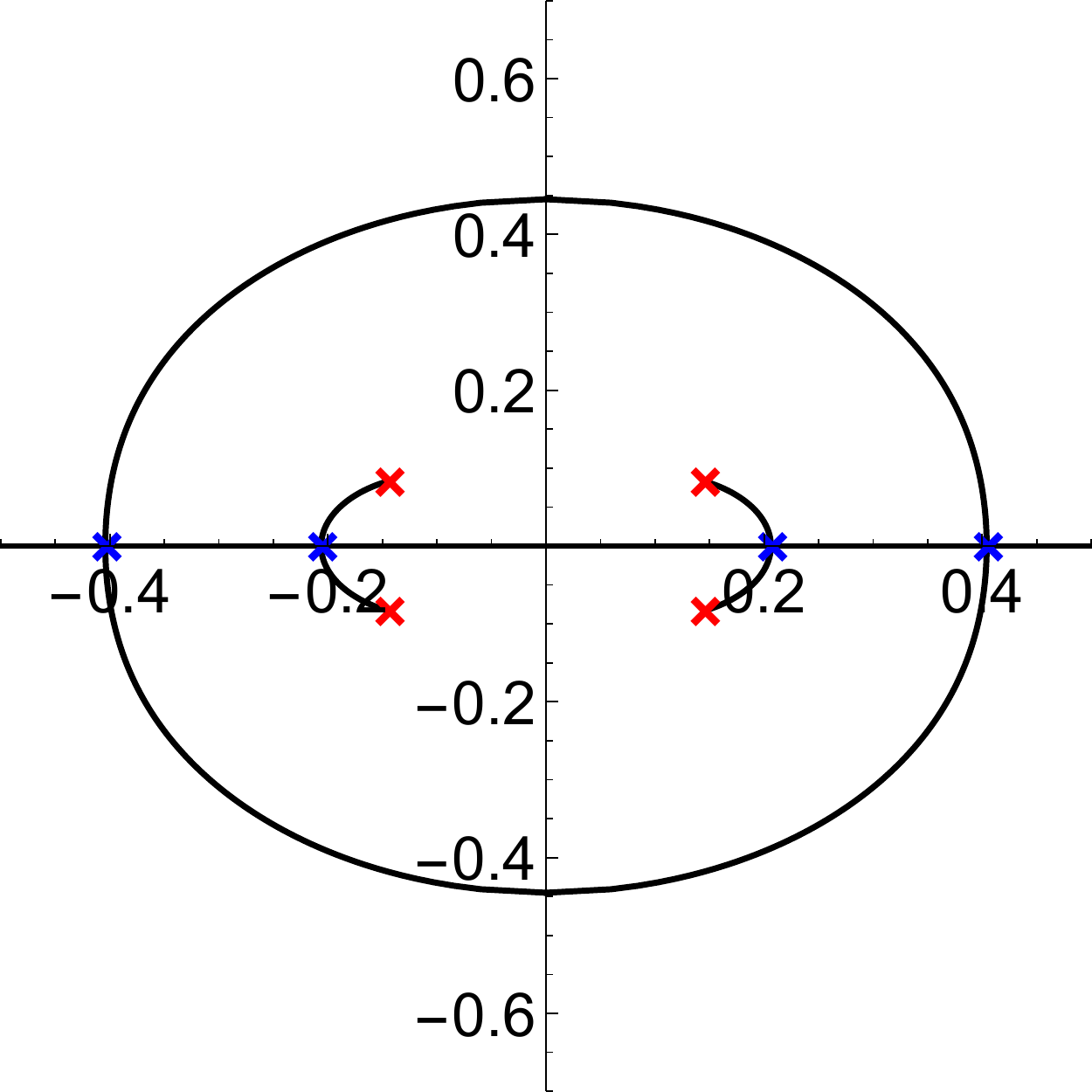}  \\
(e) & (f) & (g) 
\end{tabular}
\vspace{-2mm}
\caption{The Lax spectrum (black curves) for superluminal (a-d), subluminal (e-g), librational (a,b,f,g) and rotational (c,d,e) waves. (a) $c=1.5,\,E=1.5;$ (b) $c=1.02,\,E=1.8;$ (c) $c=1.1,\,E=2.2;$ (d) $c=1.4,\,E=2.4;$ (e) $c=0.6,\,E=-0.75;$ (f) $c=0.6,\,E=1.0;$ (g) $c=0.8,\,E=1.5.$  Red crosses signify values of $\zeta$ for which $\Omega^2(\zeta) = 0$. Blue crosses signify values of $\zeta\in \mathbb{R}$ for which $\sigma_L$ has a vertical tangent.}
\label{laxspectrumcases}
\end{figure}

and 
\beq \zeta_{t2} =\pm \frac{1}{2} \sqrt{\frac{1}{(1+c)\mathcal{K}(k)} \left(2\mathcal{E}(k)-\mathcal{K}(k)-\sqrt{4\mathcal{E}^2(k)-4\mathcal{E}(k)\mathcal{K}(k)+c^2\mathcal{K}^2(k)}\right)}, \eeq
shown as blue crosses in Figure~\ref{laxspectrumcases}(g). Mapping these points back to $\sigma_{\mathcal{L}}$ these points correspond to the top (or bottom) of the inset figure 8 in Figure~\ref{spectrumcases}(g):
\beq \lambda_{t1} = \pm \sqrt{\frac{E-c^2(E-1)}{2}+\frac{-2\mathcal{E}(k)+c \sqrt{4\mathcal{E}^2(k)-4\mathcal{E}(k)\mathcal{K}(k)+c^2\mathcal{K}^2(k)}}{2 \mathcal{K}(k)}}, \label{lambdasublib1}\eeq
and the ellipse-like curve in Figure~\ref{spectrumcases}(g):
\beq \lambda_{t2} = \pm \sqrt{\frac{E-c^2(E-1)}{2}+\frac{-2\mathcal{E}(k)-c \sqrt{4\mathcal{E}^2(k)-4\mathcal{E}(k)\mathcal{K}(k)+c^2\mathcal{K}^2(k)}}{2 \mathcal{K}(k)}} \label{lambdasublib2}.\eeq

\begin{figure}
\centering
\begin{tabular}{ccc}
  \includegraphics[height=36mm]{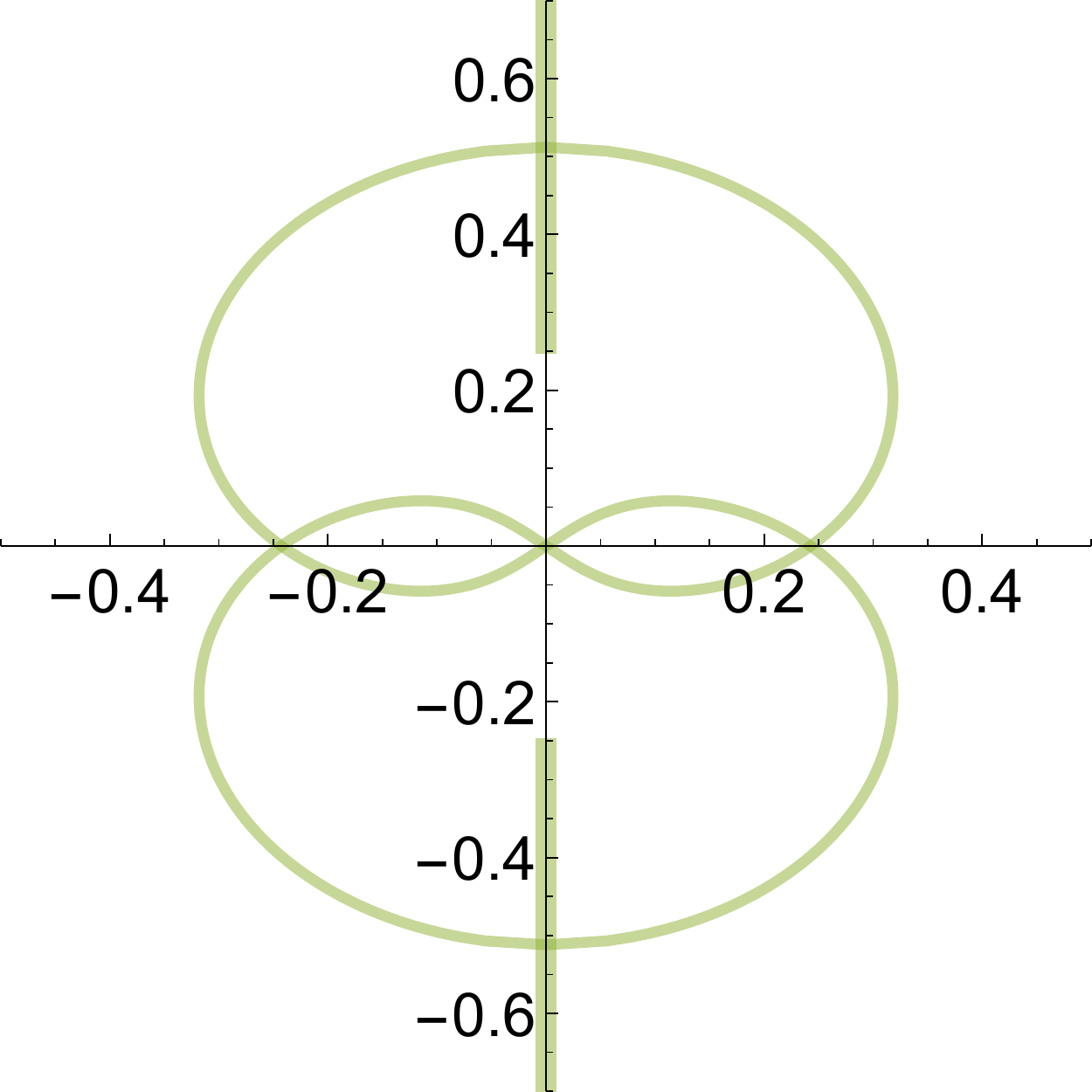} & \includegraphics[height=36mm]{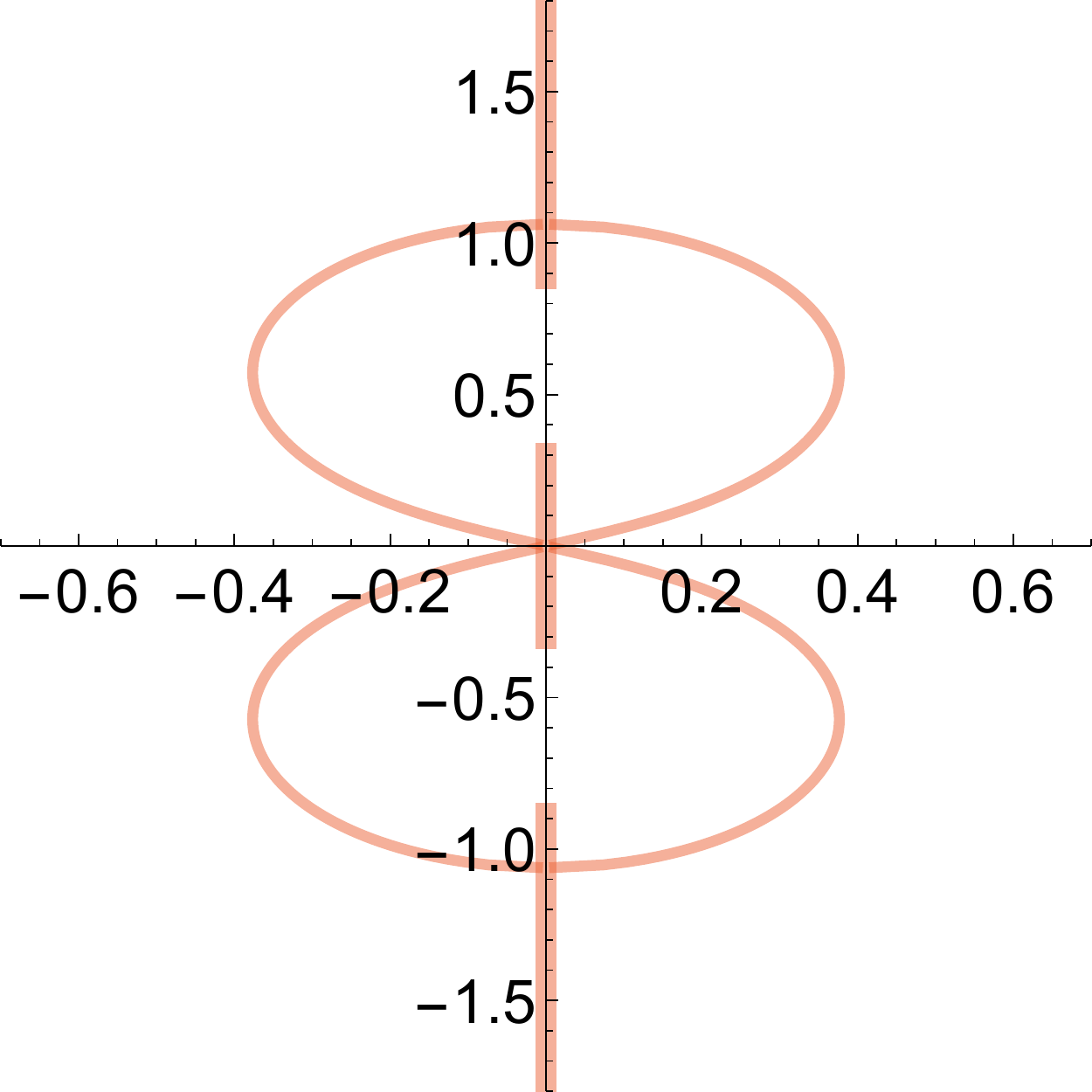} & \includegraphics[height=36mm]{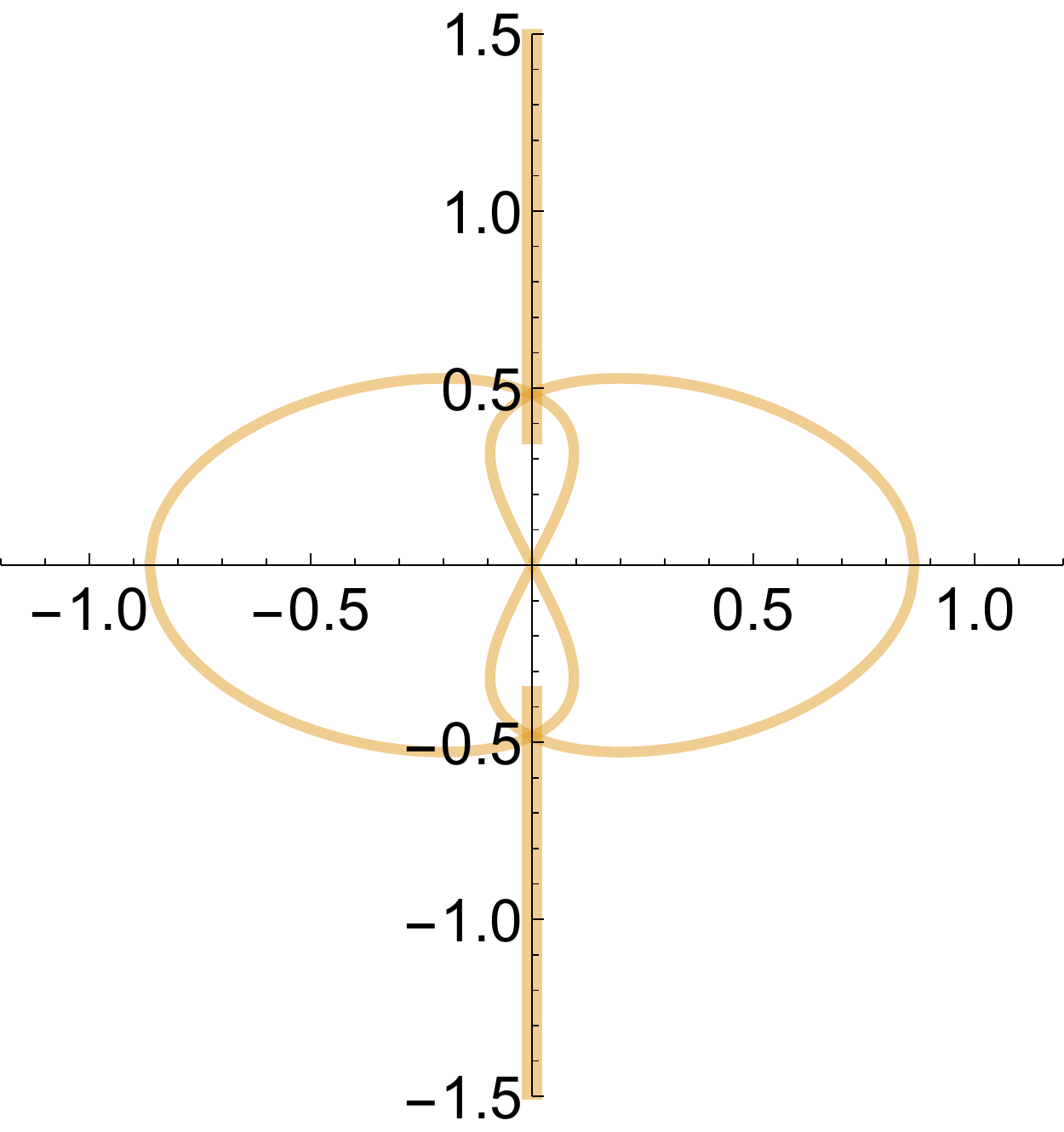} \\
(1a) & (2a) & (3a)
\end{tabular}
\begin{tabular}{ccc}
  \includegraphics[height=36mm]{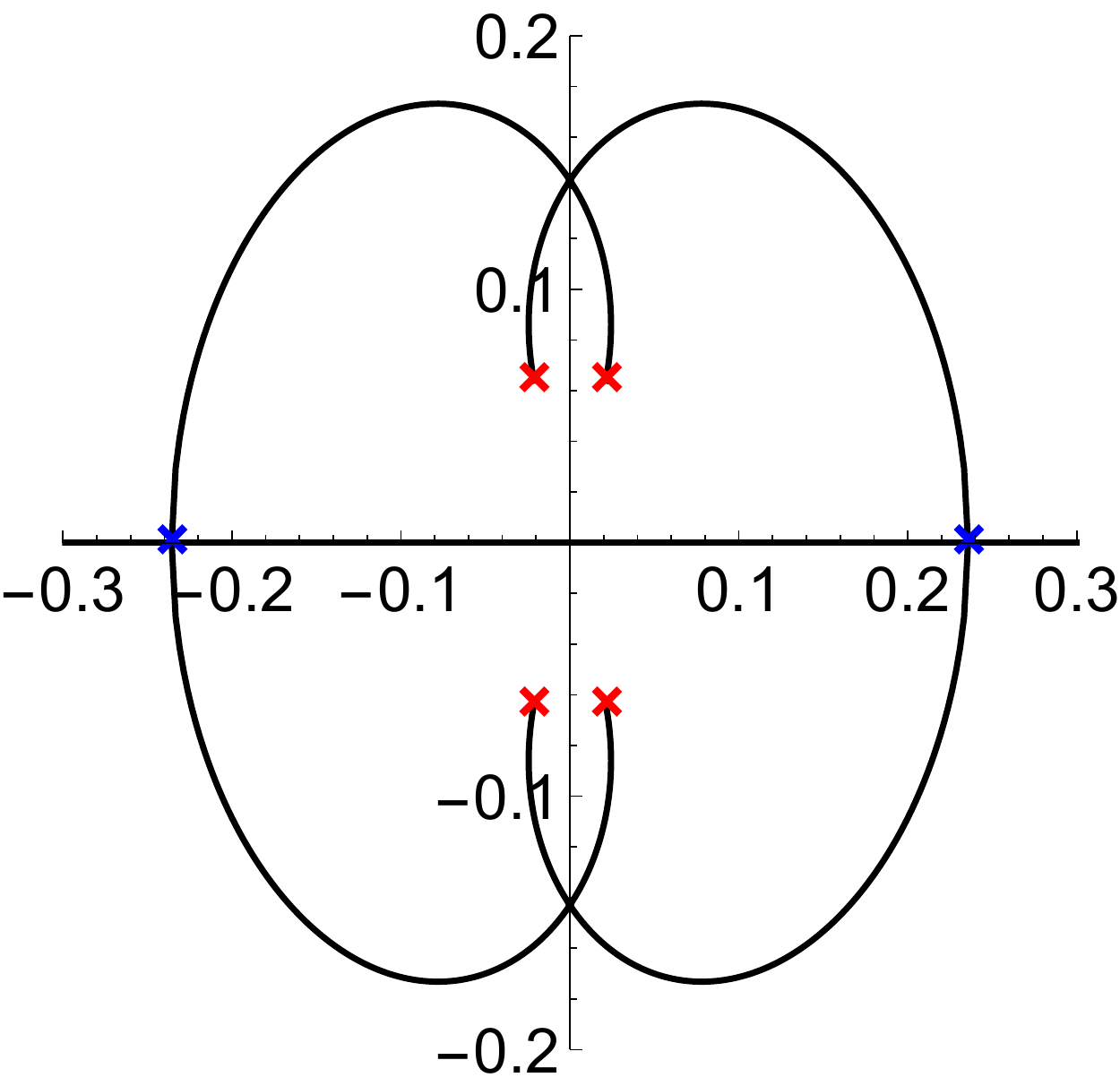} & \includegraphics[height=36mm]{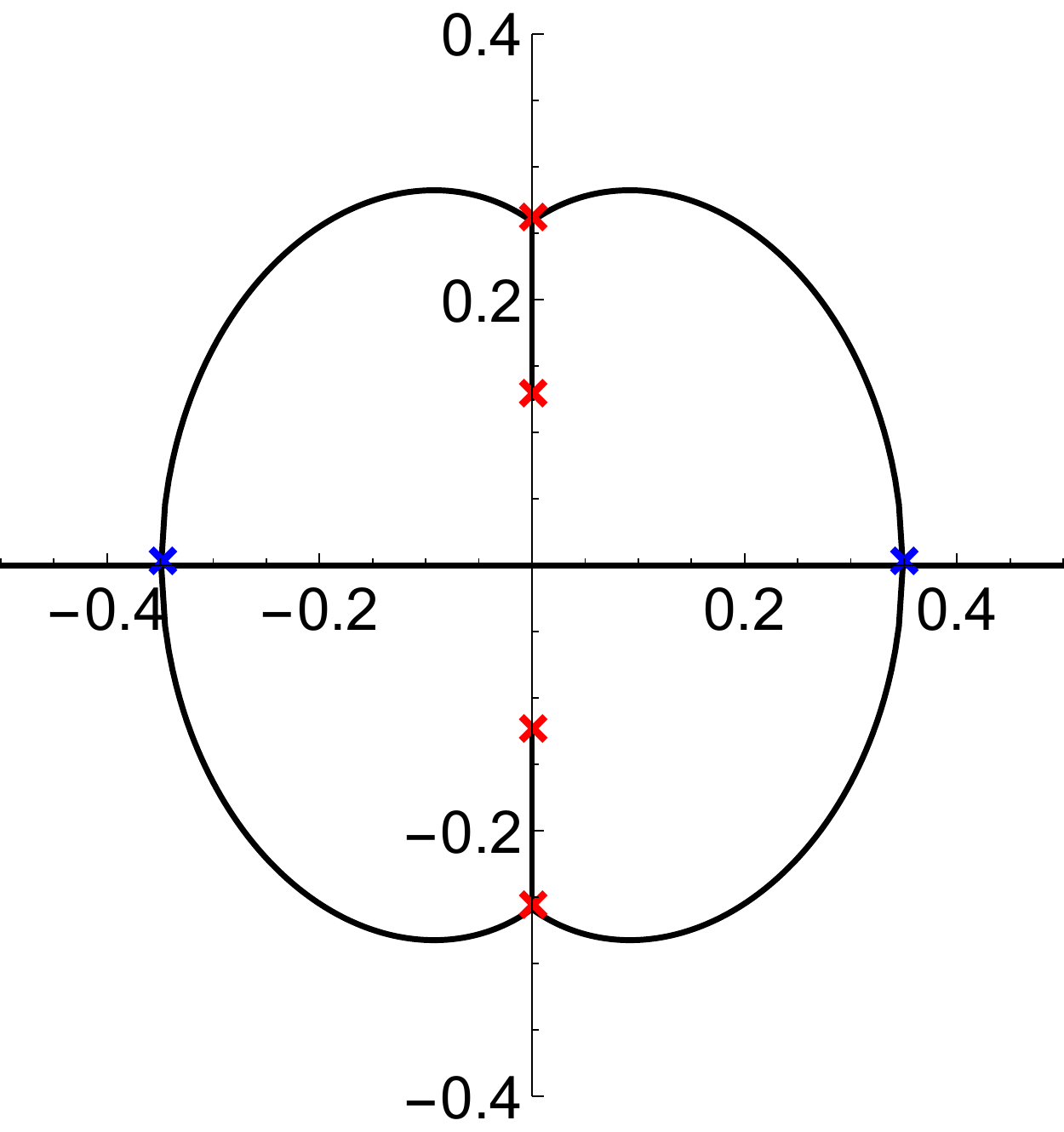} & \includegraphics[height=36mm]{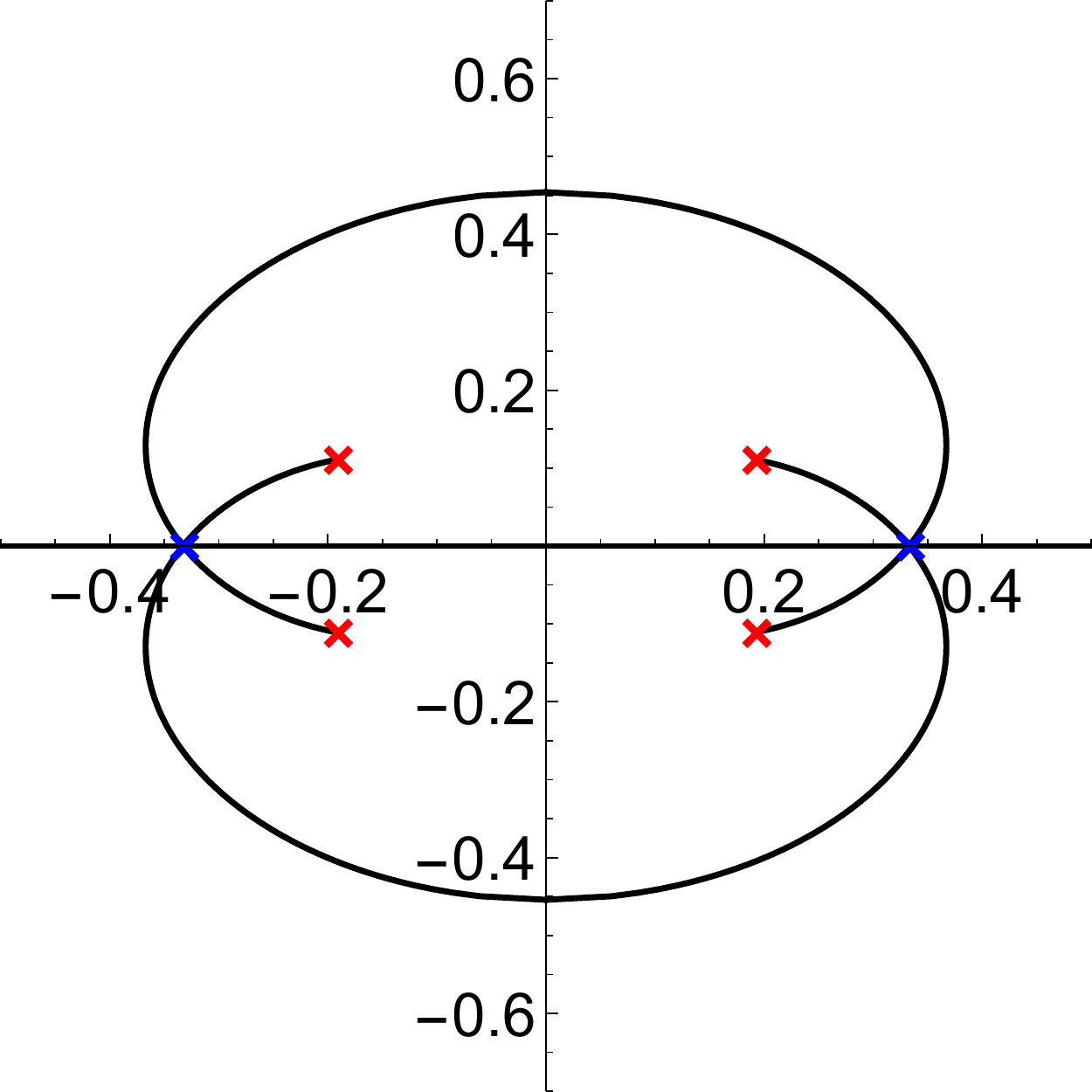}  \\
(1b) & (2b) & (3b) 
\end{tabular}
\vspace{-2mm}
\caption{(1) The stability spectrum for the cases separating subregions and (2) the corresponding Lax spectrum (black curves). Red crosses signify values of $\zeta$ for which $\Omega^2(\zeta) = 0$. Blue crosses signify values of $\zeta\in \mathbb{R}$ for which $\sigma_L$ has a vertical tangent. (a) Superluminal librational: $c=1.03702,\,E=1.8$, (b) superluminal rotational: $c=1.3,\,E=2.27060$ (c) subluminal librational: $c=0.67148,\,E=1.5$.}
\label{spectrumspecialcases}
\end{figure}

Next we examine the slopes of $\sigma_{\mathcal{L}}$ at the origin. Because $\sigma_{\mathcal{L}}= 2 S_\Omega$ it suffices to examine the slopes for the set $S_\Omega$. We let $\Omega = \Omega_r+i \Omega_i,$ and we consider $\zeta_i$ as a function of $\zeta_r$ so that $\Omega\left(\zeta_r,\zeta_i(\zeta_r)\right)$. Applying the chain rule we have that the slope at any point in the set $S_\Omega$ is
\beq \label{derivO} \frac{ \textrm{d}\Omega_i}{\textrm{d} \Omega_r} = \frac{\textrm{d}\Omega_i/\textrm{d}\zeta_r}{\textrm{d}\Omega_r/\textrm{d}\zeta_r} = \frac{ \frac{\textrm{d}\Omega_i}{\textrm{d} \zeta_r}+\frac{\textrm{d} \Omega_i}{\textrm{d} \zeta_i} \frac{ \textrm{d} \zeta_i}{\textrm{d} \zeta_r}}{\frac{\textrm{d} \Omega_r}{\textrm{d}\zeta_r}+\frac{\textrm{d} \Omega_r}{\textrm{d} \zeta_i} \frac{\textrm{d} \zeta_i}{\textrm{d} \zeta_r}},\eeq
where
\beq \label{derivZ} \frac{\textrm{d} \zeta_i}{\textrm{d} \zeta_r} = - \frac{\textrm{d}\text{Re}(I)/\textrm{d}\zeta_r}{\textrm{d}\text{Re}(I)/\textrm{d} \zeta_i}. \eeq
We examine (\ref{derivO}) near where $\Omega=0$ and $\zeta=\zeta_c$.
The slopes at the origin are 
\beq \frac{ \textrm{d}\Omega_i}{\textrm{d} \Omega_r} = \pm \frac{c \sqrt{E(2-E)} \mathcal{K}(k)}{-2\mathcal{E}(k)+E \mathcal{K}(k)}. \eeq
Further application of the chain rule can yield expressions for derivatives around the origin of any order, and the same technique can be applied around (\ref{lambdasublib1}) and (\ref{lambdasublib2}). In doing this we can obtain Taylor series approximations of $\sigma_{\mathcal{L}}$ to any order.

\subsection{Superluminal librational solutions}
The roots of $\Omega^2(\zeta)=0$ are given by
\beq \zeta_c = \left\{\frac{\sqrt{c^2-1}\sqrt{2-E}}{2\sqrt{2}(c+1)}\pm \frac{\sqrt{c^2-1}\sqrt{E}}{2\sqrt{2}(c+1)} i,-\frac{\sqrt{c^2-1}\sqrt{2-E}}{2\sqrt{2}(c+1)}\pm \frac{\sqrt{c^2-1}\sqrt{E}}{2\sqrt{2}(c+1)} i \right\}, \eeq
seen as red crosses in Figure~\ref{laxspectrumcases}(f,g). 
As in Section~\ref{sublibsubsection}, we label these four roots $\zeta_1,\,\zeta_2,\,\zeta_3,\,\zeta_4,$ where the subscript corresponds to the quadrant on the real and imaginary plane the root is in. 
In this case, (\ref{derivintcond}) is
\beq \frac{\text{d}I(\zeta)}{\text{d}\zeta} = \sqrt{c^2-1} \frac{-16 \zeta^2 \mathcal{E}(k)+\left(c-1+8 \zeta^2-16(c+1)\zeta^4\right)\mathcal{K}(k)}{32 \zeta^3 \Omega(z)}. \label{derivintcondsuplib}\eeq
Examining (\ref{tangentvectorfield}) for $\zeta\in\mathbb{R}$, for a vertical tangent in $\sigma_L$ to occur, we need the numerator of (\ref{derivintcondsuplib}) to be zero. In this case, there are always two real values of $\zeta$ for which vertical tangents in $\sigma_L$ occur: 
\beq \zeta_t =\pm \frac{1}{2} \sqrt{\frac{1}{(1+c)\mathcal{K}(k)} \left(-2\mathcal{E}(k)-\mathcal{K}(k)+\sqrt{4\mathcal{E}^2(k)-4\mathcal{E}(k)\mathcal{K}(k)+c^2\mathcal{K}^2(k)}\right)}, \eeq
shown as blue crosses in Figure~\ref{laxspectrumcases}(a,b). 
Mapping these points back to $\sigma_{\mathcal{L}}$ these points correspond to the top (or bottom) of the figure 8 in Figure~\ref{spectrumcases}(a) or the top (or bottom) of the ellipse-like curve in Figure~\ref{spectrumcases}(b):
\beq \lambda_t = \pm \sqrt{\frac{E-2-c^2(e-1)}{2}+\frac{2\mathcal{E}(k)-c \sqrt{4\mathcal{E}^2(k)-4\mathcal{E}(k)\mathcal{K}(k)+c^2\mathcal{K}^2(k)}}{2 \mathcal{K}(k)}}. \label{lambdasuplib}\eeq
In the subluminal librational case in Section~\ref{sublibsubsection}, the qualitative change in the spectrum occurred when there was a bifurcation in the real values of $\zeta$ with vertical tangents. In this case, there is no such bifurcation. The qualitative change in the spectrum occurs when there is a bifurcation in imaginary values of $\zeta$. The imaginary roots of the numerator of (\ref{derivintcondsuplib}) are 
\beq \zeta_p =\pm \frac{1}{2} i \sqrt{\frac{1}{(1+c)\mathcal{K}(k)} \left(2\mathcal{E}(k)+\mathcal{K}(k)+\sqrt{4\mathcal{E}^2(k)-4\mathcal{E}(k)\mathcal{K}(k)+c^2\mathcal{K}^2(k)}\right)}. \label{imagrootsofderivintsuplib} \eeq
The qualitative change occurs for $E$ and $c$ such that $\zeta_p$ satisfies (\ref{intcond7}). This defines the curve seen in the superluminal librational region of Figure~\ref{SGregion-split}. Representative spectral plots for $E$ and $c$ on this boundary are seen in Figure~\ref{spectrumspecialcases}(1).
The slopes of $\sigma_{\mathcal{L}}$ at the origin are computed using the method described in Section~\ref{sublibsubsection}. They are
\beq \frac{ \textrm{d}\Omega_i}{\textrm{d} \Omega_r} = \pm \frac{c \sqrt{E(2-E)} \mathcal{K}(k)}{2\mathcal{E}(k)+(E-2) \mathcal{K}(k)}. \eeq
As with the subluminal librational solutions, expressions for derivatives of any order around the origin and around (\ref{lambdasuplib}) can be computed. 

\subsection{Superluminal rotational solutions}
The roots of $\Omega^2(\zeta)=0$ are given in (\ref{suprotrootsOmega}), seen as red crosses in Figure~\ref{laxspectrumcases}(c,d). In this case, (\ref{derivintcond}) is
\beq \frac{\text{d}I(\zeta)}{\text{d}\zeta} = \sqrt{c^2-1} \frac{-8 E \zeta^2 \mathcal{E}(k)+\left(c-1+8(E-1) \zeta^2-16(c+1)\zeta^4\right)\mathcal{K}(k)}{16\sqrt{2} \zeta^3 \Omega(z)}. \label{derivintcondsuprot}\eeq
Examining (\ref{tangentvectorfield}) for $\zeta\in \mathbb{R}$, for a vertical tangent in $\sigma_L$ to occur, we need the numerator of (\ref{derivintcondsuprot}) to be zero. In this case, again, there are always two real values of $\zeta$ for which vertical tangents in $\sigma_L$ occur: 
\beq \footnotesize \zeta_t =\pm \frac{1}{2} \sqrt{\frac{1}{(1+c)\mathcal{K}(k)} \left(-E \mathcal{E}(k)+(E-1)\mathcal{K}(k)+\sqrt{E^2 \mathcal{E}^2(k)-2(E-1)E \mathcal{E}(k)\mathcal{K}(k)+(c^2+(E-2)E)\mathcal{K}^2(k)}\right)}, \eeq
shown as blue crosses in Figure~\ref{laxspectrumcases}(a,b). 
Mapping these points back to $\sigma_{\mathcal{L}}$ these points correspond to the top (or bottom) of the ellipse-like curve in Figure~\ref{spectrumcases}(c) and the top of the ellipse-like curve in the upper-half plane and the bottom of the ellipse-like curve in the lower-half plane in Figure~\ref{spectrumcases}(d):
\beq \lambda_t = \pm \sqrt{\frac{E \mathcal{E}(k)-c^2(E-1)\mathcal{K}(k)-c \sqrt{E^2 \mathcal{E}^2(k)-2(E-1)E \mathcal{E}(k)\mathcal{K}(k)+(c^2+(E-2)E)\mathcal{K}^2(k)}}{2 \mathcal{K}(k)}}. \label{lambdasuplib}\eeq
As in the superluminal librational case above, we do not have a bifurcation in the real values of $\zeta$ with vertical tangents. The qualitative change in the spectrum occurs when there is a bifurcation in imaginary values of $\zeta$. The imaginary roots of the numerator of (\ref{derivintcondsuprot}) are
\beq \footnotesize \zeta_p =\pm \frac{1}{2} i \sqrt{\frac{1}{(1+c)\mathcal{K}(k)} \left(E \mathcal{E}(k)+(1-E)\mathcal{K}(k)+\sqrt{E^2 \mathcal{E}^2(k)-2(E-1)E \mathcal{E}(k)\mathcal{K}(k)+(c^2+(E-2)E)\mathcal{K}^2(k)}\right)}. \label{imagrootsofderivintsuprot} \eeq
The qualitative change occurs for $E$ and $c$ such that $|\zeta_p| = \zeta_4$ and $-|\zeta_p|=\zeta_1$ where $\zeta_1$ and $\zeta_4$ are the smallest and largest roots of $\Omega^2(\zeta)=0$ respectively. This condition is seen in Figure~\ref{spectrumspecialcases}(2b) and is
\beq c=\frac{\mathcal{E}(k)}{\mathcal{K}(k)}\sqrt{\frac{E}{E-2}}. \eeq
For $c>\sqrt{\frac{E}{E-2}}\frac{\mathcal{E}(k)}{\mathcal{K}(k)}$, we map $\zeta_p$ to $\sigma_{\mathcal{L}}$ and find these points corresponding to the bottom of the ellipse-like curve in the upper-half plane and the top of the ellipse-like curve in the lower-half plane in Figure~\ref{spectrumcases}(d):
\beq \lambda = \pm \sqrt{\frac{E \mathcal{E}(k)-c^2(E-1)\mathcal{K}(k)+c \sqrt{E^2 \mathcal{E}^2(k)-2(E-1)E \mathcal{E}(k)\mathcal{K}(k)+(c^2+(E-2)E)\mathcal{K}^2(k)}}{2 \mathcal{K}(k)}}. \label{lambdasuplib}\eeq

\section{Floquet theory and subharmonic perturbations}\label{subharmonic}

We examine $\sigma_{\mathcal{L}}$ using a Floquet parameter description. We use this to prove spectral {\em stability} results with respect to perturbations of an integer multiple of the fundamental period of the solution, {\em i.e.}, subharmonic perturbations.

We write the eigenfunctions from (\ref{spectralproblem}) using a Floquet-Bloch decomposition
\beq \left(\ba{c} W_1(z) \\ W_2(z) \ea \right) = e^{i \mu z} \left(\ba{c} \hat{W}_1(z) \\ \hat{W}_2(z) \ea \right),\;\; \hat{W}_1(z+T(k))=\hat{W}_1(z),\;\; \hat{W}_2(z+T(k)) = \hat{W}_2(z), \label{FBE} \eeq
with $\mu \in \left[-\pi/T(k),\pi/T(k)\right)$ \cite{DK,DS17}. Here $T(k) = 2 \mathcal{K}(k)$ for all solutions. From Floquet's Theorem \cite{DK}, all bounded solutions of (\ref{spectralproblem}) are of this form, and our analysis includes perturbations of an arbitrary period. Specifically, $\mu = 2m \pi /T(k)$ for $m\in \mathbb{Z}$ corresponds to perturbations of the same period $T(k)$ of the solutions, and in general
\beq \mu = \frac{2m \pi}{P T(k)},\;\; m,P \in \mathbb{Z}, \label{muequals} \eeq
corresponds to perturbations of period $P T(K)$. The choice of the specific range of $\mu$ is arbitrary as long as it is of length $2\pi / T(k)$. For added clarity in this section, we plot figures using the larger ranges $\left[-2\pi / T(k),2\pi/T(k)\right)$, periodically extending $\mu$ beyond the basic region.

In the previous sections $\sigma_{\mathcal{L}}$ is parameterized in terms of $\zeta$. We wish to re-parameterize $\sigma_{\mathcal{L}}$ in terms of $\mu$. We examine the eigenfunction $W_1$ from (\ref{FBE}). From the periodicity of $\hat{W}_1$ we have
\beq e^{i \mu T(k)} = \frac{W_1(z+T(k))}{W_1(z)}. \eeq
Using (\ref{W12eqns}), (\ref{varphieqn}), and (\ref{gammaeqn}), we find 
\beq e^{i\mu T(k)} = \exp\left(-2\int_0^{T(k)} \frac{-BC+D(A-\Omega)+B_z}{B}\text{d}z\right), \eeq
where we have used the periodicity properties
\beq A\left(z+T(k)\right) = A(z),\;\; B\left(z+T(k)\right) = B(z). \eeq
Using (\ref{intcond7}),
\beq \mu(\zeta) =-\frac{2 i I(\zeta)}{T(k)}+\frac{2 \pi n}{T(k)}, \label{muzetaeqn} \eeq
where $I(\zeta)$ is given in (\ref{Ieqn}) and $n\in \mathbb{Z}$. 

In what follows we discuss the stability of solutions with respect to perturbations of integer multiples of their fundamental periods, so-called subharmonic perturbations \cite{GCT}. The expression (\ref{muzetaeqn}) gives an easy way to do this. Specifically, from (\ref{muequals}) we know which values of $\mu$ correspond to perturbations of what type. 
For stability with respect to perturbations of period $2\pi m/\mu =P T(k)$, we need all spectral elements associated with a given $\mu$ value to have zero real part. In Figure~\ref{mufigures} we plot the real part of $\sigma_{\mathcal{L}}$ as a function of $\mu T(k)$ using (\ref{Omegaeqn}), (\ref{lambdaconnection}), and (\ref{muzetaeqn}). We rescale $\mu$ by $T(k)$ for consistency in our figures. Here
\beq \mu T(k) = \frac{2\pi m}{P}, \eeq
corresponds to perturbations of $P T(k)$ for any integer $m$. 

The following results are obtained in each region of parameter space:
\begin{itemize}
\item For the subluminal rotational case, all solutions are spectrally stable \cite{JMMP,JMMP2}.
\item For the subluminal librational case, all solutions are spectrally unstable with respect to all subharmonic perturbations. This is shown in Section~\ref{sublibsubsectionmu}.
\item For the superluminal librational case, all solutions are spectrally unstable, but all solutions left of curve 2 in Figure~\ref{SGregionSupLibNums} are stable with respect to perturbations of twice the period and the same period, all solutions left of curve 4 are stable with respect to perturbations of four times the period, all solutions left of curve 6 are stable with respect to perturbations of six times the period, as well as three times the period, etc. This is shown in Section~\ref{suplibsubsectionmu}.
\item For the superluminal rotational case, all solutions are spectrally unstable, but there are regions of stability with respect to subharmonic perturbations, see Figure~\ref{muregionsuprot} and Section~\ref{suprotsubsectionmu} for details.
\end{itemize}

\begin{figure}
\centering
\begin{tabular}{cccc}
  \includegraphics[width=36mm]{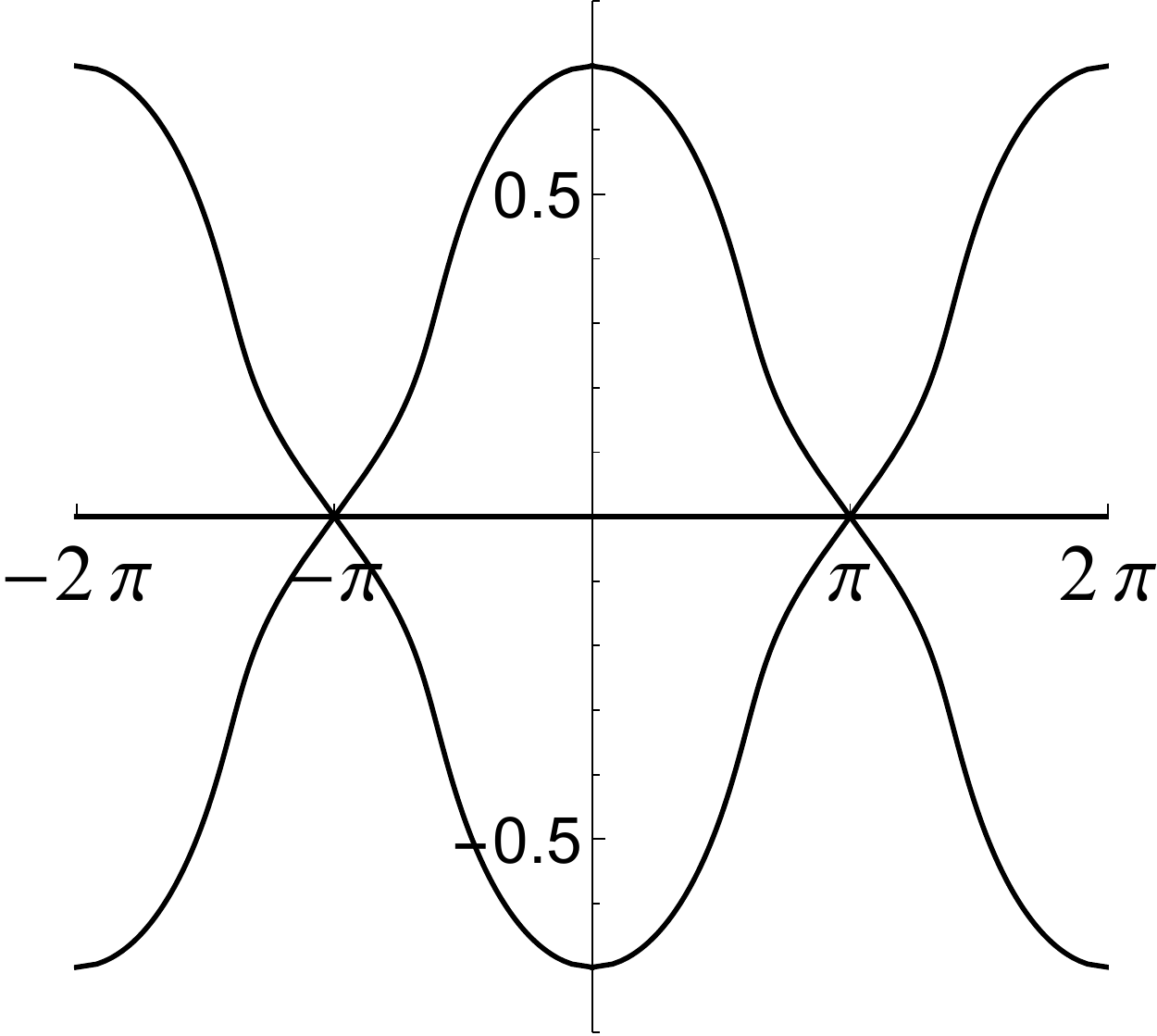} & \includegraphics[width=36mm]{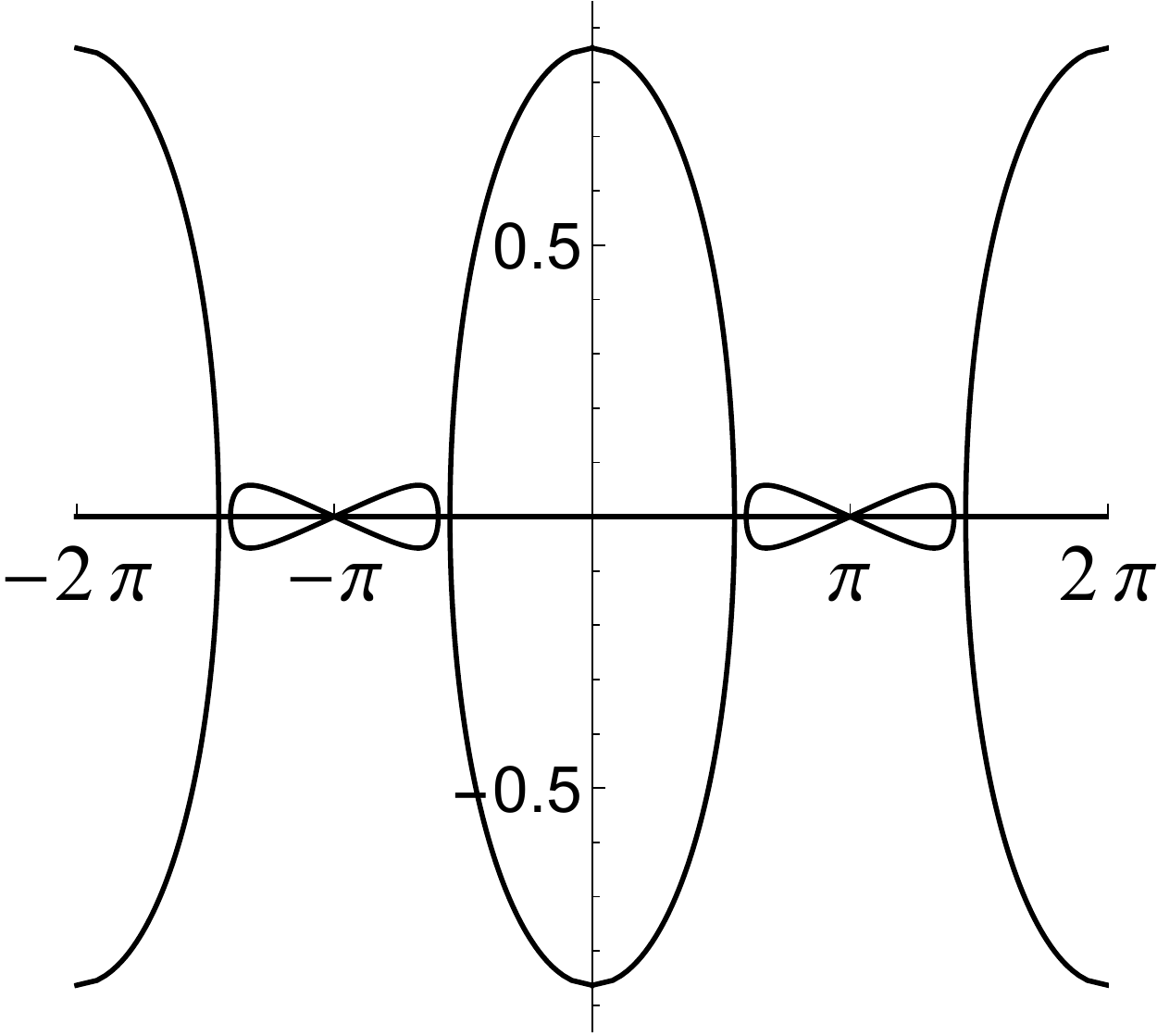} & \includegraphics[width=36mm]{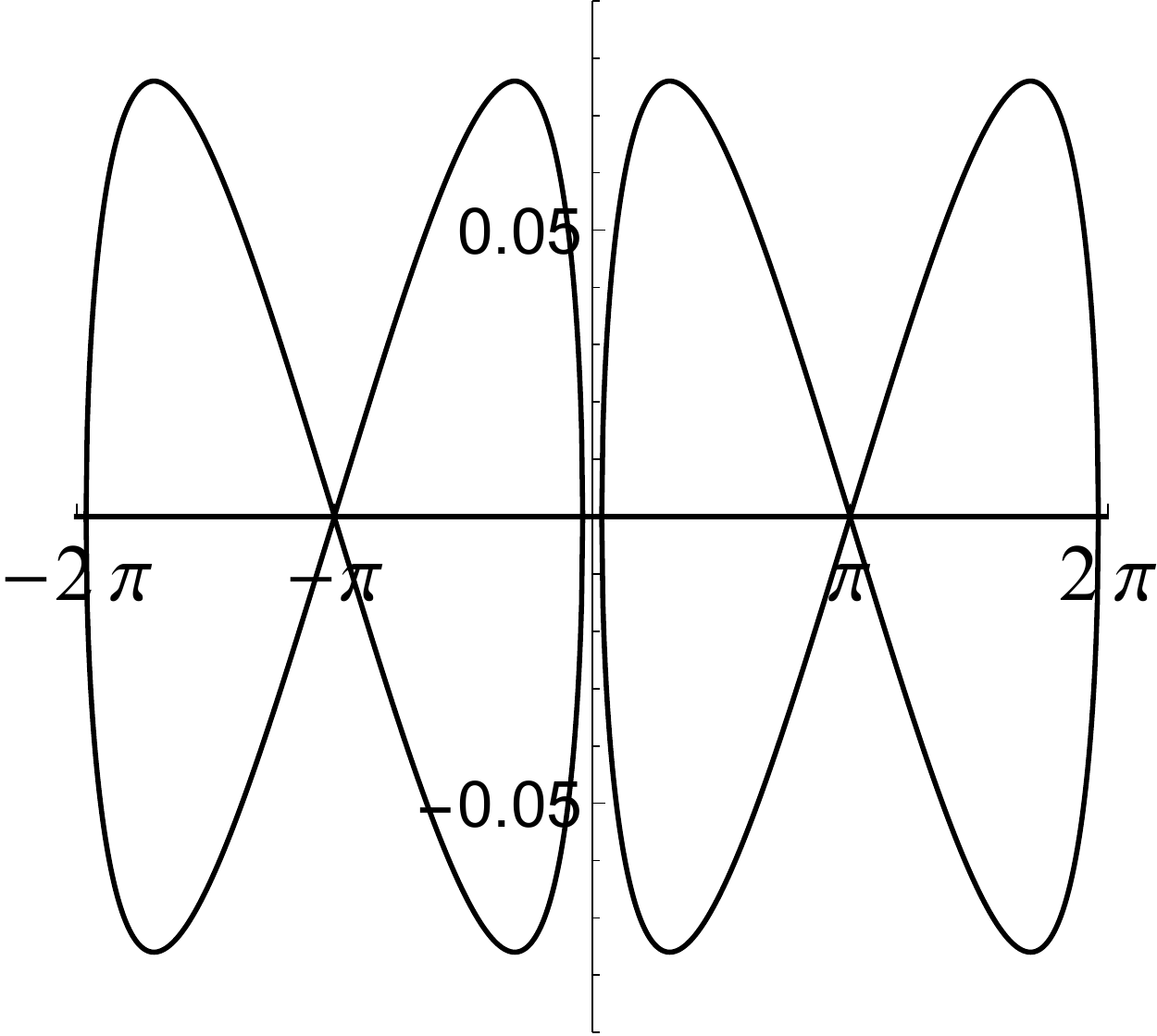} & \includegraphics[width=36mm]{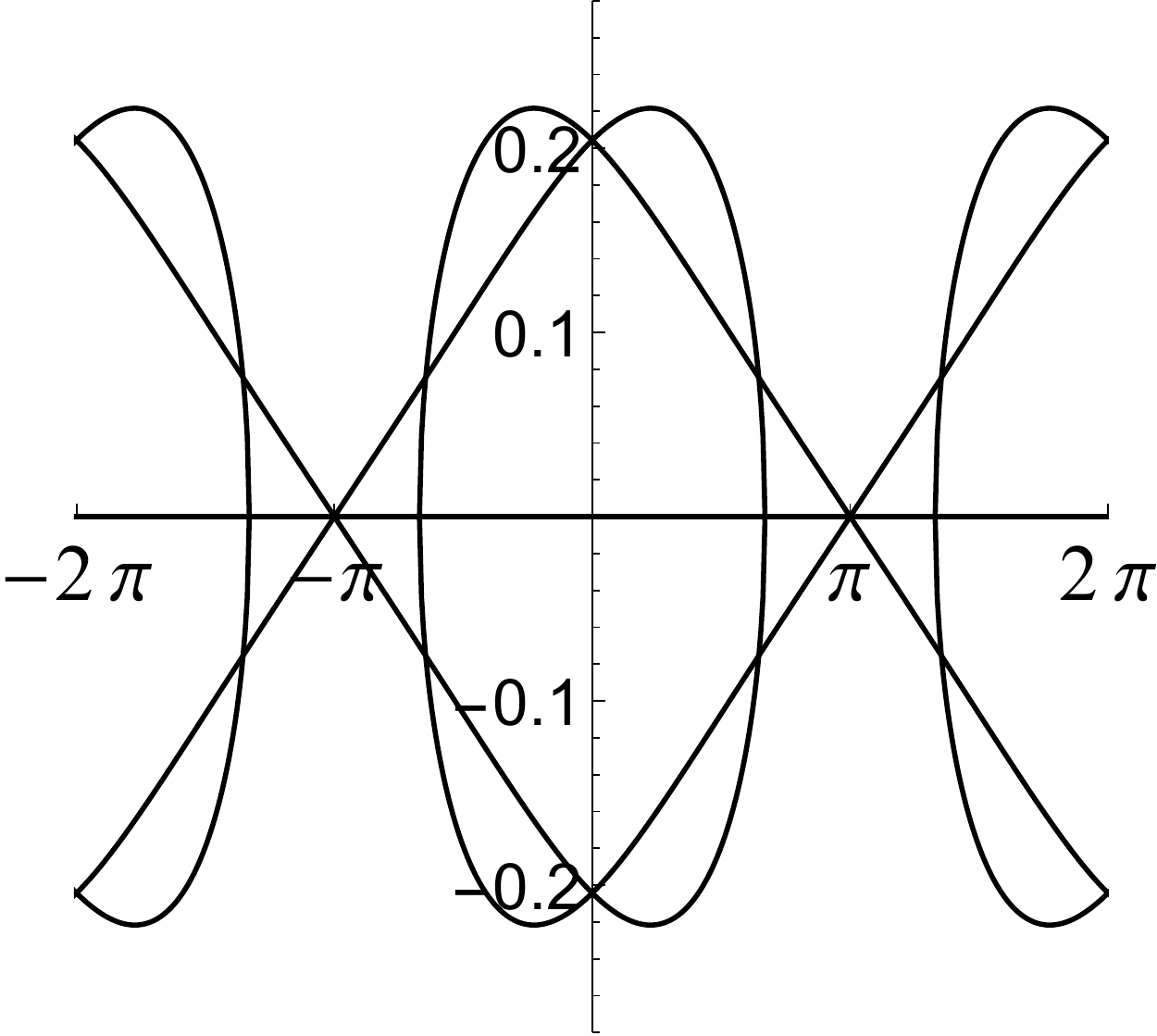} \\
(a) & (b) & (c) & (d) 
\end{tabular}
\begin{tabular}{cccc}
  \includegraphics[width=36mm]{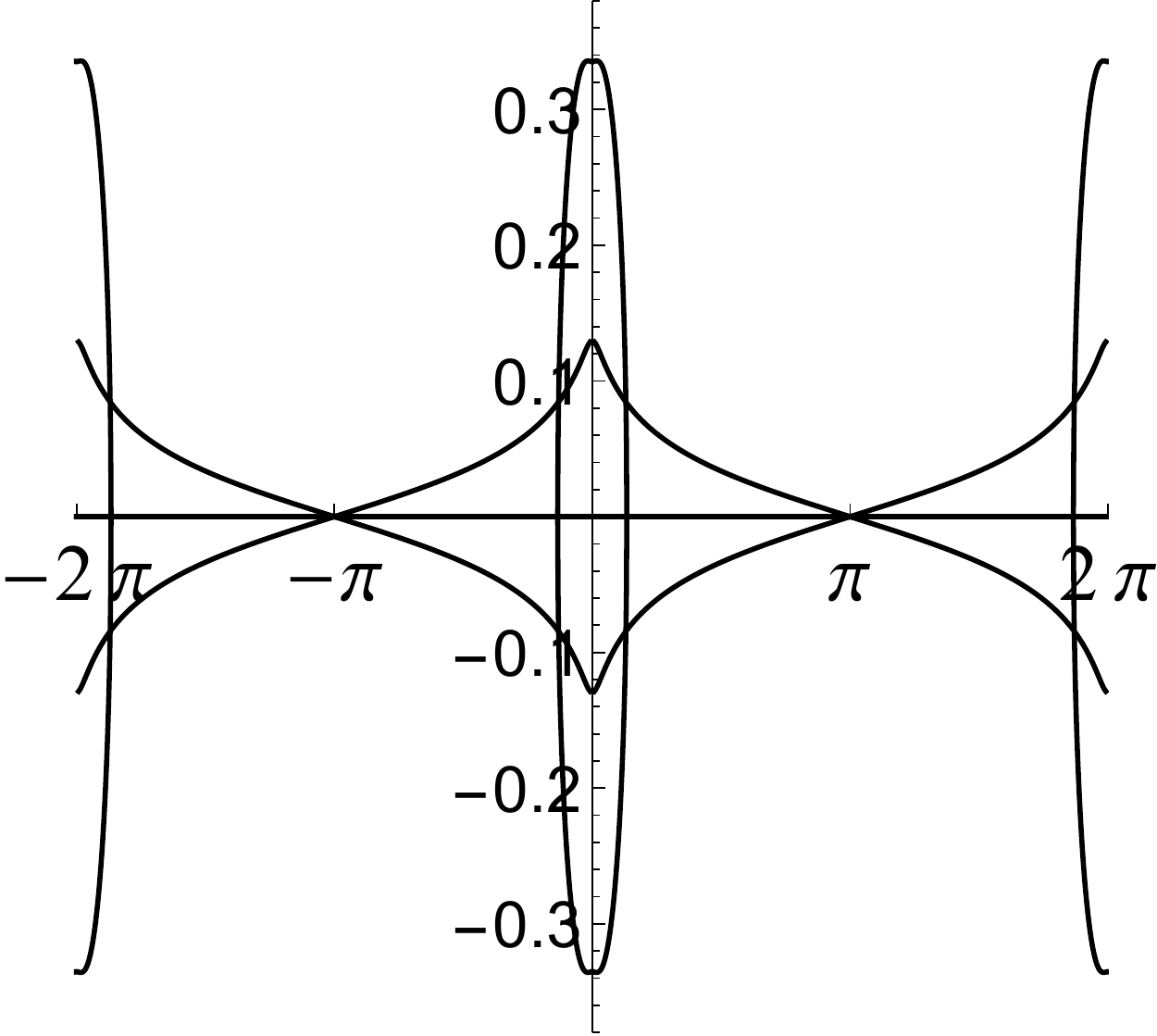} & \includegraphics[width=36mm]{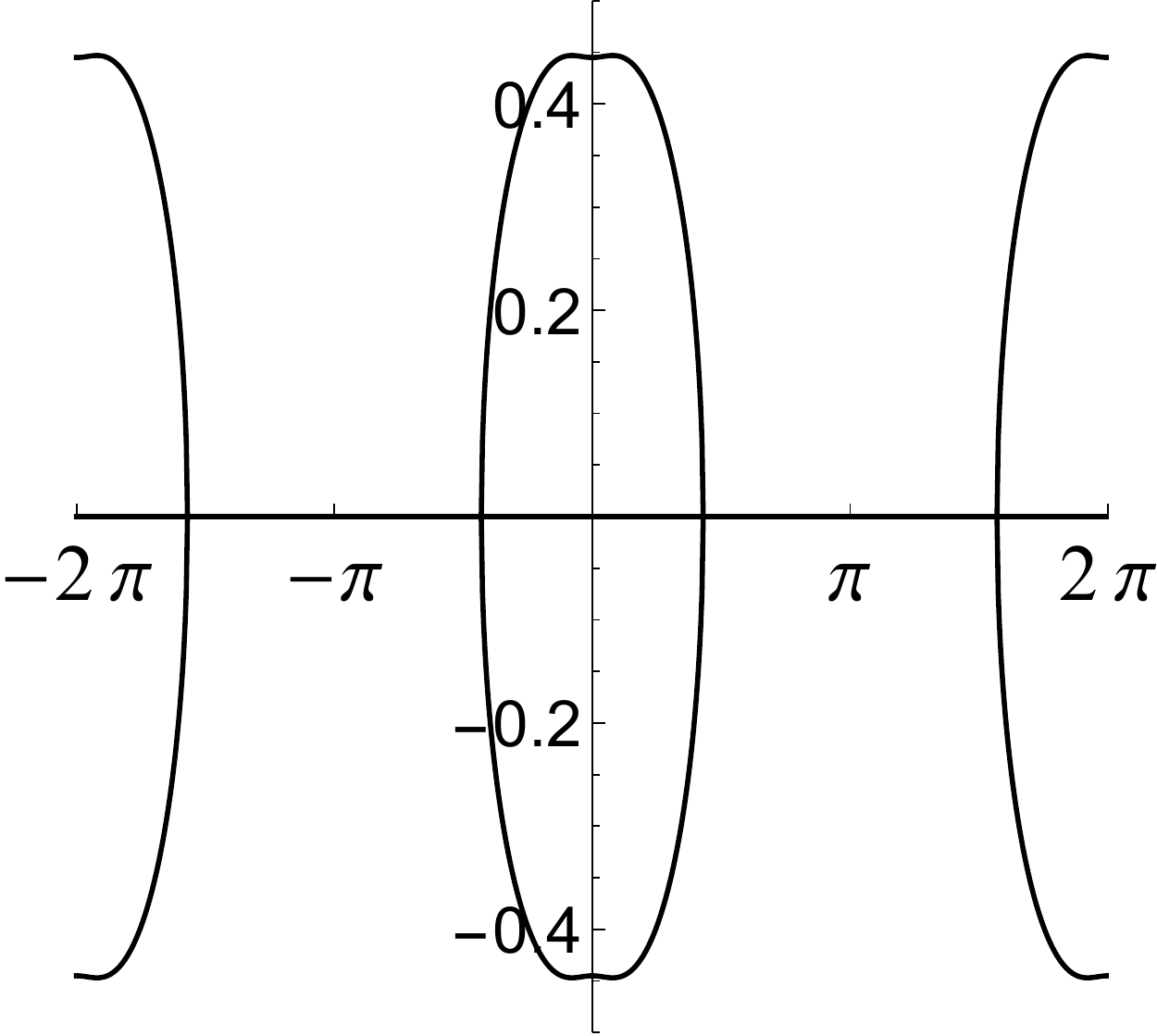} & \includegraphics[width=36mm]{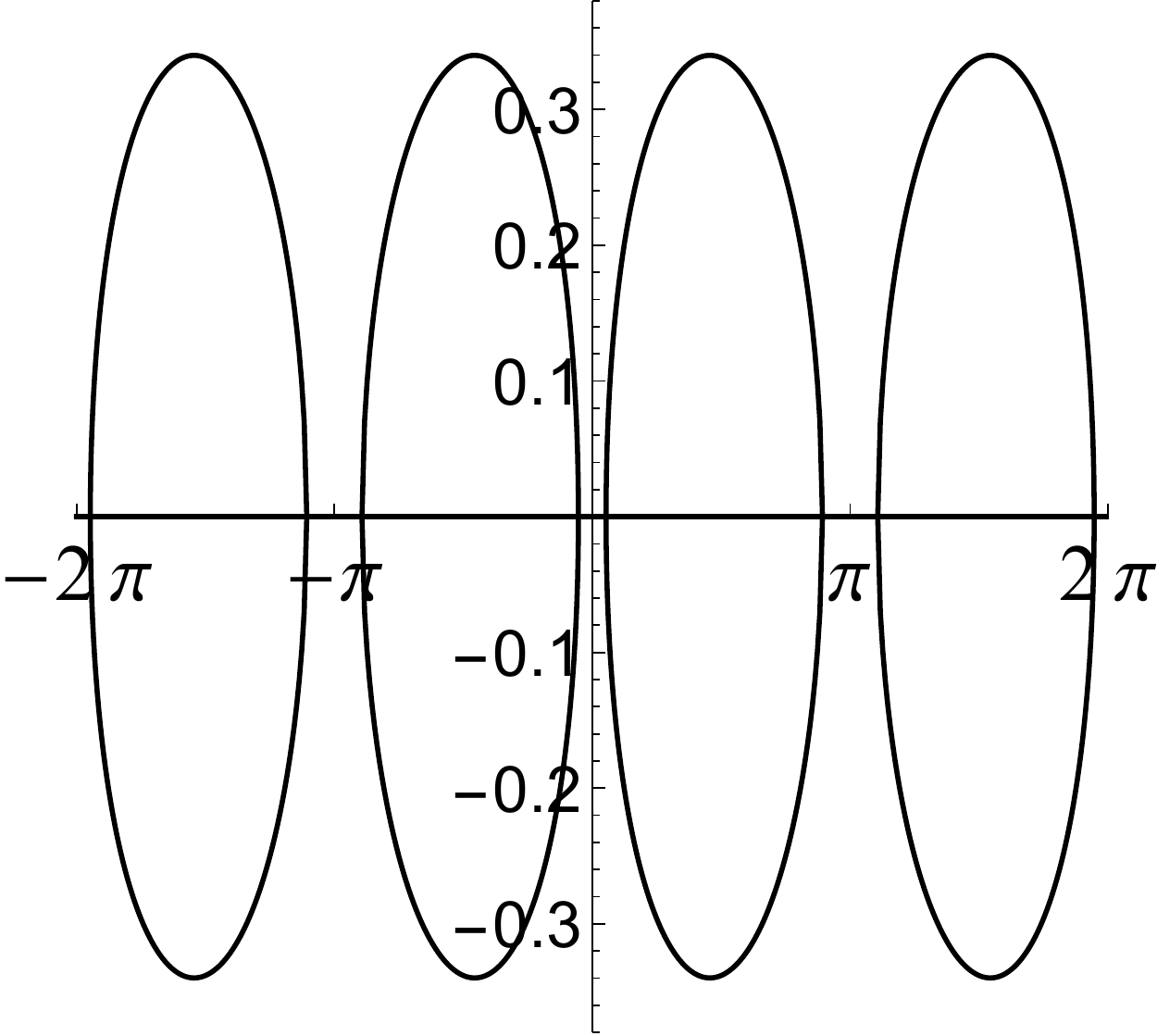} & \includegraphics[width=36mm]{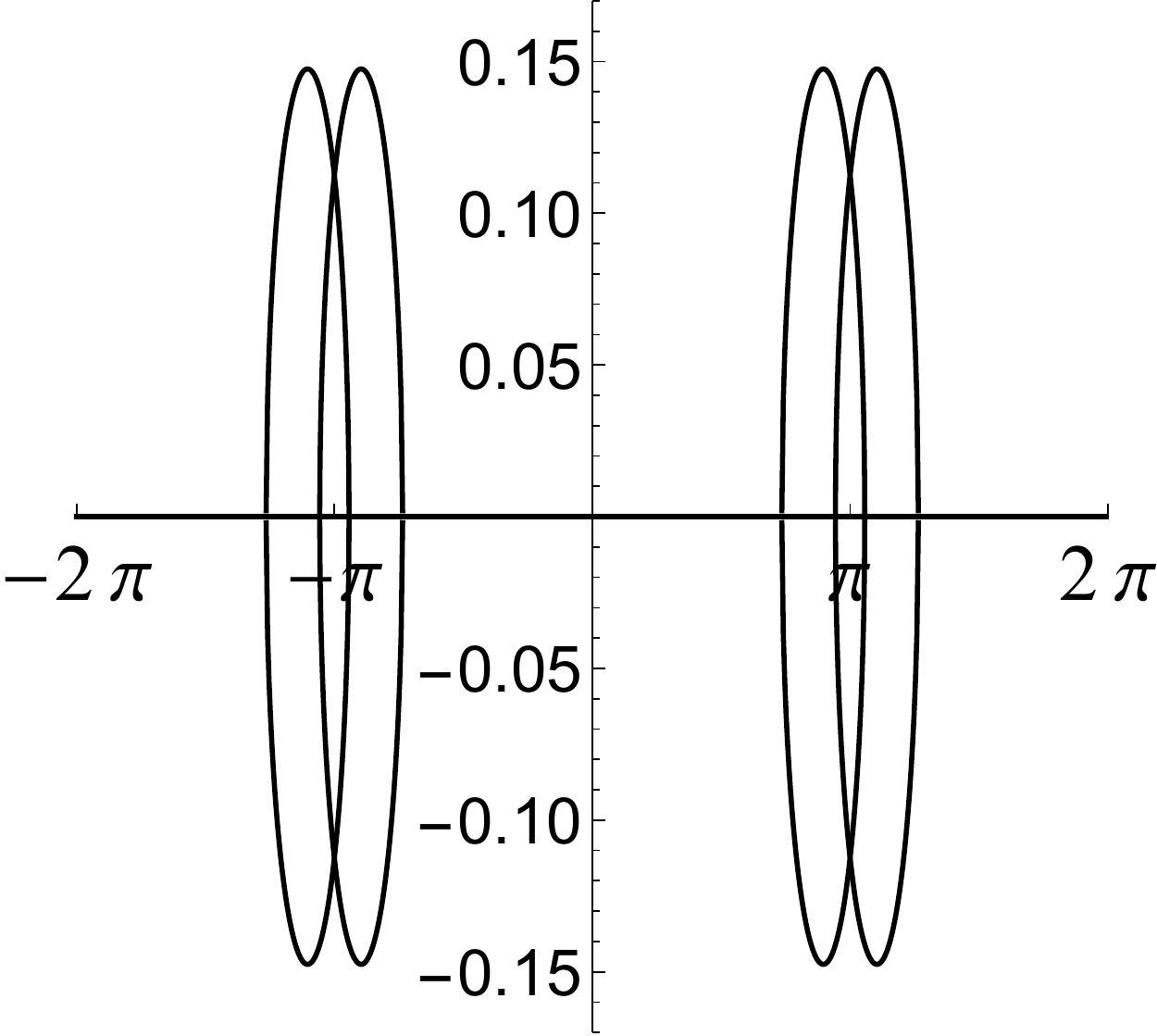}  \\
(e) & (f) & (g) & (h) 
\end{tabular}
\vspace{-2mm}
\caption{The real part of the spectrum $\text{Re}(\lambda)$ (vertical axis) as a function of $\mu T(k)$ (horizontal axis): for subluminal librational (a-b), superluminal librational (c-e), and superluminal rotational (f-h) solutions. (a) $c=0.6,\,E=1.0;$ (b) $c=0.8,\,E=1.5;$ (c) $c=1.5,\,E=0.7;$ (d) $c=1.5,\,E=1.5;$ (e) $c=1.02,\,E=1.8;$ (f) $c=1.3,\,E=2.9;$ (g) $c=1.4,\,E=2.4;$ (h) $c=2.1,\,E=6.8.$ }
\label{mufigures}
\end{figure}

We provide the following useful lemma:
\vspace{2mm}
\begin{lemma}\label{monotonelemma}For any analytic function $f(z) = u(x,y) +i v(x,y),$ on a contour where $u(x,y)=$ constant, $v(x,y)$ is strictly monotone, provided the contour does not traverse a saddle point. Similarly, on a contour where $v(x,y)=$ constant, $u(x,y)$ is strictly monotone, provided the contour does not traverse a saddle point.
\end{lemma}
\begin{proof} This is an immediate consequence of the Cauchy-Riemann relations \cite{BORN}.
\end{proof}
Thus along contours where $\textrm{Re}(I(\zeta))=0,$ if there are no saddle points, then $\textrm{Im}(I(\zeta))$ is monotone. If we fix $c$ and $E$, using (\ref{muzetaeqn}) we see that $\mu(\zeta)T(k)=2\pi n+2 \text{Im}(I(\zeta))-2 i \text{Re}(I(\zeta))$ is also monotone along curves with $\textrm{Re}\left(I(\zeta)\right) = 0$.
In what follows, we omit $\sigma_{\mathcal{L}}\cap i \mathbb{R}$. 

\subsection{Subluminal librational solutions}\label{sublibsubsectionmu}
There are two cases to consider for subluminal librational solutions, corresponding to the two qualitatively different stability spectra seen in Figure~\ref{spectrumcases}(f,g), and their corresponding Lax spectra in Figure~\ref{laxspectrumcases}(f,g). Representative plots of $\mu T(k)$ {\em vs}. $\text{Re}(\lambda)$ for these cases are shown in Figure~\ref{mufigures}(a,b). We prove the following theorem:
\vspace{2mm}
\begin{thm}
The subluminal librational solutions to (\ref{SGlab}) are unstable with respect to all subharmonic perturbations.
\end{thm} 
\begin{proof}
It suffices to show that for some $\zeta\in\sigma_L$, $\mu = 0$ and $\text{Re}(\lambda)>0$. We split into cases with qualitatively different spectra:
\begin{enumerate}
\item In the case where the stability spectrum looks qualitatively like an infinity symbol, we examine $\zeta\in\sigma_L$, see Figure~\ref{laxspectrumcases}(f). The infinity symbol spectrum is double covered, so without loss of generality, we consider only values of $\zeta$ in the upper-half plane. Specifically, we consider values of $\zeta$ ranging from $\zeta_2$ to $\zeta_1$, moving from the red cross in the second quadrant to the red cross in the first quadrant of Figure~\ref{laxspectrumcases}(f). At $\zeta_2$, $\mu T(k) = -\pi$ and $\text{Re}(\lambda) = 0$. As $\zeta$ moves from $\zeta_2$ to $\zeta_1$, $\mu T(k)$ monotonically increases (Lemma~\ref{monotonelemma}) until it reaches $\mu T(k) = \pi$ at $\zeta=\zeta_1$, where $\text{Re}(\lambda) = 0$, see Figure~\ref{mufigures}(a). Along this curve $\text{Re}(\lambda) \ne 0$ so by the intermediate value theorem at some point between $\zeta_2$ and $\zeta_2$, $\mu T(k) =0$ with $\text{Re}(\lambda) >0$.

\item In the case where the stability spectrum looks qualitatively like a figure 8 inset in an ellipse-like curve, examine $\zeta\in\sigma_L$, see Figure~\ref{laxspectrumcases}(g). The $\zeta$ spectrum has two components, $\zeta$ corresponding to the figure 8, and $\zeta$ corresponding to the ellipse-like curve. 
For instability, we only need to examine $\zeta$ corresponding to the ellipse-like curve. 
Again, we consider only values of $\zeta$ in the upper-half plane. Specifically, we consider values of $\zeta$ ranging from $-|\zeta_{t2}|$ to $|\zeta_{t2}|$, moving from the blue cross in the second quadrant to the blue cross in the first quadrant of Figure~\ref{laxspectrumcases}(g). At $-|\zeta_{t2}|$, $\mu T(k) = -2 i I(-|\zeta_{t2}|)$, and $\text{Re}(\lambda) = 0$. 
As $\zeta$ moves from $-|\zeta_{t2}|$ to $|\zeta_{t2}|$, $\mu T(k)$ monotonically increases (Lemma~\ref{monotonelemma}) until it reaches $\mu T(k) = -2 i I(|\zeta_{t2}|)$ at $\zeta = |\zeta_{t2}|$, with $\text{Re}(\lambda) = 0$, see the ellipse-like curve in Figure~\ref{mufigures}(b). 
Because of the symmetries of $I(\zeta)$ for $\zeta\in \mathbb{R}$ we have that $\mu T(k) = -2 i I(|\zeta_{t2}|)=2 i I(-|\zeta_{t2}|)$. Along this curve $\text{Re}(\lambda) \ne 0$ so again by the intermediate value theorem at some point between $-|\zeta_{t2}|$ and $|\zeta_{t2}|$, $\mu T(k) =0$ with $\text{Re}(\lambda) >0$.
\end{enumerate}
\vspace{-2mm}
\end{proof}

\subsection{Superluminal librational solutions}\label{suplibsubsectionmu}
\vspace{2mm}
\begin{thm}\label{suplibthm}
The superluminal librational solutions to (\ref{SGlab}) are stable with respect to subharmonic perturbations of period $PT(k)$ if they satisfy the condition
\beq -2i I(-|\zeta_t|) \ge \frac{(P-1)\pi}{P}, \label{suplibmucond1} \eeq
for $P$ odd, and 
\beq -2i I(-|\zeta_t|) \ge \frac{(P-2)\pi}{P}, \label{suplibmucond2} \eeq
for $P$ even.
\end{thm} 
\begin{proof}
For stability with respect to perturbations of period $PT(k)$ we need that for $\mu T(k) = 2\pi m/P$, the spectral elements $\lambda \in \sigma_{\mathcal{L}}$ have zero real part, {\em i.e.}, for $\mu T(k) = 0,\frac{2\pi}{P},\ldots,\frac{2\pi (P-1)}{P},$ $\text{Re}(\lambda)=0$.

\begin{figure}
\centering
\includegraphics[width=10cm]{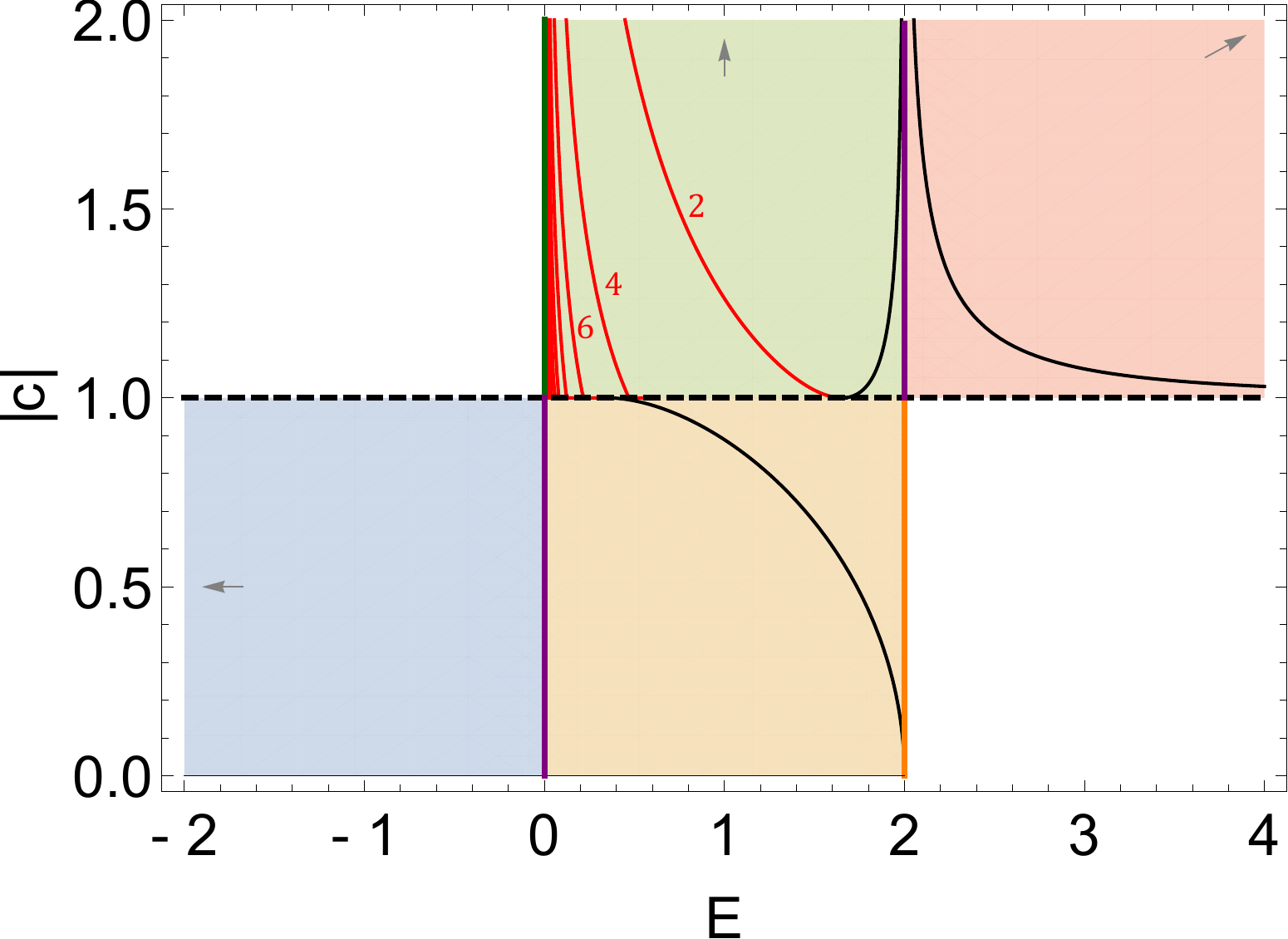} 
\caption{A plot of parameter space showing the spectral stability of superluminal librational solutions with respect to various subharmonic perturbations. Within the superluminal librational region, all solutions left of curve 2 are stable with respect to perturbations of twice the period as well as perturbations of the same period, all solutions left of curve 4 are stable with respect to perturbations of four times the period, all solutions left of curve 6 are stable with respect to perturbations of six times the period as well as perturbations of three times the period, etc.}
\label{SGregionSupLibNums}
\end{figure}

We examine $\zeta\in \sigma_L$, in the figure 8 case, see Figure~\ref{laxspectrumcases}(a). The figure 8 spectrum is double covered, so, without loss of generality, we consider only values of $\zeta$ in the left-half plane. Specifically we consider values of $\zeta$ ranging from $\zeta_3$ to $\zeta_2$ passing along the level curve through $\zeta = -|\zeta_t|$. At $\zeta_3$, $\mu T(k)=\pi$ and $\text{Re}(\lambda) = 0$. As $\zeta$ moves from $\zeta_3$ to $-|\zeta_t|$, $\mu T(k)$ monotonically decreases (Lemma~\ref{monotonelemma}) until it reaches $\mu_t T(k) = -2i I(-|\zeta_t|)$ at $\zeta=-|\zeta_t|$. At $-|\zeta_t|$, $\text{Re}(\lambda)=0$. Note that we are only considering the lower-left quarter plane. The analysis for $\zeta$ ranging from $\zeta_1$ to $|\zeta_t|$ is symmetric in $\mu T(k)$.

Qualitatively, we have figure 8s centered at $\mu T(k)=\pi+ 2\pi n$ and extending over $[\mu_t T(k) +2\pi n,\pi +(\pi - \mu_t T(k)) +2\pi n]$, see Figure~\ref{mufigures}(c,d). Relevant to the interval $[0,2\pi)$ is the figure 8 centered at $\pi$. For stability, we need the left-most edge of the figure 8 to be to the right of $\frac{(P-1)\pi}{P}$ for $P$ odd and to the right of $\frac{(P-2)\pi}{P}$ for $P$ even. Similarly, we need the right-most edge of the figure 8 to be to the left of $\frac{(P+1)\pi}{P}$ for $P$ odd and to the left of $\frac{(P+2)\pi}{P}$ for $P$ even. These conditions are for $P$ odd:
\beq \mu_t T(k) \ge \frac{(P-1)\pi}{P} \text{ and } \pi +(\pi - \mu_t T(k))\le \frac{(P+1)\pi}{P}, \eeq
and for $P$ even:
\beq \mu_t T(k) \ge \frac{(P-2)\pi}{P} \text{ and } \pi +(\pi - \mu_t T(k))\le \frac{(P+2)\pi}{P}. \eeq
These conditions simplify to give (\ref{suplibmucond1}) and (\ref{suplibmucond2}) respectively.
\end{proof}

\begin{figure}
\begin{tabular}{cccc}
  \includegraphics[width=70mm]{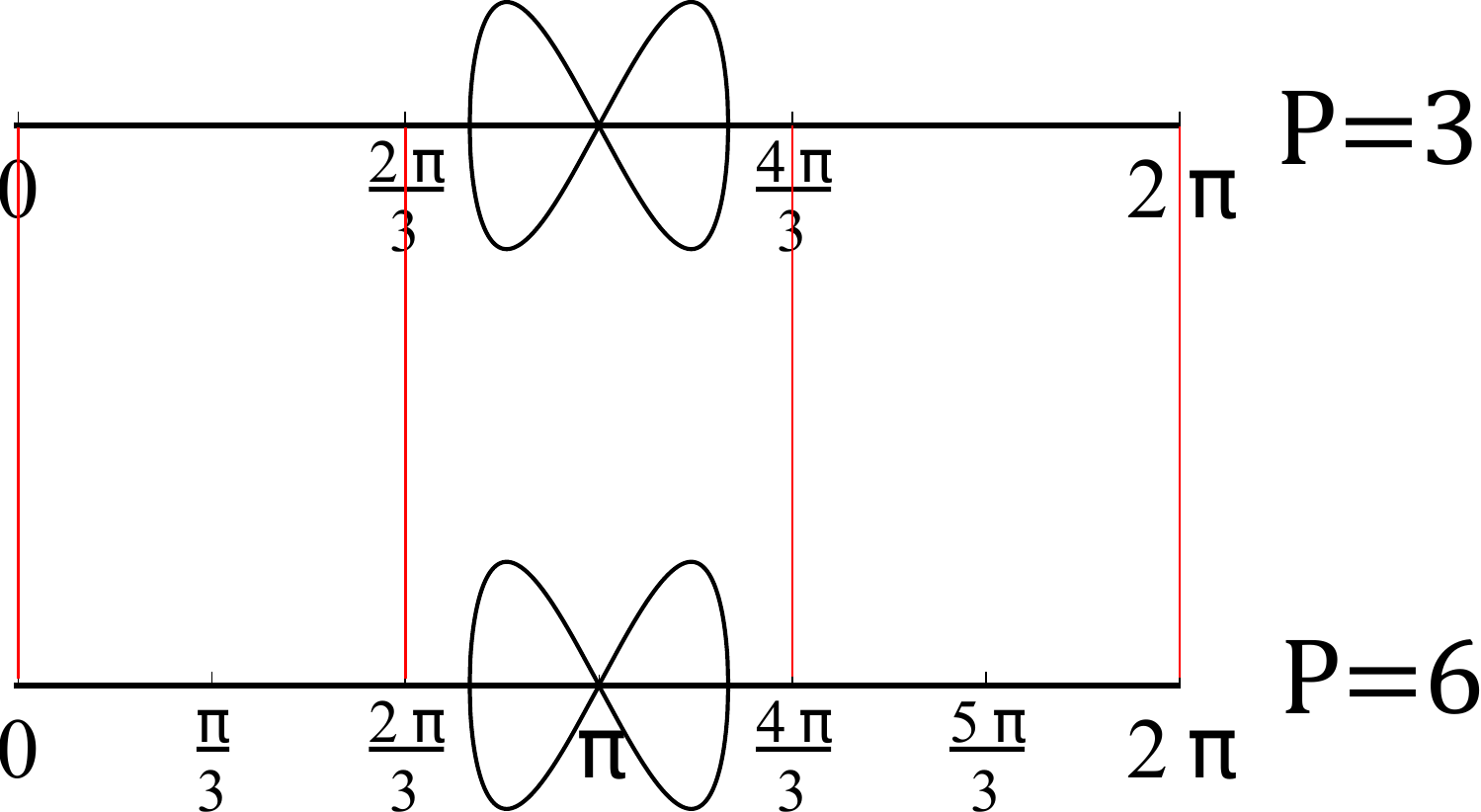} & & & \includegraphics[width=70mm]{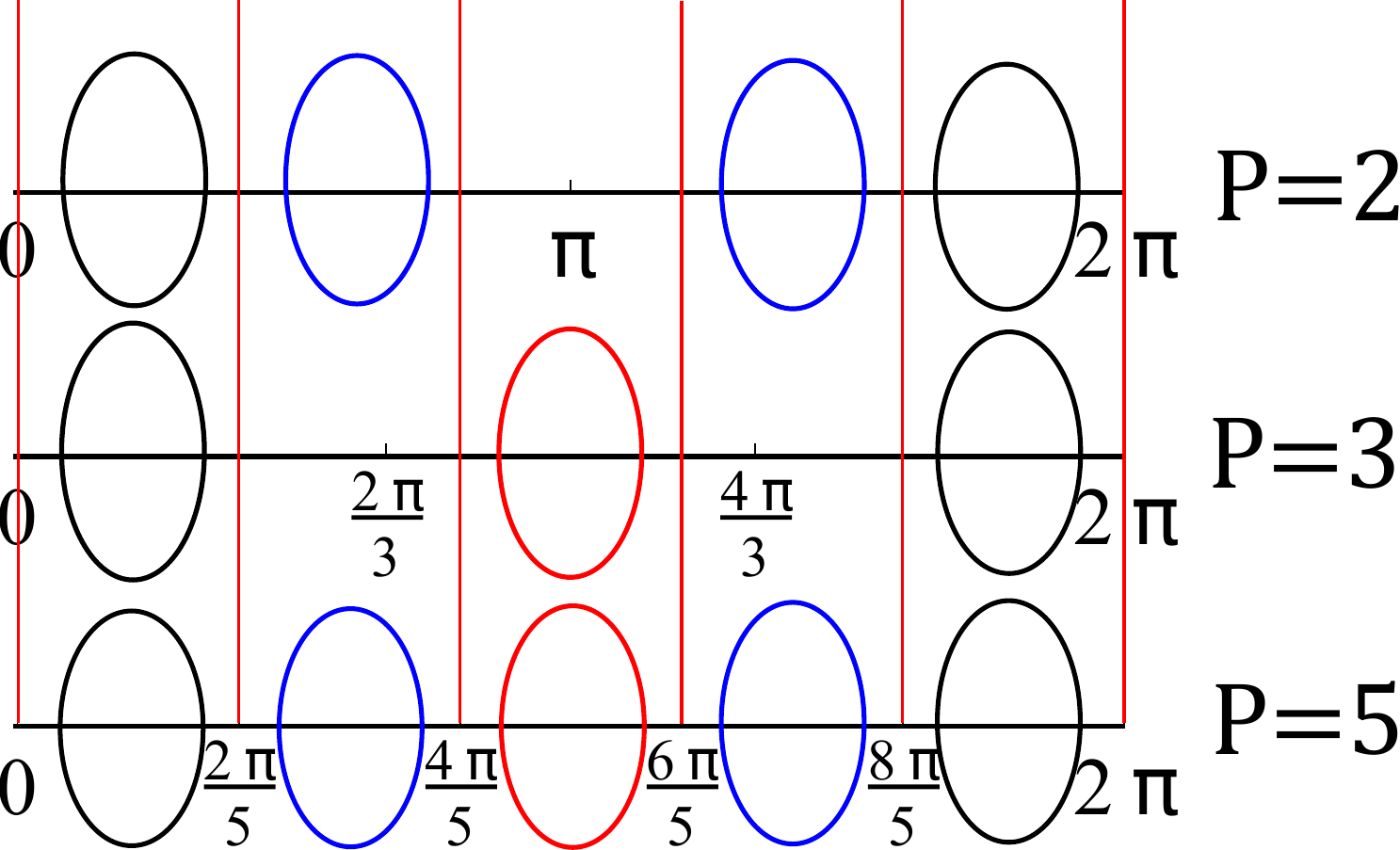} \\
(a) Superluminal librational  & & & (b) Superluminal rotational
\end{tabular}
\caption{The real part of the spectrum $\text{Re}(\lambda)$ (vertical axes) as a function of $\mu T(k)$ (horizontal axes): $\mu T(k) = 2m\pi /P$ for integers $m$ and $P$ corresponds to perturbations of period $P$ times the period of the underlying solution. (a) The superluminal solution is stable with respect to perturbations of three times its period is necessarily stable with respect to perturbations of six times its period. (b) If a superluminal rotational solution is stable with respect to perturbations of five times its period, it is stable with respect to perturbations of three times its period or perturbations of two times its period. 
(i) If the ellipse-like curves are in $(4\pi/5,6\pi/5)$ they are necessarily in $(2\pi/3,4\pi/3)$ (red),
(ii) if the ellipse-like curves are in $(2\pi/5,4\pi/5)$ and $(6\pi/5,8\pi/5)$ they are necessarily in $(0,\pi)$ and $(\pi,2\pi)$ respectively (blue), 
(iii) if the ellipse-like curves are in $(0,2\pi/5)$ and $(8\pi/5,2\pi)$ they are necessarily in both $(0,\pi)$ and $(\pi,2\pi)$ respectively as well as $(0,2\pi/3)$ and $(2\pi/3,4\pi/3)$ respectively (black).}
\label{proofpics}
\end{figure}

We remark that for a given odd $P$ the condition (\ref{suplibmucond1}) is the same as the condition (\ref{suplibmucond2}) for $2P$. Thus, for superluminal librational waves if we have stability with respect to perturbations of some odd multiple $P$ of the period $T(k)$ we also have stability with respect to perturbations of $2P T(k)$. This is shown in the case when $P=3$ in Figure~\ref{proofpics}(a). 
These results are summarized in Figure~\ref{SGregionSupLibNums} where we plot only the condition (\ref{suplibmucond2}).
We remark that it is possible for solutions to be stable with respect to perturbations of four times the period but not with respect to three times the period. Solutions of this type would lie to the left of curve 4 but to the right of curve 6 in Figure~\ref{SGregionSupLibNums}. More generally it is possible to have solutions which are stable with respect to $p$ times the period but not with respect to $q$ times the period where $p$ is even and less than $q<p<2q$.

\subsection{Superluminal rotational solutions}\label{suprotsubsectionmu}
\vspace{2mm}
\begin{thm}
The superluminal rotational solutions to (\ref{SGlab}) are stable with respect to subharmonic perturbations of period $PT(k)$ if they simultaneously satisfy the conditions
\beq 2\pi n -2i I(-|\zeta_t|) \le \frac{2\pi (m+1)}{P}, \label{suprotmucond1} \eeq
\beq 2\pi n -2i I(|\zeta_p| i ) \ge \frac{2\pi m}{P}, \label{suprotmucond2} \eeq
for some $n\in\mathbb{Z}$ and some $m\in \{0,1,\ldots, P-1\}$. Note that $\text{Re}\left(I(-|\zeta_t|)\right) = 0$ and $\text{Re}\left(I(|\zeta_p| i )\right)=0$.
\end{thm} 
\begin{proof}
For stability with respect to perturbations of period $PT(k)$ we need that for $\mu T(k) = 2\pi m/P$, the spectral elements $\lambda \in \sigma_{\mathcal{L}}$ have zero real part for all $m\in \{0,1,\ldots, P-1\}$.

We examine $\zeta\in \sigma_L$, in the case where we have ellipse-like curves in the upper- and lower-half planes, see Figure~\ref{laxspectrumcases}(d). As in Theorem~\ref{suplibthm}, using symmetries we restrict ourselves to $\zeta$ in the upper-left quarter plane. Specifically we consider values of $\zeta$ ranging from $-|\zeta_t|$ to $|\zeta_p| i$. At $-|\zeta_t|$, $\mu T(k)= -2i I(-|\zeta_t|)$ and $\text{Re}(\lambda) = 0$. As $\zeta$ moves from $-|\zeta_t|$ to $|\zeta_p| i$, $\mu T(k)$ monotonically decreases (Lemma~\ref{monotonelemma}) until it reaches $\mu_t T(k) = -2i I(|\zeta_p| i)$ at $\zeta=|\zeta_p|i$. At $|\zeta_p|i$, $\text{Re}(\lambda)=0$.

Qualitatively, we have an ellipse-like curve beginning at $ -2i I(-|\zeta_t|)+2\pi n$ and extending to $-2i I(|\zeta_p| i)+2\pi n$, see Figure~\ref{mufigures}(g,h). 
The only values of $\mu T(k)$ with $\text{Re}(\lambda)>0$ lie within the range $(2i I(|\zeta_p| i)+2\pi n,2i I(-|\zeta_t|)+2\pi n).$ So if $(2i I(|\zeta_p| i)+2\pi n,2i I(-|\zeta_t|)+2\pi n) \subset (\frac{2\pi m}{P},\frac{2\pi (m+1)}{P})$, for some $m\in \{0,1,\ldots, P-1\}$, then $\text{Re}(\lambda)=0$ for $\mu T(k) = 2\pi m/P$ for all $m\in \{0,1,\ldots, P-1\}$.

Thus for stability we need the right-most edge of each of these ellipse-like curves to be to the left of $2\pi (m+1)/P$, and the left-most edge of each of these ellipse-like curves to be to the right of $2\pi m/P$ for some $m\in \{0,1,\ldots, P-1\}$. 
This gives us conditions (\ref{suprotmucond1}) and (\ref{suprotmucond2}).
\end{proof}

These results are summarized in Figure~\ref{muregionsuprot}. We choose to rescale parameter space using the elliptic modulus $k = \sqrt{2/E}$, to show the extent of the subharmonic stability regions as $E\rightarrow\infty$. We only show regions for $P=1,2,3,4,5$ for the sake of clarity.

We see that there are many disjoint regions of subharmonic stability for each value of $P$ corresponding to the various choices for $m$. 
Within each disjoint region of stability for same period perturbations (blue) there are $P$ disjoint regions of stability with respect to perturbations of $P$ times the period. This follows directly from the conditions (\ref{suprotmucond1}) and (\ref{suprotmucond2}). 
\begin{figure}
\centering
\includegraphics[width=12cm]{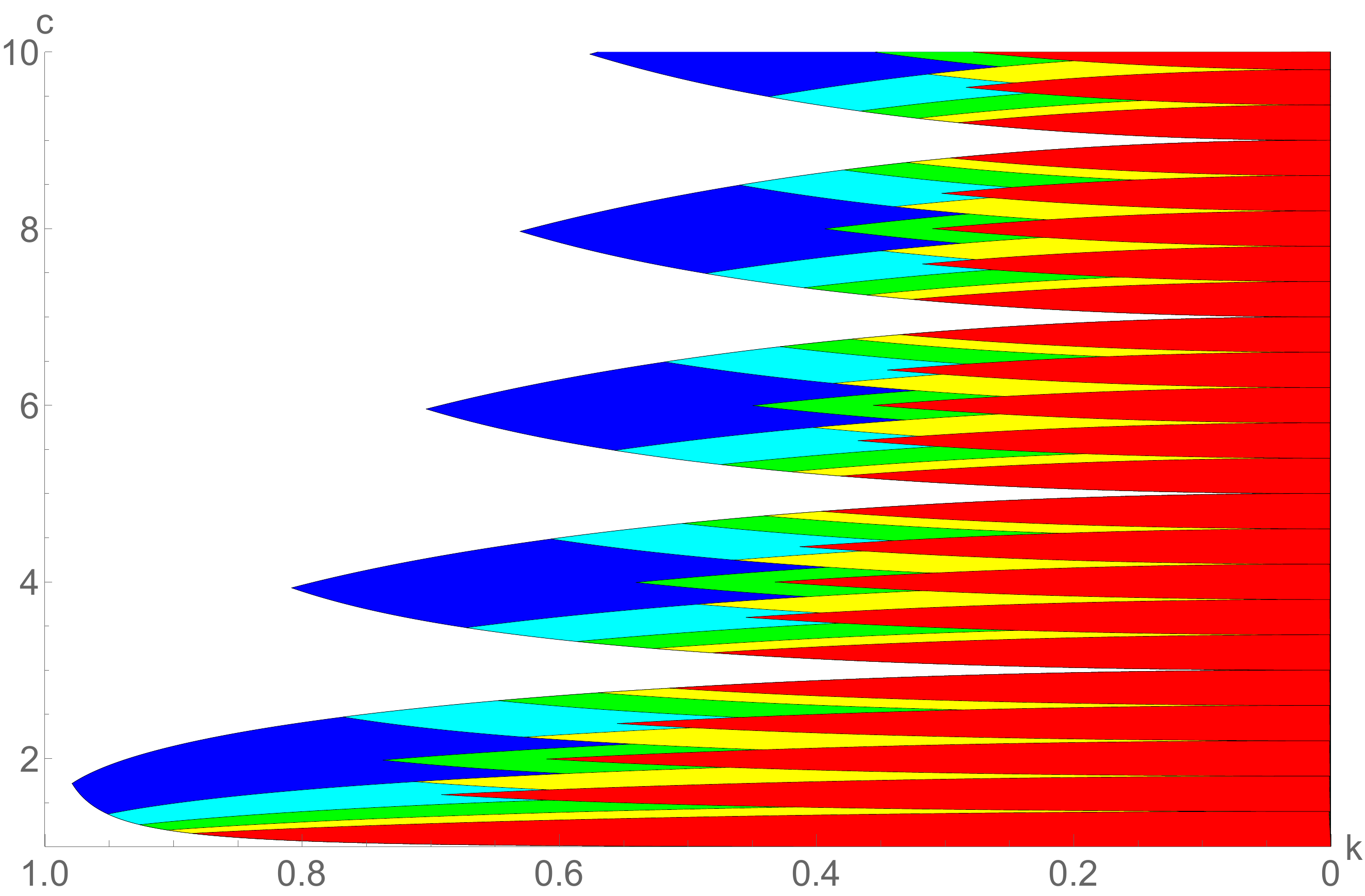} 
\caption{A plot of the superluminal rotational region of parameter space showing the spectral stability with respect to various subharmonic perturbations. Parameter space is rescaled using the elliptic modulus $k = \sqrt{2/E}$, to show the extent of the curves as $E\rightarrow\infty$. Solutions within the blue (light blue, green, yellow, red) region are stable with respect to perturbations of one 
(two, three, four, five) times the period respectively.}
\label{muregionsuprot}
\end{figure}
We note the possibility of solutions which are stable with respect to three times the period of the solution but not with respect to two times the period of the solution. An example of what $\mu T(k)$ looks like in this case is shown in Figure~\ref{mufigures}(h) with $c=2.1,\,E=6.8,\, k=0.542326.$ 
Indeed it is possible to have solutions which are stable with respect to $p$ times the period of the solution but not with respect to $q$ times the period of the solution for any $p>q$ where $q\nmid p$.
From Figure~\ref{muregionsuprot} we notice that if a solution is stable with respect to perturbations of five times the period (red) it is stable with respect to either perturbations of two times the period (light blue) or three times the period (green). This is proved by a simple topological argument shown in Figure~\ref{proofpics}(b) and explained in the caption. 

\section{Conclusion}

In this paper, the methods of
\cite{DS17} are used to examine and explicitly determine the stability spectrum of the
stationary solutions of the sine-Gordon equation. As in \cite{DS17}, we demonstrate that the parameter space for
the stationary solution separates in different regions where the
topology of the spectrum is different. An additional subdivision of
this parameter space is found for superluminal waves when considering the stability of the
solutions with respect to subharmonic perturbations of a specific
period. We find solutions which are stable with respect to perturbations of $p$ times the period but unstable with respect to $q$ times the period, where $p<q$.

\section{Acknowledgments}
This work was supported by the National Science Foundation through grant NSF-DMS-100801 (BD). Benjamin L. Segal acknowledges funding from a Department of Applied Mathematics Boeing fellowship and the Achievement Rewards for College Scientists (ARCS) fellowship. Any opinions, findings, and conclusions or recommendations expressed in this material are those of the authors and do not necessarily reflect the views of the funding sources.

\bibliography{refs}
\bibliographystyle{acm}

\end{document}